\documentclass[11pt]{article}

\usepackage{stmaryrd}
\SetSymbolFont{stmry}{bold}{U}{stmry}{m}{n}

\usepackage[normalem]{ulem}
\usepackage[inline]{enumitem}
\usepackage[dvipsnames]{xcolor}
\usepackage[scr]{rsfso}
\usepackage[Symbolsmallscale]{upgreek}
\usepackage{hyperref}
\usepackage{xcolor}
\hypersetup{
    colorlinks,
    linkcolor={red!50!black},
    citecolor={blue!50!black},
    urlcolor={blue!80!black}
}

\usepackage{wrapfig,tikz-cd,xspace,mathtools,stmaryrd,setspace}
\usepackage{amsmath,amsthm,amssymb,thmtools,thm-restate}
\usepackage{bm,graphicx,trimclip,soul,import,float,multicol,mathtools}
\usepackage{setspace,xparse,accents,algorithm2e,hyperref}
\usepackage{centernot,scalerel,xfrac,faktor,stackengine,ifthen}
\usepackage{booktabs,longtable,textcomp,microtype}
\usepackage{collectbox,fullpage,authblk}

\theoremstyle{plain}
\newtheorem{theorem}{Theorem}[section]
\newtheorem{corollary}[theorem]{Corollary}
\newtheorem{lemma}[theorem]{Lemma}
\newtheorem{proposition}[theorem]{Proposition}

\theoremstyle{definition}
\newtheorem{definition}[theorem]{Definition}
\newtheorem{definitions}[definition]{Definitions}
\newtheorem{convention}[definition]{Convention}

\newtheorem{example}[definition]{Example}
\newtheorem{examples}[definition]{Examples}
\theoremstyle{remark}
\newtheorem{notation}[definition]{Notation}
\newtheorem{notations}[notation]{Notations}
\newtheorem{remark}[notation]{Remark}

\setlist[enumerate,1]{label={(\roman*)}}

\newcommand{\mathwrap}[1]{\ensuremath{#1}\xspace}
\newcommand{\RR}{\mathbb{R}}
\newcommand{\NN}{\mathbb{N}}
\newcommand{\im}{\mathrm{im}\,}
\renewcommand{\H}{\mathrm{H}}
\newcommand{\Fun}[1]{\mathwrap{\mathbf{Fun}(#1)}}
\newcommand{\Set}{\ensuremath{\mathbf{Set}}\xspace}
\newcommand{\Top}{\ensuremath{\mathbf{Top}}\xspace}
\newcommand{\Toplc}{\ensuremath{\mathbf{Top}^{\mathbf{lc}}}\xspace}
\newcommand{\PCSh}{\ensuremath{\mathbf{PCSh}}\xspace}
\newcommand{\CSh}{\ensuremath{\mathbf{CSh}}\xspace}
\newcommand{\El}{\ensuremath{\mathbf{El}}\xspace}
\renewcommand{\Vec}[1][]{\ensuremath{\mathbf{Vec^{#1}}}\xspace}
\newcommand{\N}{\mathcal{N}}
\newcommand{\U}{\mathcal{U}}
\newcommand{\CShsp}{\ensuremath{\mathbf{CSh}^\mathbf{sp}}\xspace}
\newcommand{\M}{\mathcal{M}}
\newcommand{\defn}[1]{\uline{#1}}
\renewcommand{\P}{\mathwrap{\mathbf{P}}}

\newcommand{\T}{\mathwrap{\mathbf{T}}}
\newcommand{\sT}{\mathscr{T}}
\newcommand{\cT}{\mathcal{T}}
\newcommand{\cat}[1]{\mathwrap{\mathbf{#1}}}
\newcommand{\Q}{\cat{Q}}
\newcommand{\meet}{\wedge}
\newcommand{\join}{\vee}
\newcommand{\pro}[1]{#1^\meet}
\newcommand{\ind}[1]{#1^\join}
\newcommand{\Meet}{\bigwedge}
\renewcommand{\Join}{\bigvee}
\newcommand{\Op}{\mathbf{Op}}
\newcommand{\pOp}{\mathbf{pOp}}
\newcommand{\Pow}{\mathscr{P}}
\newcommand{\defined}{\coloneqq}
\newcommand{\lscale}[2]{\xspace#1#2\xspace}
\newcommand{\mchoicelscale}[2]{\mathchoice%
    {\displaystyle\lscale{#1\!}{\displaystyle#2}}
    {\lscale{#1\!}{#2}}
    {\scriptstyle\lscale{#1}{\scriptstyle#2}}
    {\scriptscriptstyle\lscale{#1}{\scriptscriptstyle#2}}}
\newcommand{\up}[1]{\mchoicelscale{\uparrow}{#1}}
\newcommand{\down}[1]{\mchoicelscale{\downarrow}{#1}}
\newcommand{\dis}{\mathscr{D}}
\newcommand{\oup}[1]{\mchoicelscale{\upharpoonleft}{#1}}
\newcommand{\odown}[1]{\mchoicelscale{\downharpoonleft}{#1}}
\newcommand{\id}{\mathsf{1}}
\newcommand{\sU}{\mathscr{U}}
\newcommand{\cech}{\check{\mathcal{C}}}

\newcommand{\ball}{\mathrm{ball}}
\newcommand{\dist}{\mathrm{d}}
\newcommand{\e}{\varepsilon}
\renewcommand{\ll}{\left}
\newcommand{\rr}{\right}
\renewcommand{\b}[1]{\ll[#1\rr]}
\renewcommand{\L}{\mathcal{L}}
\newcommand{\op}[1]{#1^{\mathsf{op}}}
\tikzset{%
    symbol/.style={%
        draw=none,
        every to/.append style={%
            edge node={node [sloped, allow upside down, auto=false]{$#1$}}}
    }
}
\newcommand{\tikzwrap}[2]{\tikzset{ampersand replacement=\&}\begin{tikzcd}#2#1\end{tikzcd}}
\NewDocumentCommand{\adjoint}{ m m O{""} O{""} O{} }{
\tikzwrap{
    #1  \arrow[r, shift left=0.62ex, phantom, ""{name=x,above}]
        \arrow[r, shift left=0.85ex, #3]
  \& #2 \arrow[l, shift left=0.42ex, phantom, ""{name=y,below}]
        \arrow[l, shift left=0.85ex, #4]
        \arrow[from=x, to=y, symbol={\scriptstyle\dashv}]
}{#5}}
\newcommand{\supp}[1]{\underline{#1}}
\newcommand{\Hom}{\mathrm{Hom}}
\newcommand{\proj}{\varprojlim}
\renewcommand{\lim}{\varprojlim}
\newcommand{\colim}{\varinjlim}

\newcommand{\asim}{%
\mathchoice{\raisebox{-4pt}{$\displaystyle\sim$}}
           {\raisebox{-4pt}{$\sim$}}
           {\raisebox{-2pt}{$\scriptstyle\sim$}}
           {\raisebox{-2pt}{$\scriptscriptstyle\sim$}}}
\newcommand{\longisorightarrow}{\stackrel{\asim}{\longrightarrow}}
\newcommand{\F}{\mathcal{F}}
\newcommand{\h}{\mathrm{h}}
\newcommand{\pt}{\mathbf{pt}}
\newcommand{\Nb}{\mathbf{Op}}
\newcommand{\sE}{\mathscr{E}}
\newcommand{\yields}{\rightsquigarrow}
\newcommand{\sM}{\mathscr{M}}
\newcommand{\ob}{\mathrm{Ob}}
\newcommand{\mor}{\mathrm{Mor}}
\newcommand{\Ob}{\ob}
\newcommand{\Mor}{\mor}
\newcommand{\bdot}{\boldsymbol{\cdot}}
\newcommand{\X}{\mathcal{X}}
\newcommand{\tightoverset}[2]{\mathop{#2}\limits^{\vbox to -.5ex{\kern-0.75ex\hbox{$#1$}\vss}}}
\newcommand{\tightunderset}[2]{\mathop{#2}\limits_{\vbox to .5ex{\kern-1.25ex\hbox{$#1$}\vss}}}
\newcommand{\Longuto}{\tightunderset{\bdot}{\Longrightarrow}}

\newcommand{\longuto}{\tightunderset{\bdot}{\longrightarrow}}

\makeatletter
\let\vvec\vec
\DeclareRobustCommand{\cev}[1]{%
  \mathpalette\do@cev{#1}%
}
\newcommand{\do@cev}[2]{%
  \fix@cev{#1}{+}%
  \reflectbox{$\m@th#1\vvec{\reflectbox{$\fix@cev{#1}{-}\m@th#1#2\fix@cev{#1}{+}$}}$}%
  \fix@cev{#1}{-}%
}
\newcommand{\fix@cev}[2]{%
  \ifx#1\displaystyle
    \mkern#23mu
  \else
    \ifx#1\textstyle
      \mkern#23mu
    \else
      \ifx#1\scriptstyle
        \mkern#22mu
      \else
        \mkern#22mu
      \fi
    \fi
  \fi
}

\title{Topological Hierarchical Decompositions}

\author[1]{Ian Stewart Joyce\thanks{\href{mailto:ian.joyce.1@us.af.mil}{\tt ian.joyce.1@us.af.mil}}}
\author[1]{Grant Erdmann\thanks{\href{mailto:grant.erdmann@us.af.mil}{\tt grant.erdmann@us.af.mil}}}
\author[1]{Kirk P. Gardner\thanks{\href{mailto:k.gardner48@gmail.com}{\tt k.gardner48@gmail.com}}}
\author[2]{Ryan Kramer\thanks{\href{mailto:ryan@minedxai.com}{\tt ryan@minedxai.com}}}
\author[2]{Kyle Siegrist\thanks{\href{mailto:kyle@minedxai.com}{\tt kyle@minedxai.com}}}

\affil[1]{AFRL}
\affil[2]{Mined XAI}

\date{}

\begin{document}

\maketitle

\begin{abstract}
  Topological data analysis is an emerging field that applies the study of topological invariants to data.
  Perhaps the simplest of these invariants is the number of connected components or \emph{clusters}.
  In this work, we explore a topological framework for cluster analysis and show how it can be used as a basis for explainability in unsupervised data analysis.
  Our main object of study will be hierarchical data structures referred to as \emph{Topological Hierarchical Decompositions (THDs)}~\cite{brown18heloc,holmes20topological,brown22topological}.
  We give a number of examples of how traditional clustering algorithms can be topologized,
  and provide preliminary results on the THDs associated with Reeb graphs~\cite{desilva16categorified} and the mapper algorithm~\cite{nicolau11topology}.
  In particular, we give a generalized construction of the mapper functor as a ``pixelization''~\cite{botnan20relative}
  of a cosheaf as alluded to in~\cite{brown21probabilistic} in order to generalize multiscale mapper~\cite{dey16multiscale}.
\end{abstract}

\tableofcontents

\section{Introduction}\label{sec:intro}

Clustering is an essential tool in modern data analysis.
Broadly speaking, the goal of clustering is to break data up into connected components.
In principle, this is done by imposing a suitable notion of connectivity, often defined by a metric, which reflects the local structure of the underlying space.
However, the reliability of the resulting model is dependent not only on a suitable measure of connectivity, but also on a correct choice of scale.
For example, in density-based clustering, data is modeled as a finite metric space in which the connectivity between two points is parameterized by not only the distance between them, but also the local density.
If the number of points sampled in a neighborhood does not reflect the local density of the underlying space
then important topological features may go unnoticed.

\begin{figure}[ht]
  \centering
  \includegraphics[width=\textwidth]{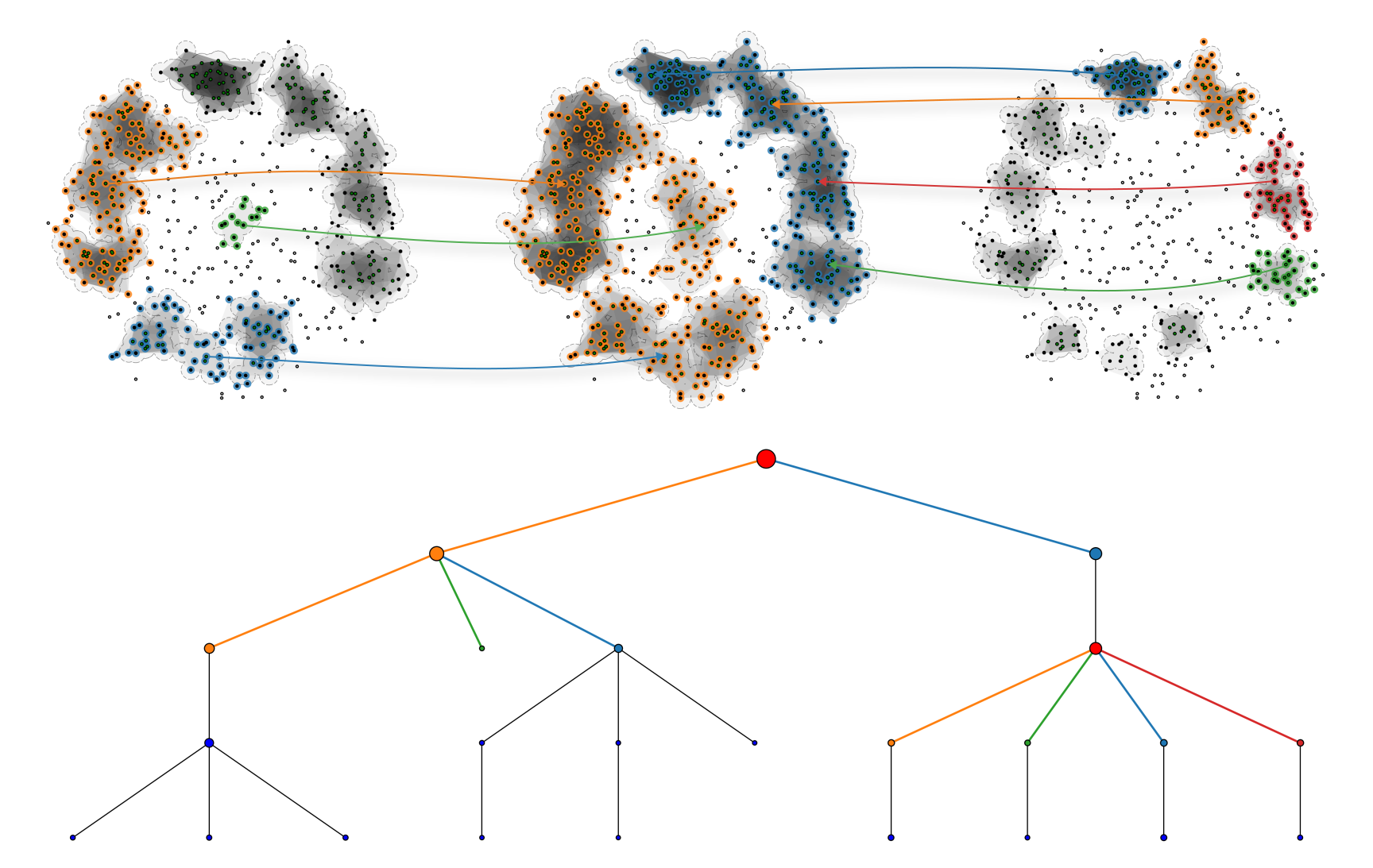}
  \caption{Hierarchical density-based clustering in $\RR^2$ with increasing radius (MNIST~\cite{grother95nist,cohen17emnist} 512 points, 2D MDS~\cite{carroll98multidimensional}).
          The left and right clusterings (top) correspond to the left and right sub-trees of the dendrogram (bottom),
          which serves as our prototypical Topological Hierarchical Decomposition (THD).
          The input is a filtration of spaces given by the union of metric balls centered at the sample points.
          Vertices of the THD correspond to clusters at each step of the filtration,
          with edges showing how these clusters are merged as the cover radius is increased.}\label{fig:dendrogram}
\end{figure}

Hierarchical clustering approaches a partial solution to this issue by parameterizing the choice of scale.
The result is a hierarchy of partitions with nodes corresponding to connected components and edges depicting how these components split (or merge) as the scale varies monotonically.
Paths from the root track data points with increasing (or decreasing) resolution, and the sequence of clusters containing a given sample point provides a novel signature for the data (Figure~\ref{fig:dendrogram}).
However, not all data admits a pointwise measure of similarity required for many standard approaches, and choosing an improper measure may lead to a model that does not accurately represent the underlying space.
That is, meaningful analysis requires not only considering scale, but also a qualitative measure of connectivity that reflects the expected local structure.

In this work, we explore a topological framework for clustering in which the connectivity between data points is qualified by open subsets of a larger topological space.
This perspective allows for powerful abstractions such as sheaves and cosheaves to be applied to clustered data, and the rich theory associated with these tools,
as well as recent progress on persistent homology~\cite{botnan20relative}, Reeb graphs~\cite{desilva16categorified}, and generalized merge trees~\cite{curry21decorated},
provide a rigorous framework for explainability.
Our results will focus on a topological summary referred to as a \emph{Topological Hierarchical Decomposition (THD)}~\cite{brown18heloc,holmes20topological,brown22topological}
generalizing dendrograms of traditional hierarchical clustering
whose accuracy with respect to the underlying ground truth can be theoretically verified.
Moreover, the topological perspective lends itself to a ``pointless'' model that focuses on the lattice of open sets instead of the individual data points, allowing models to be naturally scaled.

\subsection{Related Work}\label{sec:related}

Given a continuous function $f : Y\to X$ the pre-images of each open set $U$ of the codomain $X$ are open sets $f^{-1}(U)$ of the domain $Y$.
Using the topology of the space $Y$, these open sets can be individually clustered and arranged as a topological space over $X$ known as its \emph{Reeb graph} of $f$ (Figure~\ref{fig:reeb}, middle).
In effect, the Reeb graph provides a simplification or \emph{projection} of the space $Y$ through the lens of the continuous function $f$ which serves as a filter.

\begin{figure}[ht]
  \centering
  \includegraphics[width=\textwidth]{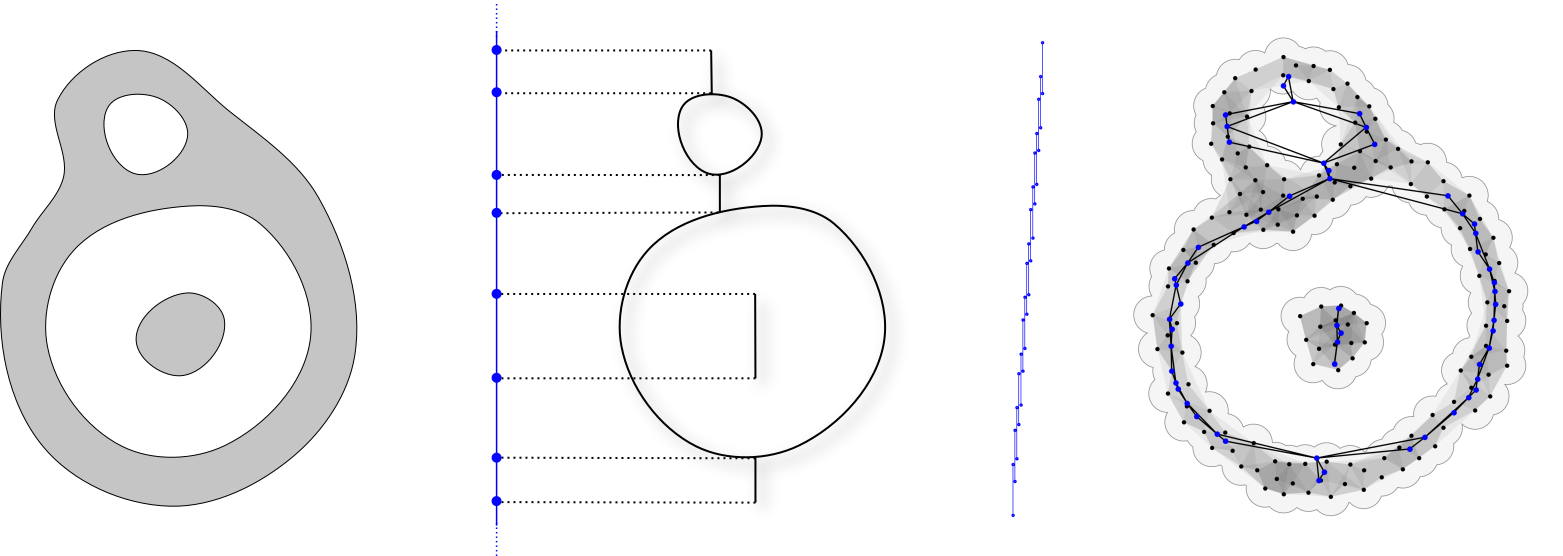}
  \caption{A subspace $X$ of $\RR^2$ (left),
          the Reeb graph of the projection $f : X\to \RR$ onto the $y$-coordinate of $X$ (mid),
          and a Mapper Reeb graph of a cover of $\im f\subset \RR$ (right).}\label{fig:reeb}
\end{figure}

In practice, this construction is often restricted to a finite cover of the space $X$.
The construction detailed above, along with this restriction, is known as the mapper algorithm~\cite{nicolau11topology},
and provides a \emph{pixelization}~\cite{botnan20relative} of the Reeb cosheaf
that can be efficiently represented as the nerve of the resulting cover (Figure~\ref{fig:reeb}, right).

Recent work on the categorification of Reeb graphs~\cite{desilva16categorified}
and persistence modules~\cite{bubenik14categorification,patel18generalized,botnan20relative}
has allowed for powerful theoretical tools from persistent homology such as the \emph{interleaving distance}
to be used in order to provide stability results for mapper~\cite{munch15convergence,carriere18structure,brown21probabilistic}.
That is, the traditional mapper algorithm computes an approximation of the Reeb graph~\cite{munch15convergence,botnan20relative,brown21probabilistic},
and the quality of the approximation is dependent on the choice of function as well as the choice of cover%
\footnote{
  Recent work by Carriere and Oudot~\cite{carriere18structure} analyze the stability of the mapper algorithm with respect to these parameters,
  but we will not address the stability of mapper in this work.}.
Multiscale mapper~\cite{dey16multiscale} proposes a partial solution to this problem by instead considering a filtration or \emph{tower} of covers.
That is, multiscale mapper looks at how the output of the mapper algorithm evolves as the cover is refined.
The result is a sequence of simplicial complexes connected by well-defined simplicial maps (Figure~\ref{fig:multiscale_mapper}).

Previous work on THDs analyzed Home Equity Line of Credit (HELOC) financial data~\cite{brown18heloc,brown22topological} and averaged sentence embeddings~\cite{holmes20topological}.
For the purposes of this work, it suffices to define a THD as the merge tree of a filtration of spaces, as generalized by Curry et. al~\cite{curry21decorated},
with the THDs in previous work corresponding to the (generalized) merge tree of the filtration of simplicial maps associated with multiscale mapper.
In future work, we will explore examples of THDs that do not necessarily arise from the connected components cosheaf~\cite{desilva16categorified}.
For example, we conjecture that the (co)homology functors have associated THDs
that are intimately related to the barcodes or diagrams of persistence modules~\cite{bauer20persistence,chubet22theory,gardner22verified}
and their decorated counterparts~\cite{curry21decorated},
with $\H_0 : \Top\to \Vec$ corresponding to the $\pi_0 :\Top\to \Set$.

\begin{figure}[ht]
  \includegraphics[width=\textwidth]{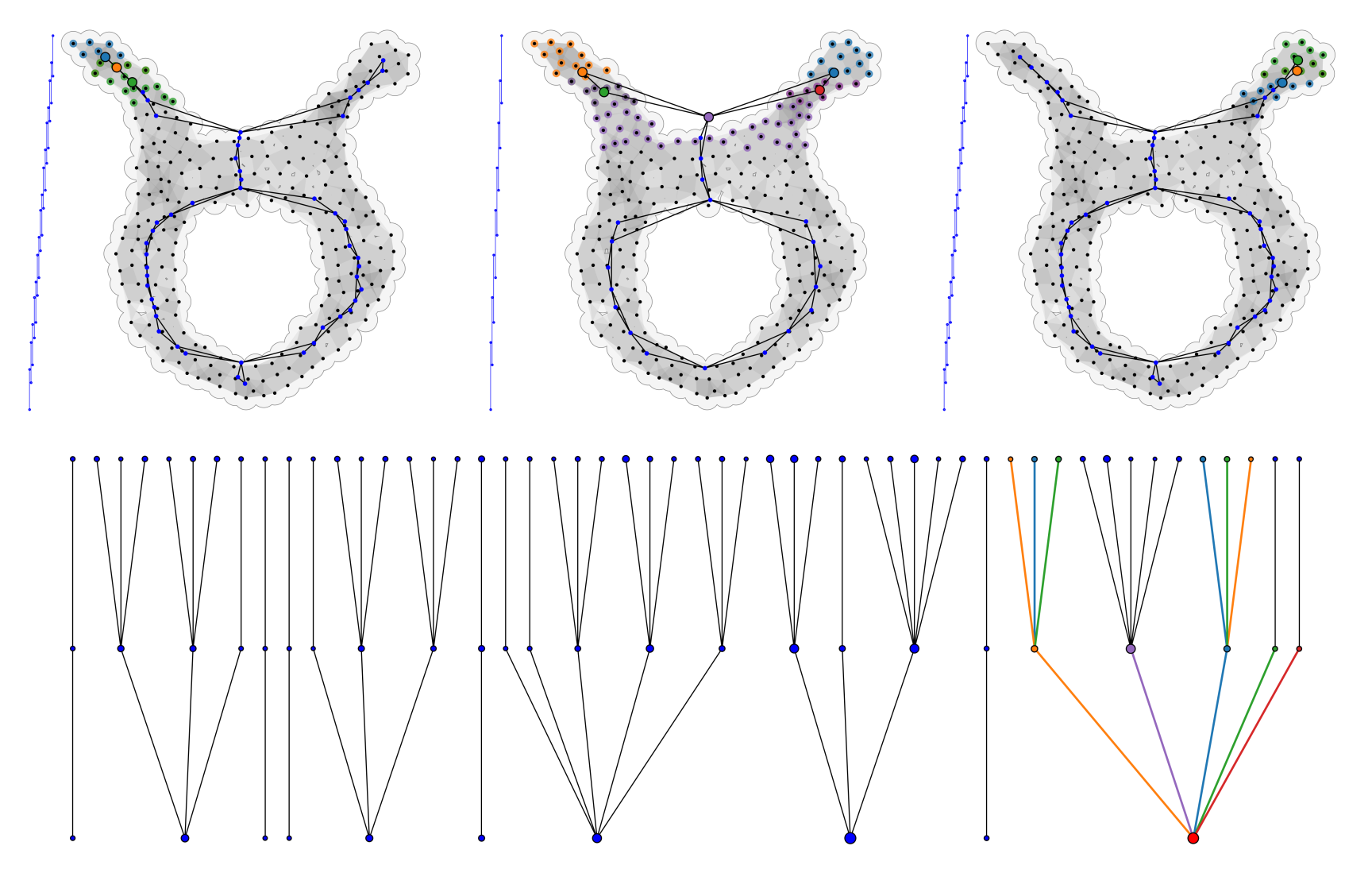}
  \caption{Multiscale mapper (top) and the corresponding THD (bottom).
          The filter function $f_P : P\to \RR$ takes each point of a finite sample in $\RR^2$ to its $y$-coordinate.
          The covers of $\im f\subset \RR$ (left of each, in blue) on the left and right (top) are refinements of the middle.}\label{fig:multiscale_mapper}
\end{figure}

\subsection{Contributions}\label{sec:contributions}

The theoretical contributions of this work are in the field of topological data analysis
and will focus on two recent results that allude to a more general construction of mapper.
\begin{enumerate}
  \item The theory of relative interleavings by Botnan, Curry, and Munch~\cite{botnan20relative}
    provide a novel interpretation of the interleaving distance for generalized persistence modules.
    In particular, the authors introduce the notion of a pixelization of a functor by the pushforward of the pullback of a monotone function,
    and provide a construction on the so-called mapper functor as a pixelization of the Reeb cosheaf.
  \item More recent results by Brown, Bobrowski, Munch, and Wang~\cite{brown21probabilistic}
    provides an explicit construction of the mapper functor for real-valued functions in order to define the
    \emph{enhanced mapper graph} as the display space of the resulting cosheaf.
    Importantly, the authors note that the mapper functor may be defined more generally as the pushforward of the pullback of a continuous function $\eta : X\to \N_\U$.
    That is, as a pixelization of the Reeb cosheaf as in the work of Botnan, Curry, and Munch~\cite{botnan20relative}.
\end{enumerate}

Our main contributions are as follows:
\begin{enumerate}
  \item In Section~\ref{sec:mapper}, we give a construction of the mapper functor $\M_\U : \PCSh(X)\to \CShsp(X)$ of~\cite{brown21probabilistic}
    as a pixelization of a precosheaf by a cover as alluded to in~\cite{botnan20relative} and~\cite{brown21probabilistic} (Theorem~\ref{thm:mapper}).
  \item In Section~\ref{sec:multi}, we extend this construction to multiscale mapper~\cite{dey16multiscale} and show that the resulting functor commutes with finite limits (Theorem~\ref{thm:multi_iso}).
  \item In Section~\ref{sec:multi_examples}, We show how the output of multiscale mapper can be expressed as the merge tree of a filtration of spaces (Proposition~\ref{prop:multi_thd}),
    providing an example of a \emph{Topological Hierarchical Decomposition (THD)} that corresponds to a divisive hierarchical clustering scheme.
\end{enumerate}
Although these results are technical, we provide a detailed discussion of clustering from the topological perspective that is supported by a faithful implementation of many of the ideas presented in this work.

Following de Silva et al.~\cite{desilva16categorified}, Munch and Wang~\cite{munch15convergence} and their more recent work with Brown~\cite{brown21probabilistic},
our main object of study will be \emph{cosheaves} which provide a way to associate data with the open neighborhoods of a topological space in a way that enforces \emph{consistency}.
Throughout, we pay specific attention to earlier work by Funk~\cite{funk95display} that established the theory of cosheaves and their display locales.
The resulting spatial cosheafification functor~\cite{woolf08fundamental} aligns this work with the dual theory of sheaves of topological spaces,
roughly corresponding to sheafification by factoring a presheaf through the corresponding \'{e}tale space (see Mac Lane and Moerdijk~\cite{maclane92sheaves} or Kashiwara and Schapira~\cite{kashiwara05categories}).

\subsection{Overview}\label{sec:overview}

We will begin with a review of topological spaces and cosheaves in Section~\ref{sec:background}%
\footnote{We refer the interested reader to Mac Lane~\cite{mac13categories,maclane92sheaves} or Kashiwara and Schapira~\cite{kashiwara05categories} for a full treatment.}.
For additional background on the relevant categories of arrows and functors, as well as adjoints and limits see Appendix~\ref{sec:cats2}.
In Section~\ref{sec:clustering_and_thds}, we discuss clustering from the topological perspective and introduce the (generalized) merge tree~\cite{curry21decorated} as our prototypical THD.
In Section~\ref{sec:mapper}, we review the adjoint Reeb and display space functors~\cite{funk95display}
and show that the mapper (pre-)cosheaf $\M_\U(F)$ of a precosheaf $F$ on $X$, as defined in~\cite{brown21probabilistic},
is equivalent to a pixelization~\cite{botnan20relative} $\eta^*\eta_*(F)$ of $F$ by a cover $\U$ (see also Appendix~\ref{sec:inverse}).
We then generalize the mapper construction to filtrations of covers and spaces and discuss their associated THDs in Section~\ref{sec:multi} (see also Appendix~\ref{sec:reeb_future}).

\subsection{Acknowledgments}

This research was supported in part by an appointment to the Department of Defense (DOD) Research Participation Program
administered by the Oak Ridge Institute for Science and Education (ORISE) through an interagency agreement between the U.S.
Department of Energy (DOE) and the DOD. ORISE is managed by ORAU under DOE contract number DE-SC0014664. All opinions
expressed in this paper are the author's and do not necessarily reflect the policies and views of DOD, DOE, or ORAU/ORISE.

We thank Jordan DeSha and Benjamin Filippenko for useful conversations and input on the manuscript.

\section{Preliminaries}\label{sec:background}

The category $\Set$ has sets as objects and functions as arrows.
A \defn{partially ordered set (poset)} $\P \defined (P,\leq)$ is a category whose objects form a set with unique arrows
$b\to c$ denoted $b\leq c$ in which $b\leq c$ and $b\geq c$ implies $b=c$.
A functor $\sigma : \P\to \Q$ between posets is a \defn{monotone function} in which $\sigma(b)\leq \sigma(c)$ in $\Q$ for all $b\leq c$ in $\P$.
The \defn{meet} (resp. \defn{join}) of a pair $b,c\in\P$ is the \emph{greatest lower bound} $b\meet c$ (resp. \emph{least upper bound} $b\join c$) of $b$ and $c$ in \P in which
\begin{align*}
                a   &\leq b\meet c  \leq b,c\ \text{ for all }\ a  \leq b,c\\
  \big(resp.\ \ b,c &\leq b\join c  \leq d\ \ \ \text{ for all }\ d  \geq b,c\big).
\end{align*}
A poset is a \defn{lattice} if it has all (binary) meets and joins, and \defn{totally ordered} if either $b\leq c$ or $b\geq c$ for all elements $b$ and $c$.

Given a subset $S\subseteq \P$, we write $a \leq S$ (resp. $d \geq S$) if $a\leq s$ (resp. $d\geq s$) for all $s\in S$.
The meet (resp. join) of a subset $S$ is an element $\Meet S$ (resp. $\Join S$) of \P in which
\begin{align*}
                a   &\leq \Meet S  \leq S\ \text{ for all }\ a  \leq S\\
  \big(resp.\ \ S &\leq \Join S  \leq d\ \text{ for all }\ d  \geq S\big).
\end{align*}
Given a functor $F :\cat{C}\to \P$ from a category $\cat{C}$ to a poset $\P$,
let $\Meet F = \Meet_{C\in\cat{C}} F(C)$ (resp. $\Join F = \Join_{C\in\cat{C}} F(C)$) denote the meet (resp. join) of $F$ in \P.

For any $A\in\Set$, the powerset $\Pow_A$ of $A$ is the lattice of subsets of $A$ ordered by inclusion.
For any function $f : A\to B$ in \Set let $f(-) : \Pow_A\to \Pow_B$ denote the \emph{image} of $f$ defined for $S\subseteq A$ as the monotone function $f(S) = \big\{f(a)\in B\mid a\in S\big\}$,
and let $f^{-1} : \Pow_B\to \Pow_A$ denote the \emph{pre-image} of $f$ defined for $T\subseteq B$ as the monotone function $f^{-1}(T) = \big\{ a\in A\mid f(a)\in T\big\}$.
Given a functor $F : \cat{C}\to \Set$ from a category \cat{C} to \Set, let $\Meet F \defined \bigcap_{C\in\cat{C}} F(C)$ (resp. $\Join F \defined \bigcup_{C\in\cat{C}} F(C)$)
denote the intersection (resp. union) of the $F(C)$ in \Set.

\subsection{Topological Spaces}\label{sec:top}

\begin{definition}[Topological Space]
  A \defn{topological space} is a pair $(X,\Op_X)$ where $X$ is a set and $\Op_X$ is a sublattice of $\Pow_X$ consisting of \emph{open subsets} $U\subseteq X$ such that
  \begin{enumerate}
    \item $\emptyset$ and $X$ are open,
    \item the intersection of finitely many open subsets is open, and
    \item the union of arbitrarily many open sets is open.
  \end{enumerate}
  We write $X$ to denote a topological space $(X,\Op_X)$ when no confusion may occur.
\end{definition}

A function $f : Y\to X$ is \defn{continuous} if the pre-image $f^{-1}(U)$ is open in $Y$ for each open set $U\in\Op_X$.
The category $\Top$ has topological spaces as objects and continuous functions as arrows.
A continuous function $f : Y\to X$ is
\defn{essential} if its pre-image $f^{-1} : \Op_X\to \Op_Y$ has a left adjoint $f_! : \Op_Y\to \Op_X$ (see Appendix~\ref{sec:limits}).

\begin{definition}[The Subspace Topology]
  A subspace of a topological space $X$ is a topological space given by a subset $S\subset X$ with the subspace topology $\Op_{S}\defined \{S\cap U\mid U\in\Op_X\}$.
  For any subspace $S\subset X$, let $\iota_S : S\hookrightarrow X$ denote the canonical embedding with $\iota_S^{-1}(U) = S\cap U$,
  and for any continuous function $f : X\to Y$ let $f |_S\defined f\circ \iota_S : S\to Y$ denote the restriction of $f$ to $S$.
\end{definition}

\begin{definitions}
  For any element $p$ of a poset $\P$ let $\up{p} \defined \{q\in\P\mid p\leq q\}$ (resp. $\oup{p} = \{q\in\P\mid p < q\}$) denote the \defn{closed (resp. open) principal up set} at $p$
  and let $\down{p} \defined \{q\in \P\mid p\geq q\}$ (resp. $\odown{p} = \{q\in\P\mid p > q\}$) denote the \defn{closed (resp. open) principal down set} at $p$.
  \begin{enumerate}
    \item The \defn{specialization (Alexandroff) topology} on a poset \P is the topology generated by closed principal up sets $\up{p}$ for $p\in \P$.
      Dually, the \defn{\emph{co}specialization topology} on $\P$ is the topology generated by closed principal \emph{down} sets $\down{p}$ for $p\in \P$.
    \item The \defn{right (resp. left) order topology} on a \emph{totally ordered} set $\T$ is the topology generated by \emph{open} principal up sets $\oup{t}$ (resp. $\odown{t}$) for $t\in\T$.
      Naturally, the \defn{order topology} on $\T$ is the topology generated by the open principal up \emph{and} down sets.
  \end{enumerate}
\end{definitions}

\begin{example}[The Real Numbers]
  The standard topology on the real numbers $\RR$ is generated by open intervals $(a,b)$ for all $a < b$ in $\RR$
  and is equivalent to the order topology on the totally ordered set $\cat{R} = (\RR,\leq)$.
  Let $\cat{R}_+$ denote the totally ordered set of non-negative real numbers and let $\RR_+$ (resp. $\pro{\RR_+}$, $\ind{\RR_+}$)
  denote the topological space of non-negative real numbers endowed with the order (resp. left order, right order) topology.
\end{example}

\begin{definition}[The Connected Components Functor]
  A topological space $X$ is \defn{connected} if it is not the disjoint union of two nonempty open sets,
  and two elements $x,y\in X$ are connected if they are both contained in a connected open set.
  The \defn{connected component} $\b{x}_X$ of a point $x\in X$ is the set of all points $y\in X$ that are connected to $x$.
  The \defn{connected components functor} $\pi_0 : \Top\to \Set$ takes each topological space to the set of its connected components.
\end{definition}

\begin{definitions}
  For any $x\in X$ let $\Op_X(x)\defined \{U\in\Op_X\mid x\in U\}$ denote the subposet of open neighborhoods of $x$.
  $X$ is \defn{locally finite} if $\Op_X(x)$ is finite for all $x\in X$ and \defn{locally connected} if it admits a basis of connected open sets.
  Let $\Toplc$ denote the full subcategory of $\Top$ restricted to locally connected topological spaces.
\end{definitions}

\begin{remark}
  The restriction of $\pi_0$ to locally connected topological spaces is left adjoint to the discrete space functor $\Delta : \Set\to \Top$ taking each set $A$ to the topological space $\Delta(A) \defined (A,\Pow_A)$,
  which is itself left adjoint to the forgetful functor $\Gamma : \Top\to \Set$.
\end{remark}

\subsection{Covers and Nerves}\label{sec:covers}

Let $X$ be a topological space.
A \defn{cover} of $X$ is a collection $\U = \big\{ \U(i) \big\}_{i\in I}$ of open sets $\U(i)\in \Op_X$ indexed by a set $I$ such that $X = \bigcup_{i\in I} \U(i)$.
In the following, we will regard covers as functors $\U : I\to \Op_X$ from a discrete category $I$ such that $\U(i)\neq \U(j)$ for all $i\neq j$ in $I$.
$\U$ is a \defn{good (open) cover} if the intersection $\bigcap_{i\in\sigma} \U(i)$ of cover sets is empty or contractible for all finite $\sigma\subseteq I$,
and \defn{locally finite} if each $x\in X$ is contained in finitely many cover sets.

\begin{definition}[Nerve]
  The \defn{nerve} of a cover $\U$ is a poset of \emph{finite} subsets $\sigma\subset I$ ordered by inclusion defined
  \[ \N_\U \defined\big\{\sigma\subseteq I\mid \bigcap_{i\in\sigma} \U(i) \neq \emptyset\big\}.\]
  We endow the nerve of a cover with the specialization topology $\Op_{\N_\U}$ generated by principal up sets
  $\up{\sigma}\defined \{ \tau\in\N_\U\mid \tau\supseteq \sigma\}$.
\end{definition}

\begin{notations}\label{notation:nerve}
  For any locally finite cover $\U : I\to \Op_X$ let $\eta_\U : X\to \N_\U$ denote the canonical map that takes each point $x\in X$ to the element (simplex) of $\N_\U$ corresponding to cover sets containing $x$:
  \[ \eta_\U(x) \defined \{ i\in I\mid x\in \U(i)\}.\]
  We omit the subscript and write $\eta : X\to \N_\U$ when no confusion may occur.
  For any $\sigma\in\N_\U$ let $\U_\sigma\defined \bigcap_{i\in\sigma} \U(i)$ denote the corresponding intersection of cover sets,
  and for any $x\in X$ let $\U_x\defined \U_{\eta(x)}$ denote the smallest open set containing $x$ that is supported by $\U$.
  For any open set $S\in\Op_{\N_\U}$ let $\U_S \defined \bigcup_{\sigma\in S}\U_\sigma$
  and for any $U\in\Op_X$ let $\U_U \defined \bigcup_{x\in U} \U_x$.
\end{notations}

\begin{convention}\label{con:surjective}
  We implicitly assume that  $\U_\sigma \neq \U_\tau$ for all $\sigma\neq \tau$ in $\N_\U$ so that $\eta$ is surjective
  and $\N_\U$ is in bijective correspondence to the \emph{basic} open cover $\sigma\mapsto \U_\sigma$ of $X$ indexed by $\N_\U$.
  In practice, this convention amounts to removing any redundancies in the cover which may be formalized by defining the nerve as a poset of cover sets $\U_\sigma$ as in~\cite{brown21probabilistic}.
\end{convention}

We often require that $\U$ be a \emph{locally finite good open cover}
so that $\U_x$ is a contractible (and therefore connected) open set for all $x\in X$.
This allows us to make use of the following preliminary result
which implies that $\eta^{-1}$ has a left adjoint $\eta_!\dashv \eta^{-1}$ defined as the upward closure of the image
$\eta_!(U) = \up{\eta(U)} = \bigcup_{x\in U} \up{\eta(x)}$.
Its proof can be found in Appendix~\ref{sec:proofs}.

\begin{proposition}\label{prop:eta_open}
  If $\U$ is a locally connected good open cover of $X$ then $\eta$ is essential.
\end{proposition}

We also make use of the following standard result in order to cluster a sample of a topological space by pixelizing by a cover that satisfies mild regularity conditions.
Its usefulness is due primarily to the fact that the number of connected components is a topological invariant, and is therefore preserved under homotopy.

\begin{theorem}[The Nerve Theorem (Hatcher~\cite{hatcher01} Corollary 4G.3)]\label{thm:nerve}
  If $\U$ is a good open cover of $X$ then $\N_\U$ is homotopy equivalent to $X$.
\end{theorem}

\subsection{Cosheaves}\label{sec:cosheaves}

Let $X$ be a topological space.
A \defn{precosheaf} on $X$ is a functor $F : \Op_X\to \Set$ associating each open set $U\in\Op_X$ with a \emph{set} $F(U)$.
A precosheaf is a \emph{cosheaf} if it commutes with colimits as defined in Appendix~\ref{sec:cosheaves2}.
Let $\PCSh(X)\defined \Fun{\Op_X,\Set}$ denote the category of precosheaves on $X$
and let $\CSh(X)$ denote the category with objects being the cosheaves on $X$.

Because the restriction of the connected components functor to locally connected spaces is left adjoint to the discrete space functor $\Delta : \Set\to \Toplc$,
the induced functor $\pi_X :\Op_X\to \Set$ defined $U\mapsto \pi_0(U)$ (with $U$ regarded as a subspace) is a cosheaf whenever $X$ is locally connected~\cite{woolf08fundamental},
and will be referred to as the \defn{connected components cosheaf} on $X$.

\begin{definition}[Direct Image]
  The \defn{(cosheaf-theoretic) direct image functor} associated with a continuous function $f : Y\to X$ is given by precomposition with $f^{-1} :\Op_X\to \Op_Y$:
  \begin{align*}
    f_* : \PCSh(X) &\longrightarrow \PCSh(Y)\\
    F &\longmapsto Ff^{-1}.
  \end{align*}
  If $F$ is a cosheaf on $X$ then $f_*(F)$ is a cosheaf on $Y$.
\end{definition}

\begin{definition}[The Reeb Functor~\cite{desilva16categorified}]
  The \defn{Reeb functor} $\lambda : \Top/X\to \PCSh(X)$ takes each space over $X$ (see Appendix~\ref{sec:slice}) to the direct image of the connected components (pre)cosheaf $\pi_Y : \Op_Y\to \Set$ along $f$:
  \begin{align*}
    \lambda : \Top/X&\longrightarrow \PCSh(X)\\
    (Y,f) &\longmapsto f_*(\pi_Y) = \pi_Y f^{-1}
  \end{align*}
  Importantly, the Reeb functor takes locally connected spaces $(L,f)$ over $X$ to cosheaves on $X$.
  Formally, the \defn{Reeb cosheaf} $R_f : \Op_X\to \Set$ of a locally connected space $(L,f)$ over $X$ is the cosheaf on $X$ given by the Reeb functor:
  \begin{align*}
    R_f\defined f_*(\pi_L): \Op_X &\longrightarrow \Set\\
      U &\longmapsto \pi_L\big(f^{-1}(U)\big).
  \end{align*}
\end{definition}

\begin{examples}
  Let $L$ be a locally connected space and let $\id_L : L\to L$ denote the identity.
  The Reeb cosheaf of $(L,\id_L)\in\Toplc/L$ is precisely the connected components cosheaf $R_L = \pi_L : \Op_L\to \Set$.
  Similarly, for any subspace $S\subset L$ the Reeb cosheaf of the space $(S,\iota_S)$ over $L$ is a cosheaf on $X$ that takes each open set $U$ of $X$ to the connected components of $S\cap U$:
  \begin{align*}
    R_S\defined R_{\iota_S} : \Op_X&\longrightarrow \Set\\
      U &\longmapsto \pi_S(S\cap U).
  \end{align*}
\end{examples}

\section{Topological Hierarchical Decompositions}\label{sec:clustering_and_thds}

The goal of this section is to motivate the intuition that the topology of a space determines a ``clustering'' of its points that is carried out by the connected components functor.
We can then exploit the rich theory associated with this functor in order to establish a rigorous topological framework for cluster analysis.
In particular, this allows us to generate hypotheses by applying known theory to experimentally verified conditions.

We begin with a discussion of clustering from the topological perspective in Section~\ref{sec:clustering}, followed by detailed examples of metric and density-based clustering in Sections~\ref{sec:metric} and~\ref{sec:dbscan}.
We then show how hierarchical clustering can be topologized
and introduce our prototypical THD as the (generalized) merge tree~\cite{curry21decorated} of a filtration of spaces in Section~\ref{sec:hierarchical}.

\subsection{Topological Clustering}\label{sec:clustering}

The connected components functor $\pi_0: \Top\to \Set$ takes each topological space $X$ to the set of its connected components,
which may be thought of as a \emph{clustering} $\pi_0(X)$ of its points: a partition of $X$ into connected subsets $\b{x}_X$.
Given a finite sample $P\subset X$, our goal is to compute a clustering of $P$ that is induced by the topology of $X$.
However, the subspace topology on $P$ inherited from $X$ does not provide the desired clustering of the sample points.
For example, for any finite subset $P\subset \RR$, we can construct a cover by open balls containing one point each, thus, $\pi_0(P) \simeq P$.
We therefore define a canonical clustering associated with a subspace as follows.

\begin{definition}[Canonical Clustering]
  Let $X$ be a topological space.
  For any sample $S\subset X$ the associated \defn{canonical clustering functor} takes each open set $U$ of $X$ to the corresponding clustering of $S\cap U$:
  \begin{align*}
    \pi_X^S : \Op_X&\longrightarrow \Set\\
    U&\longmapsto \big\{S\cap \beta \mid \beta\in \pi_X(U)\big\}
  \end{align*}
  The \defn{canonical clustering} of $S$ is the union of the $\pi_X^S(U)$ in \Set:
  \[ \Join \pi_X^S = \bigcup_{U\in\Op_X} \pi^S_X(U)
    = \{S\cap \beta \mid \beta\in \pi_0(X)\} = \pi_X^S(X).\]
\end{definition}

Alternatively, the connected components of $X$ may be characterized as the connected components of the lattice of open sets~\cite{funk95display}.
That is, we can think of the connected components of $X$ as either a partition of the \emph{points} of $X$ as above,
or as a partition of (connected) open sets%
\footnote{
  Locales~\cite{maclane92sheaves,borceux94handbook3} are generalizations of topological spaces
  that formalize the sense in which the lattice of open sets $\Op_X$ (a Heyting algebra) is dual to the underlying space $X$.
}.
We will therefore regard our input not as a finite sample of points $P\subset X$,
but as a finite sample of (connected) open sets: a good open cover $\U : I\to \Op_X$.
In the next section, we will discuss a common situtaion in which these two representations coincide.

Let $\U : I\to \Op_X$ be a good open cover.
We can use $\U$ to cluster the points of a finite sample $P\subset X$ by precomposition:
\[ \Join \pi_X^P\U = \bigcup_{i\in I} \pi_X^P\big(\U(i)\big) = \bigcup_{U\in\Op_X} \pi_X^P(U) = \Join \pi_X^P.\]
Moreover, by the Nerve Theorem (Theorem~\ref{thm:nerve}), the nerve $\N_\U$ of $\U$ is a topological space that is homotopy equivalent to $X$.
In particular, this means that there is a bijective correspondence between the connected components of the nerve and that of the underlying space $X$.
The resulting clustering $\pi_0(\N_\U)$ of the elements (simplices) of the nerve can then be extended to a clustering of the underlying cover $\U : I\to \Op_X$ as follows.

The image of the canonical map $\eta : X\to \N_\U$ (see Section~\ref{sec:covers}) takes a subset $S\subset X$ to the subset $\eta(S) = \{\eta(x)\mid x\in S\}$ of the nerve $\N_\U$.
The join of the image is therefore the subset
\[ \Join \eta(S) = \{i \in I\mid \U(i)\cap S\neq\emptyset\}\]
of the index set $I$ of $\U$ corresponding to cover sets that intersect $S$.
We define a functor
\begin{align*}
  \pi_\U : \Op_X &\longrightarrow \Set\\
  U &\longmapsto \Big\{\Join \eta(\beta)\mid \beta\in \pi_X(U)\Big\}
\end{align*}
so that $\Join\pi_\U = \big\{ \Join \eta(\beta)\mid \beta\in\pi_0(X)\big\}$ is a \emph{clustering} of $\U$:
a partition of the index set of $\U : I\to \Op_X$ induced by the connected components of $X$.

\subsection{Metric Clustering}\label{sec:metric}

In this section, we will consider a common situation in which the sample points $P\subset X$ generate a cover by metric balls.

\begin{definition}[Metric Space]
  A \defn{metric space} is a pair $(X,\dist)$ where $X$ is a set and $\dist : X\times X\to \RR_+$ is a function satisfying
  \begin{enumerate}
    \item $\dist(x,x) = 0$ for all $x\in X$,
    \item $\dist(x,y) > 0$ for all $x\neq y$,
    \item $\dist(x,y) = \dist(y,x)$ for all $x,y$, and
    \item $\dist(x,y) + \dist(y,x) \geq \dist(x,z)$ for all $x,y,z$.
  \end{enumerate}
  The \defn{metric topology} on $(X,\dist)$ is generated by the collection of basic open sets
  \[ \ball^\e(x) \defined \{ y\in X\mid \dist(x,y) < \e\} \text{ for } x\in X \text{ and } \e > 0.\]
  Equivalently, the pre-image of the function $\dist_x\defined \dist(x,-): X\to\pro{\RR_+}$ at a basic open set $\odown{\e}$ of $\pro{\RR_+}$ is precisely the $\e$-ball centered at $x$,
  so the metric topology is equivalent to the initial topology associated with the family of maps $\dist_x$.
\end{definition}

\begin{definition}[$\e$-offset]
  The map taking each point $x\in X$ to its $\e$-neighborhood will be denoted
  \begin{align*}
    B^\e : X&\longrightarrow \Op_X\\
    x&\longmapsto \dist_x^{-1}(\odown{\e}) = \ball^\e(x)
  \end{align*}
  and for any subspace $S\subset X$, let $B_S^\e\defined B^\e |_S : S\to \Op_X$ denote the restriction of $B^\e$ to $S$.
  The \defn{$\e$-offset} $S^\e$ of a subspace is the join of the restriction $B_S^\e : S\to \Op_X$;
  that is, the union of $\e$-balls centered at the points of $S$:
  \[ S^\e\defined \Join B_S^\e = \bigcup_{s\in S}\ball^\e(s).\]
\end{definition}

\begin{figure}[ht]
  \centering
  \includegraphics[width=0.9\textwidth]{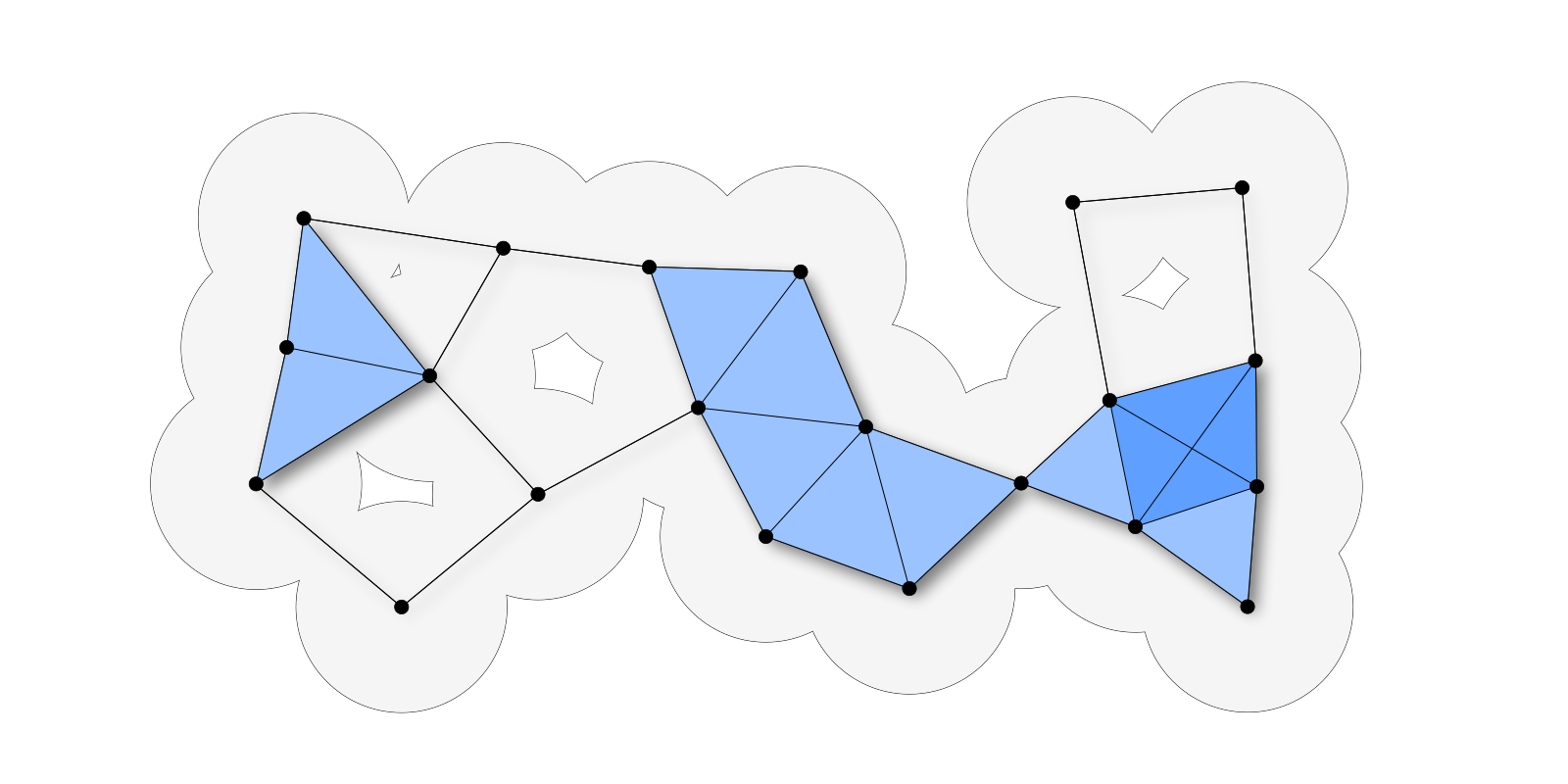}
  \caption{The \v{C}ech complex of a cover $B_P^\e : P\to \Op_X$ of a metric space $X = \Join B_P^\e = P^\e$ by a finite sample $P\subset X$.
          Vertices of the complex correspond to the sample points $p\in P$; more specifically, \emph{$0$-simplices} $\{p\}\in\N_\U$.
          Similarly, edges correspond to two-element subsets $\{p,q\}\in\N_\U$ with $\ball^\e(p)\cap \ball^\e(q)\neq\emptyset$,
          triangles correspond to three-element subsets $\{p,q,r\}\in\N_\U$ with $\ball^\e(p)\cap \ball^\e(q)\cap \ball^\e(r)\neq\emptyset$, etc.}\label{fig:clustering_cech}
\end{figure}

\begin{definition}[\v{C}ech Complex]
  A finite subspace $P\subset X$ is an \defn{$\e$-sample} of $X$ if the restriction $B^\e_P : P\to \Op_X$ is a cover of $X$:
  \[ P^\e = \Join B_P^\e = \bigcup_{p\in P}\dist_p^{-1}(\odown{\e}) = X.\]
  Equivalently, an \defn{$\e$-cover} of a metric space is a cover by metric balls indexed by a finite subset.
  The \defn{\v{C}ech complex} of an $\e$-sample $P$ is the nerve of the corresponding $\e$-cover:
  \[ \cech^\e_P\defined \big\{\sigma\subseteq P\mid \bigcap_{p\in \sigma} \ball^\e(p)\neq \emptyset\big\} = \N_{B^\e_P}.\]
\end{definition}

\begin{example}
  Let $P$ be an $\e$-sample of $(X,\dist)$.
  If the metric balls of $X$ are sufficiently convex~\cite{chazal09analysis} with respect to $\e$,
  then the \v{C}ech complex of an $\e$-sample is equivalent up to homotopy by the Nerve Theorem (Theorem~\ref{thm:nerve}),
  and can therefore be used to cluster the sample points.
  Noting that
  \[ \Join \eta(\beta) = \big\{p\in P\mid \exists x\in \beta \text{ such that } x\in\ball^\e(p)\big\} = P\cap \beta,\]
  it follows that the canonical clustering of $P$ is equivalent to the clustering of the cover $B_P^\e : P\to \Op_X$:
  \[ \Join \pi_X^P = \{ P\cap \beta\mid \beta\in \pi_0(X)\} = \Big\{ \Join \eta(\beta) \mid \beta\in \pi_0(X) \Big\} = \Join \pi_{B_P^\e}.\]
  That is, the canonical clustering of the sample points can be computed directly from the \v{C}ech complex
  (or a sufficient approximation%
  \footnote{In practice, it suffices to compute the image of a pair of Vietoris-Rips complexes that factors through the \v{C}ech,
            which can be done using pairwise proximity information alone~\cite{desilva07coverage,cavanna17when}.})
  using a standard graph search (Figure~\ref{fig:clustering_cech}).
\end{example}

\begin{remark}
  Because $P$ is a \emph{finite} sample of $X$, the $\e$-cover $B^\e_P$ is locally finite,
  so the canonical map $\eta : X\to \cech^\e_P$ is essential and induces a bijective correspondence between the connected components of $\cech^\e_P$ and $X$:
  \begin{align*}
    \pi_0\b{\eta} : \pi_0(X)&\longisorightarrow \pi_0(\cech^\e_P)\\
    \b{x}_X &\longrightarrow \b{\eta(x)}_{\cech^\e_P}
  \end{align*}
  Moreover, Because the canonical map $\eta : X\to \cech^\e_P$ is \emph{essential}, the pre-image $\eta^{-1} : \Op_{\cech^\e_P}\to \Op_X$ has a left adjoint $\eta_! : \Op_X\to \Op_{\cech^\e_P}$.
  Because the composition of $\eta$ with the embedding $\iota_P : P\hookrightarrow X$ is essential as well,
  we obtain a canonical essential map $\eta|_P : P \to \cech^\e_P$ that factors uniquely through $X$.
  In Section~\ref{sec:mapper} we will show how the unit of the adjunction associated with an essential map can be formalized as a \emph{pixelization}~\cite{botnan20relative} of a cosheaf by a locally finite good open cover.
\end{remark}

\subsection{Density-Based Clustering}\label{sec:dbscan}

Let $(X,\dist)$ be a metric space and for any subspace $S\subset X$
let $S_\e\defined \iota_S^{-1}\circ B^\e : X\to \Op_S$ denote the cover of $S$ defined by $x\mapsto S\cap \ball^\e(x)$.

Let $P\subset X$ be an $\e$-sample
such that, for all $x\in X$,
\[ \# P_\e(x) = \#\big(P\cap \ball^\e(x)\big)\geq k \text{ for some } k > 0.\]
That is, $P$ has uniform density in the sense that each $\e$-neighborhood contains at least $k$ points of $P$;
in other words, $P$ is a \emph{$k$-cover} of $X$~\cite{cavanna17when}.
Using this assumption, we can cluster the points of $P$ from a subsample $Q\subset P$ with sufficiently high local density.
This approach is particularly useful in the presence of noise where sample points with insufficient density can be identified as outliers.
More generally, we can extend this approach to arbitrary finite samples $Q\not\subset P$ as follows.

\begin{figure}[ht]
  \centering
  \includegraphics[width=0.65\textwidth]{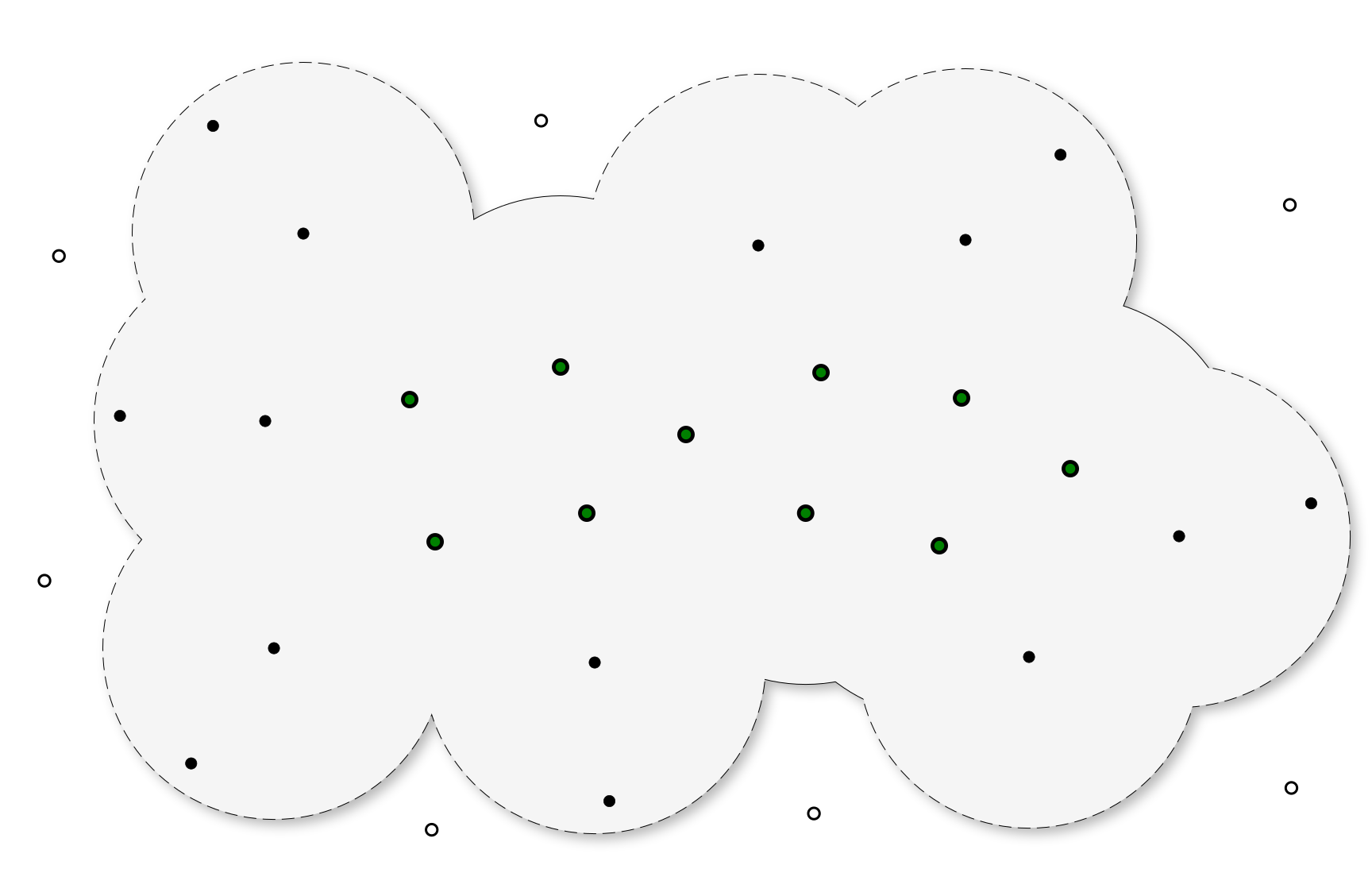}
  \caption{Density-based clustering with $k=4$ and $\delta = \e$.
    Each core point $q\in P_{4,\delta}$ (in green) corresponds to a cover set $B_{P_{4,\delta}}^\e = \bigcup_{p\in P_\delta|_{P_{4,\delta}}} \ball^\e(p)$
    equal to the $\e$-offset of the $\delta$-neighborhood of $q$ in $P$.
    Outlier points (empty circles) are not within $\e$ of any neighborhood point.}\label{fig:dbscan}
\end{figure}

Let $P,Q\subset X$ be finite subspaces and let $\delta,\e > 0$.
If $Q$ is a $\delta$-sample of $X$ then the restriction $P_\delta|_Q : Q\to \Op_P$ of $P_\delta$ to $Q$ is a cover of $P$.
If, in addition, $P$ is an $\e$-cover of $X$, then we can extend $P_\delta|_Q$ to a cover of $X$ that takes each point of $Q$ to the $\e$-offset of its $\delta$-neighborhood in $P$ (see Figure~\ref{fig:subsample_cover}):
\begin{align*}
  B^\e_{{P_\delta |_Q}}\defined \big(P\cap \ball^\delta(-)\big)^\e : Q &\longrightarrow \Op_X\\
    q &\longmapsto \bigcup \big\{  \ball^\e(p)\mid p\in P\cap \ball^\delta(q)\big\}.
\end{align*}

\begin{figure}[ht]
  \centering
  \includegraphics[width=0.35\textwidth]{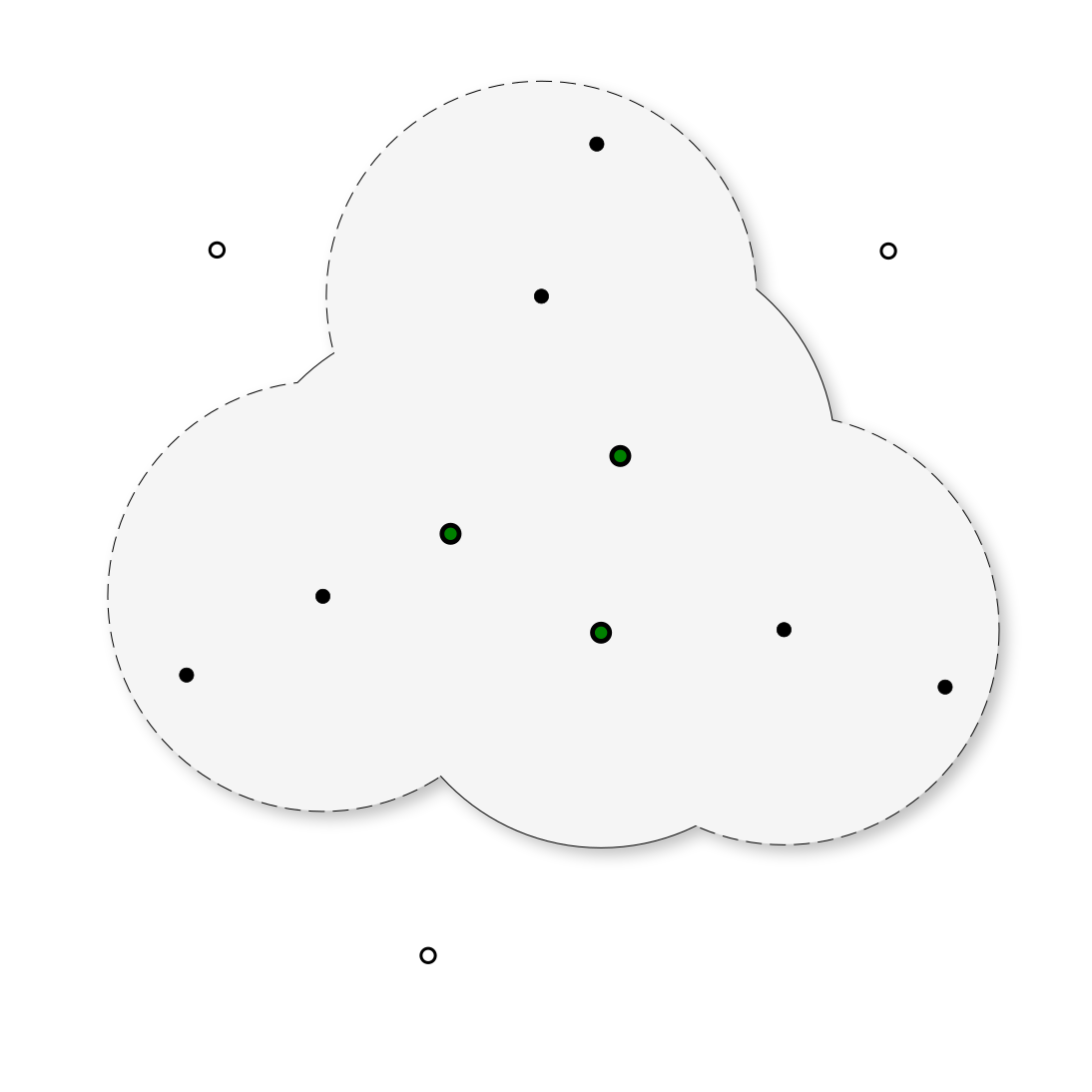}
  \hspace{10ex}
  \includegraphics[width=0.35\textwidth]{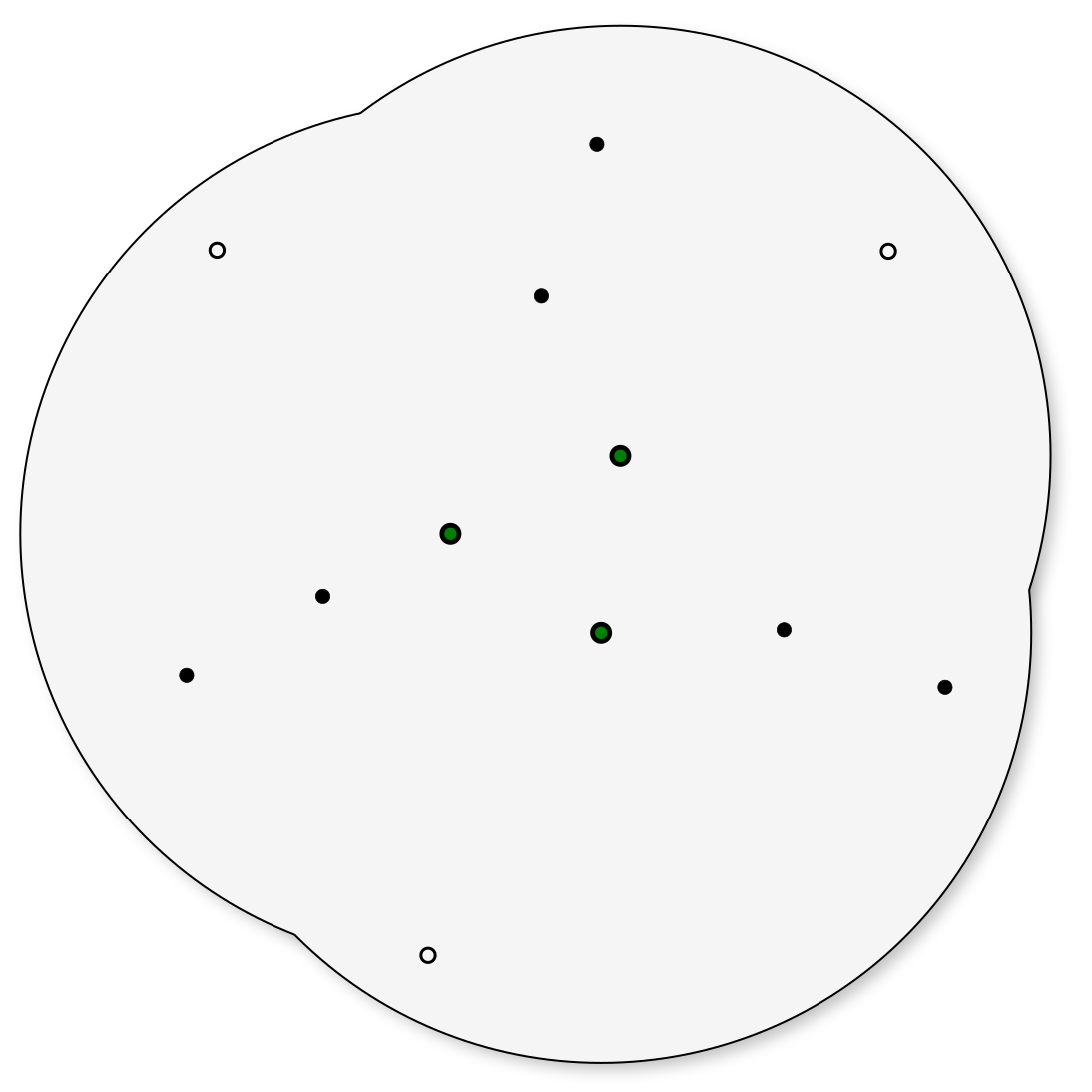}
  \caption{Although $P_{4,\e}$ is a $2\e$-sample of $X$ the cover $B^\e_{P_{4,\e}}$ (left) is finer than $B_{P_{4,\e}}^{2\e}$ (right).}\label{fig:subsample_cover}
\end{figure}

\begin{definition}[Density-Based Clustering]
  Let $P\subset X$ be a finite subspace and assume that $\# P_\delta \geq k$ for some $k,\delta > 0$.
  Then the set $P_{k,\delta} \defined \{p\in P\mid \# P_\delta(p)\geq k\}$ of \defn{core points} of $P$ is a $\delta$-sample of $P$ and,
  as above, $B^\e_{P_{k,\delta}} : P_{k,\delta}\to \Op_X$ is a cover of $X$.
  Composition with the canonical clustering functor $\pi_X^P : \Op_X\to \Set$ yields a functor $\pi_X^P B_{P_{k,\delta}}^\e : P_{k,\delta}\to \Set$ that takes each core point to the set containing the corresponding connected component.
  Thus, the desired clustering of $P$ is given by the join
  \[ \mathsf{DBS}_P(k,\delta,\e)\defined \Join \pi_X^P B_{P_{k,\delta}}^\e.\]
\end{definition}

\subsection{Merge Trees and Topological Hierarchical Decompositions}\label{sec:hierarchical}

In the previous examples, the clustering of the input points is parameterized by the distance between them,
and the correctness of the clustering depends on the assumption that our sample covers the underlying space at a known scale.
In practice, we often do not know at what scale a sample covers the underlying space%
\footnote{%
    There is a topological criterion for coverage introduced by De Silva and Ghrist~\cite{desilva07coverage,cavanna17when}
    that can be efficiently computed from the \v{C}ech complex under mild sampling conditions.
}.
This is the primary motivation for hierarchical clustering, which has two main variations:
\begin{enumerate}
    \item \emph{Agglomerative clustering} starts with a fine clustering of points and merges them,
        often starting with the connected components of the \emph{initial topology}: the finest topology in which each point belongs to its own cluster.
    \item \emph{Divisive clustering} starts with a coarse clustering of points and divides them,
        often starting with the connected components of the \emph{final topology} in which all points are in the same connected component.
\end{enumerate}

In Section~\ref{sec:multi}, we provide a novel example of divisive hierarchical clustering~\cite{dey16multiscale}.
In this section, we will focus on agglomerative clustering.
Specifically, we formalize single and complete linkage clustering and show how the associated functions bound the Hausdorff distance~\cite{basalto08hausdorff}.
Motivated by the examples of the previous section, we will regard the data of a hierarchical clustering scheme as a filtration of topological spaces, defined formally as follows.

\begin{definition}[Filtration of Spaces]
    A \defn{filtration of spaces} is a functor $\F : \P\to \Top$ from a poset $\P$ to the category of topological spaces.
    That is, a collection of topological spaces $\F(p)$ indexed by elements $p\in\P$ with continuous functions $\F\b{p\leq q} : \F(p) \to \F(q)$ for all $p\leq q$.
\end{definition}

\begin{example}[Metric Clustering]
    Let $P\subset X$ be a finite subspace of a metric space $(X,\dist)$ and let $B_P :\cat{R}_+\to \Top$ be the filtration of spaces defined
    \[ B_P(\e) \defined \begin{cases}  P &\text{ if } \e = 0,\\ P^\e &\text{ otherwise.}\end{cases}\]
    Composition with the connected components functor yields a filtration of \emph{sets} $\pi_0 B_P : \cat{R}_+\to\Set$
    that not only takes each $\e > 0$ to a clustering $\pi_0 B_P(\e) = \pi_0(P^\e)$ of the $\e$-neighborhood of the sample points,
    but also comes equipped with transition maps $\pi_0 B_P\b{\e < \delta} : \pi_0(P^\e)\to \pi_0(P^\delta)$.
    Using this information, we can calculate the corresponding \emph{merge tree}, recently generalized to filtrations of spaces by Curry et al.~\cite{curry21decorated},
    as a poset with elements $\beta_p\defined (p, \beta)$ for each $p\in\P$ and each connected component $\beta\in \pi_0(\F(p))$ (see Appendix~\ref{sec:cats2}).
    This construction serves as a prototypical example of a THD: a dendrogram enriched with the topological structure provided by the filtration of spaces $\F$ (Figure~\ref{fig:clustering_merge}).
\end{example}

\begin{definition}[Generalized Merge Tree~\cite{curry21decorated}]
    The \defn{merge tree} of a filtration of spaces $\F : \P\to \Top$ is the \emph{category of elements} (see Appendix~\ref{sec:slice})
    of the composition $\pi_0 \F : \P\to \Set$:
    \[ \cT_\F\defined \El(\pi_0 \F).\]
    Vertices $\beta_p\in\cT_\F$ of the merge tree correspond to connected components $\beta\in\pi_0 \F(p)$,
    and edges $(\alpha_p,\beta_q)$ of the tree correspond to pairs with $p\leq q$ and $\F\b{p\leq q}(\alpha) = \beta$.
    Let $\sT_\F$ denote the topological space given by endowing the merge tree of $\F$ with the \emph{co}specialization topology.
\end{definition}

In this work, all of the THDs we encounter will arise as the merge tree of a filtration of spaces.
This includes multiscale mapper~\cite{dey16multiscale} which takes as input a filtration of \emph{covers},
but produces a THD that is equivalent to the merge tree of a filtration of nerves (see Section~\ref{sec:multi}, Proposition~\ref{prop:multi_thd}).
Formally, for the purposes of this work, the \defn{Topological Hierarchical Decomposition (THD)} of a filtration of spaces $\F : \P\to \Top$ is defined as the generalized merge tree $\cT_\F = \El(\pi_0\F)$.

\begin{remark}
    Providing a more general characterization of hierarchical decompositions that arise from topological spaces,
    including their relationship with the total locale of a cosheaf~\cite{funk95display} (see also Appendix~\ref{sec:display}),
    persistent homology, and sheaf theory is the subject of future work (see also Appendix~\ref{sec:reeb_future}).
\end{remark}

\begin{example}[Metric Clustering cont.]
    The THD associated with the filtration of spaces $B_P$ defined above can be thought of as a dendrogram that is parameterized continuously by $\cat{R}_+$.
    That is, $\cT_{B_P}$ is a poset with vertices $\alpha_\e = (\e,\alpha)\in \El(\pi_0 B_P)$ for all $\e > 0$ and each connected component $\alpha\in\pi_0(P^\e)$,
    and edges $\alpha_\e\preceq \beta_\delta$ for all $\delta > \e > 0$ and $\beta\in \pi_0 B_P(\delta)$ such that $\pi_0 B_P\b{\e\leq\delta}(\alpha) = \beta$.
\end{example}

We proceed to define clustering linkage and show how it can be viewed as a reparameterization of the filtration $B_P$ according to a function on subspaces.
Importantly, the functions associated with single and complete linkage clustering do not constitute metrics on subspaces of a metric space,
but they do provide lower and upper bounds on the Hausdorff distance, defined formally as follows.

\begin{figure}[ht]
    \includegraphics[width=\textwidth]{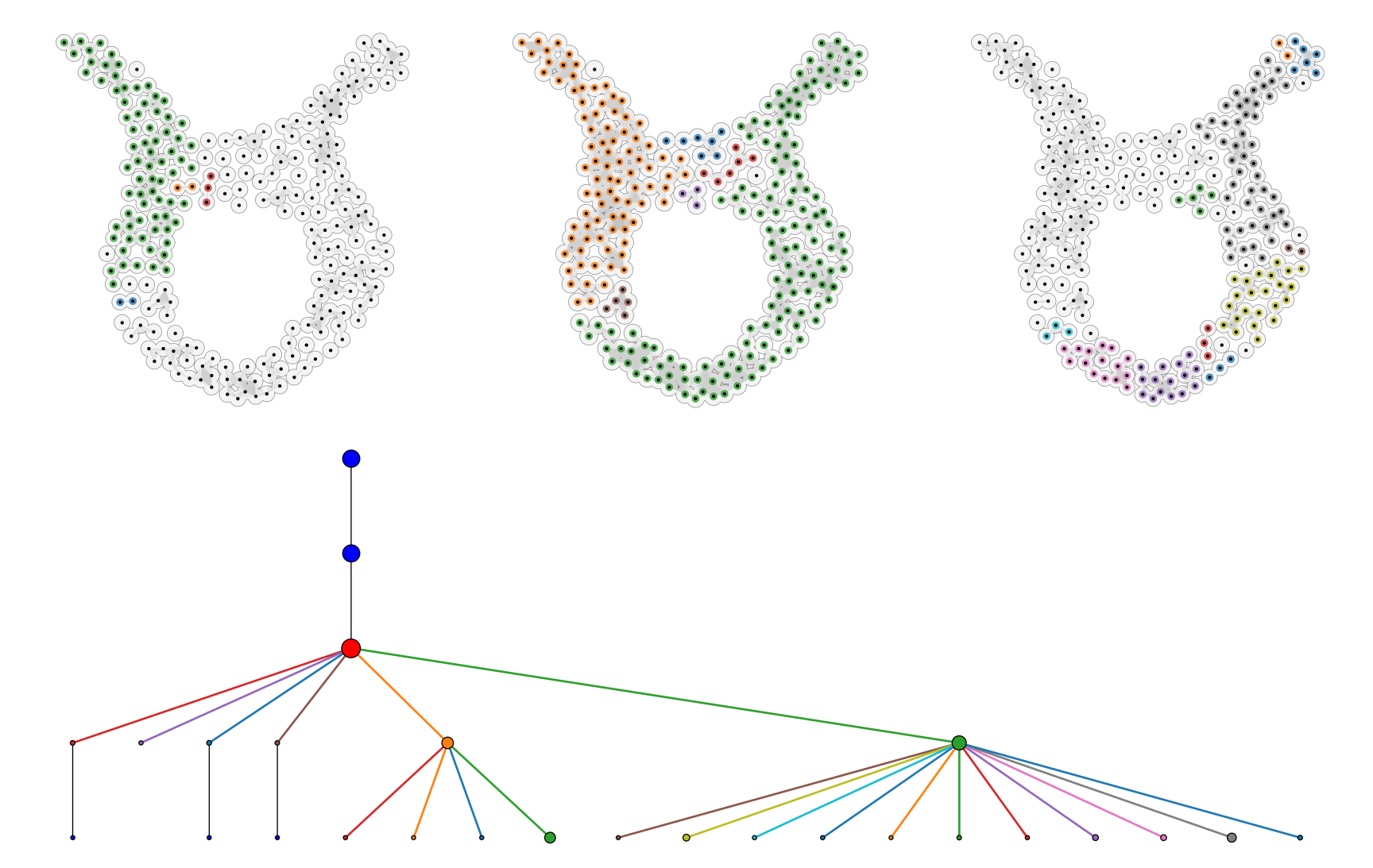}
    \caption{Our prototypical THD is the merge tree of a filtration of spaces.
              Here, each space is generated by a metric cover indexed by the same set of sample points.
              As the coverage radius is increased (left, right into middle),
              clusters corresponding to vertices of the THD merge,
              which is represented by an edge in the THD.}\label{fig:clustering_merge}
\end{figure}

Let $S\subset X$ be a subspace.
For any $x\in X$ the restriction of $\dist_x$ to the subspace $S$ is a map $\dist_x |_S : S\to \pro{\RR_+}$.
The meet or \emph{infimum} of this restriction defines the distance to $S$:
\begin{align*}
    \dist_S : X &\longrightarrow \pro{\RR_+}\\
    x & \longmapsto \Meet\dist_x |_S = \inf_{s\in S}\dist(x,s),
\end{align*}
and has a pre-image $\dist_S^{-1} : \Op_{\pro{\RR_+}}\to \Op_X$ taking each basic open set $\odown{\e}$ of $\pro{\RR_+}$ to the $\e$-offset of $S$:
\[ \dist_S^{-1}(\odown{\e}) = \{x\in X\mid \dist_S(x) < \e\} = S^\e.\]

\begin{definition}[Hausdorff Distance]
    The \defn{partial hausdorff distance to $S$} is the join or \emph{supremum} of the restriction of $\dist_S$ to a subspace $T\subset X$:
    \[ \Join \dist_S|_T = \sup_{t\in T}\inf_{s\in S} \dist(s,t).\]
    In general, the partial hausdorff distances between two subspaces are not equal,
    and the \defn{hausdorff distance} is the maximum:
    \[ \dist_\mathrm{H}(S,T)\defined \max\Big\{\Join \dist_S|_T,\Join \dist_T|_S\Big\}.\]
\end{definition}

\begin{figure}[ht]
    \centering
    \includegraphics[width=0.75\textwidth]{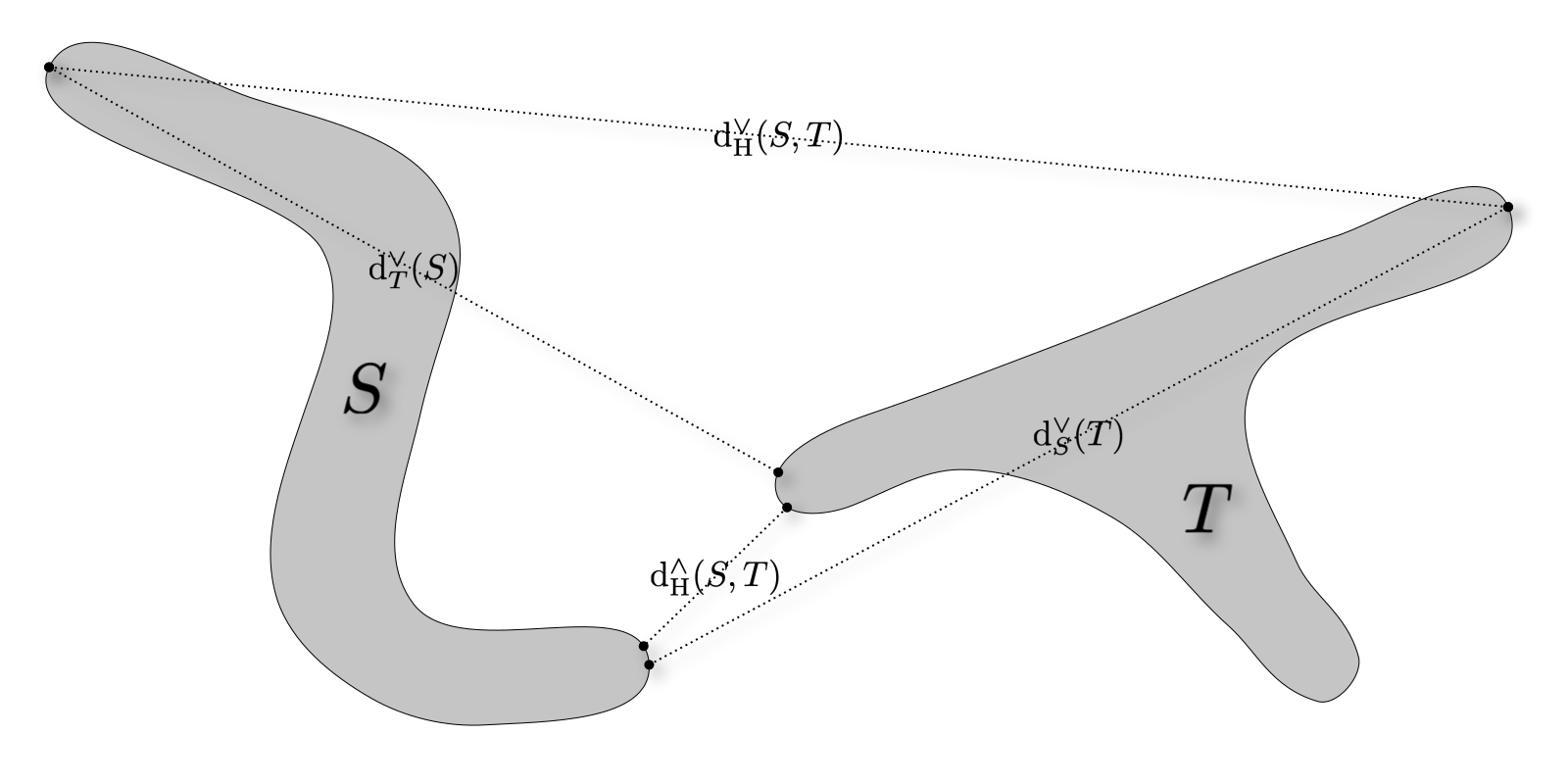}
    \caption{The minimum and maximum distance from $T$ to the closest point in $S$
        correspond to the partial hausdorff distance to $S$ and the function used in single linkage clustering, respectively.
        The minimum and maximum distance from $T$ to the \emph{farthest} point in $S$
        correspond to the function used in complete linkage clustering and the Jaccard distance, respectively.}\label{fig:set_min_distance}
\end{figure}

The function associated with single linkage clustering is the smallest distance between the two subspaces:
\[ \pro{\dist_\mathrm{H}}(S,T) \defined \Meet \dist_S|_T = \Meet \dist_T|_S = \inf_{s\in S,t\in T}\dist(s,t).\] 
This function is not a metric on subspaces of $X$, but it is a lower bound of the hausdorff distance.
On the other hand, the function associated with complete linkage clustering is defined as the largest distance between subspaces,
and provides an upper bound on the hausdorff distance:
\[ \ind{\dist_\mathrm{H}}(S,T) \defined \sup_{s\in S,t\in T}\dist(s,t),\]
Thus, we have the following sequence of inequalities:
\[ \pro{\dist_{\mathrm{H}}}\leq \dist_{\mathrm{H}} \leq \ind{\dist_\mathrm{H}}.\]

\begin{definitions}[Single and Complete Linkage THDs]
    Let $P\subset X$ be a finite subspace.
    \begin{enumerate}
        \item Let $\pro{\L_P} : \NN\to \RR_+$ be the function defined inductively for all $n\in \NN$ as
            \[ \pro{\L_P}(n) \defined \min\big\{\pro{\dist_{\mathrm{H}}}(\alpha,\beta)\mid \alpha\neq \beta \in \pi_0 B_P(n-1)\big\}.\]
            The \defn{single linkage THD} of $B_P$ is given by the merge tree $\pro{\sT_P}\defined \El(\pi_0\pro{B_P})$ of the filtration
            \begin{align*}
                \pro{B_P}\defined B_P\circ \pro{\L_P} : \NN&\longrightarrow \Top\\
                    n &\longmapsto P^{\pro{\L_P}(n)/2} = \bigcup_{p\in P} \ball^{\pro{\L_P}(n)/2}(x).
            \end{align*}
        \item Let $\ind{\L_P} : \NN\to \RR_+$ be the function defined inductively for all $n\in \NN$ as
            \[ \ind{\L_P}(n) \defined \min\big\{\ind{\dist_{\mathrm{H}}}(\alpha,\beta)\mid \alpha\neq \beta \in \pi_0 B_P(n-1)\big\}.\]
            The \defn{complete linkage THD} of $B_P$ is given by the merge tree $\ind{\sT_P}\defined \El(\pi_0\ind{B_P})$ of the filtration
            \begin{align*}
                \ind{B_P}\defined B_P\circ \ind{\L_P} : \NN&\longrightarrow \Top\\
                    n &\longmapsto P^{\ind{\L_P}(n)/2} = \bigcup_{p\in P} \ball^{\ind{\L_P}(n)/2}(x).
            \end{align*}
    \end{enumerate}
\end{definitions}

\begin{remark}
    This approach merges \emph{all} clusters within the minimum distance.
    Breaking ties amounts to modeling the topology locally around the subspaces in question,
    and is the subject of future work.
\end{remark}

\section{Reeb Graphs and the Mapper Functor}\label{sec:mapper}

In this section we give a generalized construction of the mapper functor of Brown et al.~\cite{brown21probabilistic}.
We will follow the approach of Botnan et al.~\cite{botnan20relative} alluded to in~\cite{brown21probabilistic} in which the mapper functor
is defined as the pixelization of a (pre)cosheaf by a cover $\U : I\to \Op_X$;
that is, the pullback of the pushforward of the canonical map $\eta : X\to \N_\U$.

We begin by reviewing the adjoint Reeb and display space functors in Section~\ref{sec:reeb}.
We then show that the mapper functor $\M_\U : \PCSh(X)\to \CShsp(X)$ is equivalent to the pixelization~\cite{botnan20relative} of a precosheaf on $X$ in Section~\ref{sec:mapper_fun}.
We conclude with some examples of how these results can be used in practice in Section~\ref{sec:mapper_examples}.

\subsection{The Reeb and Display Space Functors}\label{sec:reeb}

In recent work, De Silva et al.~\cite{desilva16categorified} (re)introduced the \emph{Reeb cosheaf} $R_f : \Op_X\to \Set$ associated with a locally connected space $(L,f)$ over $X$
in order to categorify Reeb graphs.
In this section, following earlier work by Funk~\cite{funk95display} and Woolf~\cite{woolf08fundamental}, as well as the more recent work by Brown et al.~\cite{brown21probabilistic},
we review the display space construction that provides a right adjoint to the restriction of the Reeb functor to locally connected spaces.
For a more combinatorial treatment of the Reeb cosheaf and its associated display space see~\cite{desilva16categorified}.

Let $F$ be a precosheaf on $X$ and let $x\in X$.
Letting $F^x\defined F |_{\Op_X(x)}$ denote the restriction of $F$ to the neighborhoods of $x$,
the \defn{costalk}\footnote{usually denoted $F_x\defined \lim F^x$ in the literature} of $F$ at $x$ is the limit
\[ \lim F^x = \lim_{U\ni x} F(U) = \Big\{\{\alpha_U\}\in\prod_{U\ni x} F(U)\ \big|\ F\b{U'\subseteq U}(\alpha_{U'}) = \alpha_U\text{ for all } U'\subseteq U\Big\}.\]
We topologize the costalks with the initial topology associated with the system of canonical projections
$\lim F^x\to \sE_X(F)$ to the \emph{total locale} $\sE_X(F)$ given by endowing the category of elements $\El(F)$ of $F$
with the \emph{co}specialization topology~\cite{funk95display} (see also Appendix~\ref{sec:display}).

\begin{definition}[Display Space]
  The \defn{display space} of $F$ is the coproduct of costalks in \Top:
  \[ \dis(F) = \coprod_{x\in X} \dis_x(F).\]
  The points of $\dis(F)$ are pairs $(\alpha,x)$ for each $x\in X$ and $\alpha\in \dis_x(F)$,
  and the basic open sets of $\dis(F)$ are defined for each open set $U\in\Op_X(x)$ and $\beta\in F(U)$ as
  \[ \cev{\beta}_U \defined \bigsqcup_{x\in U}\cev{\beta}_U^x = \big\{(\alpha,x)\mid x\in U\text{ and }\rho_U^x(\alpha) = \beta\big\}.\]
\end{definition}

\begin{definition}[Display Space Functor]
  Let $\gamma_F : \dis(F)\to X$ denote the continuous projection $(\alpha,x)\mapsto x$ with pre-image defined
  \begin{align*}
    \gamma_F^{-1}(U)
      = \bigsqcup_{x\in U} \dis_x(F)
      =\! \bigcup_{\beta\in F(U)}\!\cev{\beta}_U.
  \end{align*}
  The \defn{display space functor} $\gamma : \PCSh(X)\to\Top/X$ takes each precosheaf $F$ on $X$ to the space $\big(\dis(F),\gamma_F\big)$ over $X$:
  \begin{align*}
    \gamma : \PCSh(X)&\longrightarrow \Top/X\\
    F &\longmapsto \big(\dis(F),\gamma_F\big).
  \end{align*}
\end{definition}

\begin{definition}[Spatial Cosheaf~\cite{funk95display,woolf08fundamental}]
  A precosheaf is \defn{spatial} if the basic open sets of its display space are non-empty and connected.
\end{definition}

Importantly, every spatial precosheaf is a cosheaf~\cite{funk95display},
so the restriction of $\PCSh(X)$ to spatial cosheaves is a full subcategory $\CShsp(X)$ of cosheaves $\CSh(X)$ on $X$.
Moreover, because the display space of any spatial cosheaf is generated by connected open sets,
$\gamma$ takes spatial cosheaves to \emph{locally connected} spaces over $X$.
As a result, the composite $\lambda\gamma : \PCSh(X)\to \CShsp(X)$ is a ``spatial cosheafification functor''~\cite{woolf08fundamental}
that takes any precosheaf $F : \Op_X\to \Set$ to the Reeb cosheaf of the projection $\gamma_F : \dis(F)\to X$:
a spatial cosheaf by the following standard result.

\begin{proposition}[Proposition 5.14~\cite{funk95display}]\label{prop:reeb_spatial}
  Let $(L,f)$ be a locally connected space over $X$.
  Then the Reeb cosheaf $R_f$ is spatial and its display space $\dis(R_f)$ is locally connected.
\end{proposition}

This behavior is formalized by the following adjunction.

\begin{theorem}[Funk~\cite{funk95display} Theorem 5.9 (see also Theorem 3.1) and Woolf~\cite{woolf08fundamental} Proposition B.2]\label{thm:dis_adjoint}
  The restriction of the Reeb functor to locally connected spaces is left adjoint to the display space functor:
  \[ \adjoint{\Toplc/X}{\CShsp(X).}["\lambda"]["\gamma"]\]
\end{theorem}

\begin{proposition}[Woolf~\cite{woolf08fundamental} Proposition B.4. (see also Funk~\cite{funk95display} Definition 5.12)]\label{prop:spatial_iso}
  A precosheaf is spatial if and only if the counit of $\lambda\dashv\gamma$ yields a natural isomorphism.
\end{proposition}

\begin{example}[Merge Trees and Reeb Graphs]
  Let $(L,f)$ be a locally connected space over $X$.
  Regarding the pre-image of $f$ as a filtration $f^{-1}: \Op_X\to \Toplc$ of subspaces $f^{-1}(U)\subset L$,
  we can extend the earlier definition of merge trees to locally connected spaces over $X$ as
  \[ \cT_f\defined \cT_{R_f} = \El\big(\lambda(L,f)\big) = \El\big(\pi_L f^{-1}\big).\]
  The \defn{Reeb graph} of $(L,f)$ is the locally connected space over $X$ given by the unit of the adjunction $\lambda\dashv \gamma$:
  \[ \gamma(R_f) = \gamma\lambda(L,f) = \big(\dis(R_f), \gamma_f\big).\]
  Equivalently, as in the original work by Funk~\cite{funk95display}, the Reeb graph can be described by the pullback
  of the unique maps from $X$ and the merge tree $\sT_f$ of $(L,f)$ into the total locale of the terminal precosheaf (see Appendix~\ref{sec:display}):
  \begin{equation}
    \begin{tikzcd}[sep=huge]
      \dis(R_f)
        \arrow[r, dotted, "\pi_{f}"]
        \arrow[d, dotted, "\gamma_f"]
        &\sT_f \arrow[d, dotted, "\exists!"]\\
      X\arrow[r, "\Op_X(-)"]
      & \sE_X(\id_X).
    \end{tikzcd}
  \end{equation}
\end{example}

\subsection{The Mapper Functor}\label{sec:mapper_fun}

Let $\U : I\to \Op_X$ be a locally finite good open cover of a topological space $X$
and recall the canonical map $\eta : X\to \N_\U$ takes each point $x\in X$ to the set of elements $i\in I$ corresponding to cover sets containing $x$.
In recent work Brown et al.~\cite{brown21probabilistic} define the \emph{mapper (pre)cosheaf} associated with a precosheaf $F$ and $\U$ as
\begin{align}\label{eq:brown_mapper}
  \nonumber\M_\U(F) : \Op_X &\longrightarrow \Set\\
    U &\longmapsto F\Big(\bigcup_{x\in U}\bigcap_{i\in\eta(x)} \U(i)  \Big) = F(\U_U),
\end{align}
(see Notations~\ref{notation:nerve})
and the authors assert that this formulation is equivalent to the pixelization of $F$ by $\eta$.

The goal of this section is to formalize this equivalence.
In analog with the construction of the inverse image for sheaves by passing to the corresponding \'{e}tale space,
we will do so by passing through spatial cosheafification in order to define a spatial inverse image functor as follows.

\begin{definition}[Spatial Inverse Image]
  The \defn{(spatial) inverse image functor} associated with a locally connected space $(L,f)$ over $X$ is the composition $f^*\defined \lambda f^\star\gamma : \PCSh(X)\to \CShsp(L)$
  where $f^\star : \Toplc/X\to \Toplc/L$ is the change of base functor associated with $f$ (see Appendix~\ref{sec:cats2}).
\end{definition}

\begin{definition}[$\U$-Pixelization]
  The \defn{$\U$-pixelization}~\cite{botnan20relative} of a precosheaf $F$ by $\U$ is defined as the \emph{spatial} cosheaf
  given by the pullback of the pushforward of $F$ along $\eta$:
  \[ F^\U\defined \eta^*\eta_*(F) : \Op_X\longrightarrow \Set.\]
\end{definition}

In order to show that the $\U$-pixelization of a precosheaf is equivalent to the corresponding mapper (pre)cosheaf
we will use the fact that $\eta$ is essential to make use of a result due to Funk~\cite{funk95display} which may be found in Appendix~\ref{sec:inverse} (Theorem~\ref{thm:funk_pullback}).
Specifically, Proposition~\ref{prop:eta_open} implies that $\eta^{-1} : \Op_{\N_\U}\to \Op_X$ has a left adjoint defined as the upward closure of the image:
\begin{align*}
  \eta_! : \Op_X&\longrightarrow \Op_{\N_\U}\\
  U&\longmapsto \bigcup_{x\in U}\up{\eta(x)}.
\end{align*}
It follows that the direct image $\eta_*$ has a right adjoint $\eta^\dagger$ defined as precomposition with the left adjoint $\eta_!$:
\[ \adjoint{\PCSh(X)}{\PCSh(\N_\U).}["\eta_*"]["\eta^\dagger"]\]
Following Convention~\ref{con:surjective}, Theorem~\ref{thm:poset_pullback} implies that $\M_\U(F) = \eta^\dagger\eta_*(F)$ is a spatial cosheaf for any precosheaf $F$ on $X$.
We can therefore define the mapper functor more formally as follows.

\begin{definition}[Mapper Functor]
  The \defn{$\U$-mapper functor} is defined
  \begin{align*}
    \M_\U \defined \eta^\dagger\eta_* : \PCSh(X) &\longrightarrow \CShsp(X)\\
      F &\longrightarrow \eta^\dagger\eta_*(F) = F\eta^{-1}\eta_!
  \end{align*}
\end{definition}

\begin{remark}
  The equivalence of the mapper (pre)cosheaf defined in~\cite{brown21probabilistic} (Equation~\eqref{eq:brown_mapper})
  with the $\U$-mapper functor defined above is given by the following equalities:
  \[ \U_U = \bigcup_{x\in U}\bigcap_{i\in\eta(x)} \U(i) = \eta^{-1}\Big(\bigcup_{x\in U} \up{\eta(x)} \Big) = \eta^{-1}\eta_!(U).\]
  Moreover, because every spatial (pre)cosheaf is a cosheaf~\cite{woolf08fundamental}, we have the following equivalence by the cosheaf axiom (see Appendix~\ref{sec:cosheaves2}):
  \begin{equation}\label{eq:mapper}
    \M_\U(F)(U)
      = F\eta^{-1}\Big( \bigcup_{\sigma\in\eta_!(U)} \up{\sigma}\Big)
      = F\Big(\bigcup_{\sigma\in\eta_!(U)} \U_\sigma \Big)
      \simeq \colim_{\sigma\in\eta_!(U)} F(\U_\sigma).
  \end{equation}
\end{remark}

Theorem~\ref{thm:mapper} follows directly from Corollary~\ref{cor:poset_pullback}.
We include a proof here for completeness.

\begin{theorem}\label{thm:mapper}
  Let $X$ be a locally connected topological space and let $\U : I \to \Op_X$ be a locally finite good open cover.
  Then for all precosheaves $F$ on $X$ the counit of $\lambda\dashv\gamma$ yeilds a natural isomorphism
  \[ \Lambda_F^\U : F^\U\simeq\M_\U(F).\]
\end{theorem}
\begin{proof}
  By Corollary~\ref{cor:adjoint_equiv}, $\eta^\star\gamma\eta_* \simeq \gamma\eta^\dagger\eta_* = \gamma\M_\U$,
  so $F^\U \simeq \lambda\gamma\M_\U$.
  Because $\N_\U$ is a poset endowed with the specialization topology,
  Lemma~\ref{lem:special_spatial} implies that $\eta_*(F)$ is spatial,
  so the counit of $\lambda\dashv\gamma$ at $F^\U$ yields a unique natural isomorphism
  $\Lambda_F^\U : F^\U\simeq\M_\U(F)$ as desired.
\end{proof}

\begin{remark}
  Also in the work by Brown et al.~\cite{brown21probabilistic},
  the authors introduce the \emph{enhanced mapper graph} as the display space of the mapper cosheaf $\gamma\M_\U(F)$.
  However, because $\eta$ is essential, the pullback $\eta^\star\gamma$ is equivalent to $\gamma\eta^\dagger$ by Corollary~\ref{cor:adjoint_equiv}.
  It follows that the enhanced mapper graph is equivalent to the pullback of $\gamma\eta_*$ along $\eta$:
  \[ \gamma\M_\U  = \gamma\eta^\dagger\eta_* \simeq \eta^\star\gamma \eta_*.\]
\end{remark}

\begin{example}[$\U$-Mapper Reeb Graph]
  Let $(L,f)$ be a locally connected space over $X$.
  The projection of the pullback of $\gamma_{\eta\circ f} = \gamma_{\eta_*(R_f)}$ along $\eta$ onto $X$ defines a space
  $\big(X\times_{\N_\U} \dis(R_{\eta\circ f}\big), \gamma_{\eta\circ f}^\eta)$ over $X$:
  \begin{equation}\label{dgm:reeb_pullback}
    \begin{tikzcd}[sep=huge]
      X\times_{\N_\U} \dis(R_{\eta\circ f})
        \arrow[r, dotted]
        \arrow[d, dotted,"\gamma_{\eta\circ f}^\eta"]
      & \dis(R_{\eta\circ f})\arrow[d, "\gamma_{\eta\circ f}"]
      \\
      X\arrow[r,"\eta"]
      & \N_\U
    \end{tikzcd}
  \end{equation}
  The merge tree of this space is the total locale of the pixelization $R_f^\U\simeq \M_\U(R_f)$ of the Reeb cosheaf $R_f$ by $\U$:
  \[ \sT_f^\U\defined \sE_X(R_f^\U).\]
  The points of $\sT_f^\U$ are pairs $\beta_U = (U,\beta)$ for each $U\in\Op_X$ and each connected component $\beta\in R_f^\U(U)$.
  That is, $\beta$ is a connected component of $f^{-1}(\U_U)$ where $\U_U = \eta^{-1}\eta_!(U) \supseteq U$.
\end{example}

\subsection{Mapper in Practice}\label{sec:mapper_examples}

Traditionally, the mapper algorithm takes as input
\begin{enumerate*}
  \item a cover $\U$ of a topological space $X$,
  \item a filter function $f : P\to X$ from a \emph{set} $P$ of data points, and
  \item a \emph{clustering} of the data.
\end{enumerate*}
The result is the $\U$-mapper Reeb graph: a space over $X$ that has been ``pixelized''~\cite{botnan20relative} by the cover $\U$.
However, in general, even for locally finite covers $\U : I\to \Op_X$, the pullback $\dis^\U(F) = X\times_{\N_\U} \dis\big(\eta_*(F)\big)$ is not finite, nor does it have finitely generated topology.

On the other hand, the intermediate space $\dis\big(\eta_*(F)\big)$ has finitely generated topology and is finite for finite covers.
Moreover, because $\N_\U$ is endowed with the specialization topology,
Proposition~\ref{prop:rho_bi} implies that the display space $\dis\big(\eta_*(F)\big)$ is homeomorphic to the total locale of $\eta_*(F)$;
in particular, for any locally connected space $(L,f)$ over $X$,
this implies that the display space of $\eta_*(R_f) = \lambda(L, \eta\circ f) = R_{\eta\circ f}$ is homeomorphic to the merge tree of $(L,\eta\circ f)$:
\[ \dis\big(\eta_*(R_f)\big) = \dis\big( \lambda(L,\eta\circ f) \big) \cong \sT_{\eta\circ f}.\]

Let $(L,f)$ be a locally connected space over $X$ and let $\U: I \to\Op_X$ be a locally finite good open cover.
Then $f^{-1}\U : I\to \Op_L$ is a cover of $L$ and $R_f\U : I\to\Set$ is a function taking each $i\in I$ to the connected components of the pullback $f^{-1}\big(\U(i)\big)$.
Because $\U$ is discretely parameterized, the merge tree of $f^{-1}\U$ (regarded as a filtration of spaces) is equivalent to the coproduct $\bigsqcup R_f\U$ and,
because $L$ is locally connected, each connected component $\beta\in R_f\big(\U(i)\big)$ is a connected open set in $L$.
Thus, the projection $(i,\beta)\mapsto \beta$ defines a cover $\U^f : \bigsqcup R_f\U\to \Op_L$ of $L$.
The nerve of $\U^f$ is the simplicial complex traditionally associated with the output of the mapper algorithm:
\[ \N_\U^f = \Big\{B\subseteq \bigsqcup R_f\U \mid \bigcap_{(\beta,i)\in B} \beta \neq \emptyset\Big\}.\]
Proof of Proposition~\ref{prop:display_nerve_iso} can be found in Appendix~\ref{sec:proofs}.
\begin{proposition}\label{prop:display_nerve_iso}
  The nerve of $\U^f$ is isomorphic to the display space of $R_{\eta\circ f}$.
\end{proposition}

\begin{figure}[ht]
  \centering
  \includegraphics[width=0.9\textwidth]{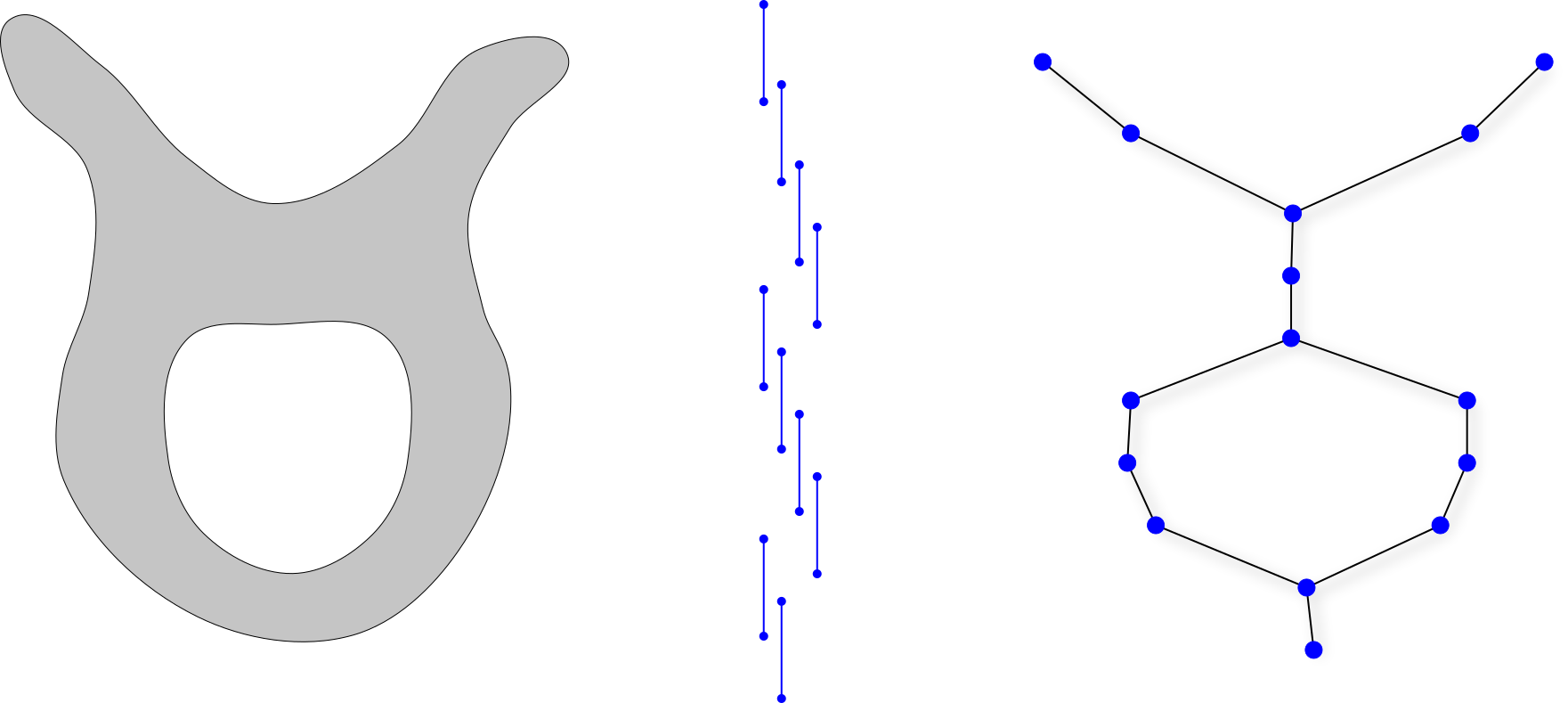}
  \caption{(Left) a locally connected space $(L,f)$ over $\RR$ taking each point to its $y$-coordinate and a cover of $\im f$ (middle).
          (Right) the nerve $\N_\U^f$ of the cover of $L$ produced by the mapper algorithm (right).}\label{fig:nerve_reeb}
\end{figure}

\begin{examples}\hfill
  \begin{enumerate}
    \item Let $\U : I\to \Op_L$ be a locally finite good open cover of a locally connected space $L$.
      Then the Reeb cosheaf of the space $(L,\eta)$ is a cosheaf on $\N_\U$ that,
      in particular, takes each principal up set $\up{\sigma}$ to the connected components of $\eta^{-1}(\up{\sigma}) = \U_\sigma$
      which, because each $\U_\sigma$ is connected, is the singleton $\pi_0(\U_\sigma) = \{\U_\sigma\}$.
      Proposition~\ref{prop:rho_bi} therefore implies that
      \[ \dis(R_\eta) = \bigsqcup_{\sigma\in\N_\U}\lim_{S\ni \sigma} \pi_L(\U_S)\simeq \bigsqcup_{\sigma\in\N_\U}\pi_L(\U_\sigma) = \{(\U_\sigma,\sigma)\mid \sigma\in \N_\U\} \simeq \N_\U, \]
      so Reeb graph of $(L,\eta)$ is isomorphic to the identity over the nerve $\N_\U$.
    \item Let $(X,\dist)$ be a metric space and let $P$ be a locally finite $\e$-sample.
      Let $B^\e : P\to \Op_X$ be the cover of $X$ defined as the collection of $\e$-balls $B^\e(p) \defined \ball^\e(p)$ centered at sample points $p\in P$.
      The canonical map $\eta : X\to \cech^\e(P)$ has a pre-image defined for basic open sets $\up{\sigma}$
      as the intersection of metric balls $\eta^{-1}(\up{\sigma}) = \bigcap_{p\in \sigma} \ball^\e(p)$.
      It follows that
      \[ B^\e_x = \bigcap_{p\in\eta(x)} \U(p) =\bigcap_{p\in P_\e(x)} \ball^\e(p)\]
      is the intersection of $\e$-balls centered at points in $p\in P$ within distance $\e$ of $x$,
      and the pixelization of an open set $U\in \Op_X$ is the union
      \[ B^\e_U = \bigcup_{x\in U}B^\e_x = \bigcup_{x\in U}\bigcap_{p\in P_\e(x)}\ball^\e(p) = \bigcup_{\sigma\in \eta_!(U)}\bigcap_{p\in\sigma}\ball^\e(p).\]
    \end{enumerate}
\end{examples}

\section{Multiscale Mapper}\label{sec:multi}

There are two essential parameters to the mapper algorithm: the choice of \emph{filter function}, and the choice of \emph{cover},
and the usefulness of the mapper algorithm depends on both~\cite{carriere18structure}.
Multiscale mapper~\cite{dey16multiscale} proposes a persistence-based solution by instead considering a filtration or \emph{tower} of covers.
The result is a divisive hierarchical clustering scheme that analyzes how the output of the mapper algorithm evolves as the cover is refined.

In this section, we adapt the results of the previous section to the multiscale setting.
We will begin by reviewing recent work by Curry et al.~\cite{curry21decorated} that provides a convenient language for filtrations of covers in Section~\ref{sec:param}.
In Section~\ref{sec:mapper_filt} we define multiscale mapper filtrations and show that the mapper construction commutes with finite limits.
We conclude with some discussion on the resulting THDs in Section~\ref{sec:multi_examples}.

\subsection{Parameterized Categories and Filtrations of Covers}\label{sec:param}

Recent work by Curry et al.~\cite{curry21decorated} makes use of \emph{parameterized categories}
to decorate persistence modules by the connected components of the underlying filtration of spaces.
We find that this formalism provides a natural language for covers of topological spaces
in which a cover $\U : I\to\Op_X$ is a \emph{parameterized object} in $\Op_X$
that corresponds to an object of the parameterized category $\mathbf{pOp}_X$.

\begin{definition}[Parameterized Category~\cite{curry21decorated}]
  For any category \cat{C} the category $\cat{pC}$ of \defn{(discretely) parameterized objects in $\cat{C}$} has functors $\X : I\to \cat{C}$ from a discrete category $I$ as objects, and morphisms $\phi : \X\Longuto \X'$ given by a reindexing map $\supp{\phi} : I\to I'$ and a natural transformation $\phi : \X\Rightarrow \X'\supp{\phi}$.
  \begin{equation}
    \begin{tikzcd}[sep=large]
      I \arrow[rr,dotted,"\supp{\phi}"description]
        \arrow[dr,"\X"description]
        \arrow[dr,phantom,""{name=f,above}]
      & & I'
        \arrow[dl,"\X'"description]
        \arrow[dl,phantom,pos=0.5,""{name=g,above}]
      \\
      & \Op_X. &
      \arrow[Rightarrow, from=f,to=g,shift right=1,shorten >=2ex,shorten <=1.5ex,"\phi"description]
    \end{tikzcd}
  \end{equation}
\end{definition}

\begin{definitions}[Filtrations of Covers]
  Let $\P$ be a poset.
  A functor $\sU : \P\to \pOp_X$ is a \defn{filtration of covers} of $X$ if $\sU(p)$ is a cover of $X$ for all $p\in\P$.
  Naturally, $\sU$ is a good (resp. locally finite) filtration of covers if the components are good (resp. locally finite) covers of $X$.
  Importantly, the nerve of a filtration of locally finite covers is a filtration of locally connected spaces $\N_\sU : \P\to \Toplc$.
\end{definitions}

\begin{notation}
  For any $p\in \P$ let $\sU^p \defined \sU(p) : I^p\to \Op_X$ and for any $q\geq p$ let $\sU^q_p\defined \sU^q\circ\supp{\sU}\b{p\leq q} : I^p\to \Op_X$ denote the composition of the reindexing map
  $\supp{\sU}\b{p\leq q} : I^p \to I^q$ with $\sU^q : I^q\to \Op_X$ so that $\sU\b{p\leq q} : \sU^p\Rightarrow \sU^q_p$.
\end{notation}

\begin{definition}[Refinement]
  A \defn{($\P$-indexed) refinement} is a $\op{\P}$-indexed filtration of covers:
  a functor $\sU : \op{\P}\to\pOp_X$ in which $\sU^q(j)\subseteq \sU^p_q(j)$ for all $q\geq p$ and $j\in I^q$.
  \begin{equation}
    \begin{tikzcd}[sep=large]
      I^q \arrow[rr,dotted,"\supp{\sU}\b{q\geq p}"description]
        \arrow[dr,"\sU^q"description]
        \arrow[dr,phantom,""{name=f,above}]
      & & I^p
        \arrow[dl,"\sU^p"description]
        \arrow[dl,phantom,pos=0.5,""{name=g,above}]
      \\
      & \Op_X. &
      \arrow[Rightarrow, from=f,to=g,shift right=1,shorten >=2ex,shorten <=1.5ex,"\sU\b{q\geq p}"description]
    \end{tikzcd}
  \end{equation}
\end{definition}

\begin{definition}[Strict Refinement]
  A refinement $\sU : \op{\P}\to \pOp_X$ is \defn{strict} if Diagram~\eqref{dgm:strict_natural} commutes for all $q\geq p$.
  That is, if $\sU$ is strict then the family of canonical maps $\eta_{\sU(p)} : X\to \N_\sU(p)$ is a cone $\eta_\sU : \Delta_{\op{\P}}(X)\Rightarrow \N_\sU$.
  \begin{equation}\label{dgm:strict_natural}
    \begin{tikzcd}[sep=huge]
      X\arrow[r,"\eta_{\sU(q)}"]\arrow[dr,"\eta_{\sU(p)}"']
      & \N_\sU(q)\arrow[d,"\N_\sU\b{q\geq p}"]\\
      & \N_\sU(p).
    \end{tikzcd}
  \end{equation}
\end{definition}

\begin{notation}
  We often drop the subscript and write $\eta \defined \eta_\sU$ and $\eta_p\defined \eta_{\sU(p)}$ when no confusion may occur.
  For convenience, let $\eta_*^p \defined (\eta_p)_*$ denote the direct image along $\eta_p$
  and let $\eta_!^p\defined (\eta_p)_!$ denote the left adjoint of $\eta_p^{-1}$.
\end{notation}

\subsection{Mapper Filtrations}\label{sec:mapper_filt}

The goal of this section is to show that the mapper functor commutes with \emph{finite} limits.
Formally, let $X$ be a topological space and let $\sU : \op{\P}\to\pOp_X$ be a strict refinement of locally finite good open covers.
The \defn{$\sU$-pixelization} of a precosheaf $F$ on $X$ is a filtration of spatial cosheaves defined
\begin{align*}
  F^\sU : \op{\P} &\longrightarrow \CShsp(X)\\
    p &\longmapsto F^{\sU(p)} = \eta^*_p\eta_*^p(F),
\end{align*}
We will show that the universal arrow $\hat{\eta} : X\to \lim\N_\U$ for which Diagram~\eqref{dgm:nerve_limit} commutes is essential and,
if $\P$ is finite, that the pixelization of a preosheaf by $\hat{\eta}$ is equivalent to the limit of $F^\sU$.

\begin{equation}\label{dgm:nerve_limit}
  \begin{tikzcd}[column sep=huge]
    & & \N_\sU(q)
      \arrow[dd, "\N_\sU\b{q\geq p}"description]\\
    X
      \arrow[r, dotted, "\exists !\hat{\eta}"description]
      \arrow[urr, bend left=15, "\eta_q"description]
      \arrow[drr, bend right=15, "\eta_p"description]
    & \lim \N_\sU
      \arrow[ur, dotted,"\rho_q"description]
      \arrow[dr, dotted,"\rho_p"description]\\
    & & \N_\sU(p)
  \end{tikzcd}
\end{equation}

The limit $\lim\N_\sU$ of the nerve of $\sU$ consists of elements $\sigma = \{\sigma^p\defined \rho_p(\sigma)\} \in \prod_{p\in\P} \N_\sU(p)$ of the product
that are \emph{consistent} in the sense that $\N_\sU\b{q\geq p}(\sigma^q) = \sigma^p$ for all $p\leq q$.
We will regard the limit as a poset in which $\sigma\subseteq \tau$ if $\sigma^p\subseteq \tau^p$ for all $p\in\P$
and endow it with the specialization topology generated by the collection of principal up-sets $\{\up{\sigma}\mid \sigma\in\lim\N_\sU\}$.

\begin{remark}
  The specialization topology on $\lim\N_\sU$ is equivalent to the initial topology associated with the system of canonical projections
  \[ \big\{ \rho_p : \lim\N_\sU\to \N_\sU(p)\big\}_{p\in\P}.\]
  Proof of this fact can be found in Appendix~\ref{sec:proofs}
\end{remark}

\begin{notation}
  For any $\sigma\in\lim\N_{\sU}$ let $\sU^p_\sigma\defined \eta_p^{-1}(\up{\sigma^p})$
  and for any $x\in X$ let $\sU_x^p \defined \eta_p^{-1}\big(\up{\eta_p(x)}\big) = (\sU^p)_{\eta_p(x)}$.
  Similarly, for any $S\in\Op_{\lim\N_\sU}$ (resp. $U\in \Op_X$) let
  \begin{align*}
    \sU^p_S&\defined \bigcup_{\sigma\in S} \sU_\sigma^p\\
    \Big(\text{resp. }\ \sU^p_U&\defined \bigcup_{x\in U} \sU_x^p\ \Big).
  \end{align*}
  Naturally, for any $\sigma\in\lim \N_\sU$ (resp. $x\in X$)
  let $\sU_\sigma : \op{\P}\to \Op_X$ (resp. $\sU_x: \op{\P}\to \Op_X$)
  be the functor defined $p\mapsto \sU_\sigma^p$ (resp. $p\mapsto\sU_x^p$),
  and for any $S\in\Op_{\lim \N_\sU}$ (resp. $U\in\Op_X$) let $\sU_S,\sU_U :\op{\P}\to \Op_X$
  be defined $p\mapsto \sU_S^p, \sU_U^p$, respectively.
\end{notation}

\begin{lemma}\label{lem:eta_multi_cts}
  If $\sU$ is strict refinement of locally finite good open covers then $\hat{\eta}$ is essential.
\end{lemma}
\begin{proof}
  We will begin by showing that $\hat{\eta}^{-1}(\up{\sigma}) = \Meet \sU_\sigma$ for all $\sigma\in\lim \N_\sU$.

  Let $\sigma\in \lim\N_\sU$.
  If $x\in \hat{\eta}^{-1}(\up{\sigma})$ then $\hat{\eta}(x) \in\, \up{\sigma}$ implies that $\hat{\eta}(x) \supseteq \sigma$,
  so $\eta^p(x) \supseteq \sigma^p$ implies $x\in\sU_x^p\subseteq \sU_\sigma^p$ for all $p\in \P$, so $x\in \Meet \sU_\sigma$.
  Conversely, if $x\in\Meet\sU_\sigma$ then $x\in \sU_\sigma^p$ for all $p\in \P$.
  Because $x\in \sU_\sigma^p$ implies $x\in \sU^p(i)$ for all $i\in I^p$,
  it follows that $\eta_p(x) \supseteq \sigma^p$ for all $p\in\P$ so $\hat{\eta}(x)\supseteq \sigma$ implies $x\in\hat{\eta}^{-1}(\up{\sigma})$.

  Let $S\in\Op_{\lim\N_\sU}$.
  Because $\Op_{\lim\N_\sU}$ is generated by principal up-sets $S$ is of the form $S = \bigcup_{\sigma\in S}\up{\sigma}$.
  It follows that
  \[ \hat{\eta}^{-1}(S)
    = \bigcup_{\sigma\in S} \hat{\eta}^{-1}(\up{\sigma})
    = \bigcup_{\sigma\in S} \bigcap_{t\in\P}\eta_t^{-1}(\up{\sigma^t})
    = \bigcup_{\sigma\in S} \bigcap_{t\in \P} \sU_\sigma^t.\]
  Because $\P$ is finite and $S$ is directed,
  \[\hat{\eta}^{-1}(S) = \bigcup_{\sigma\in S} \bigcap_{t\in \P} \sU_\sigma^t = \bigcap_{t\in \P} \bigcup_{\sigma\in S}  \sU_\sigma^t = \Meet \sU_S,\]
  is a finite intersection of open sets, and is therefore open in $X$, so $\hat{\eta}$ is continuous as desired.
  It remains to show that $\hat{\eta}$ is essential, which requires showing that $\hat{\eta}^{-1}$ has a left adjoint
  $\hat{\eta}_!$ such that $U \subseteq \hat{\eta}^{-1}\hat{\eta}_!(U)$ for each open set $U\in\Op_X$
  and $\hat{\eta}_!\hat{\eta}^{-1}(\up{\sigma}) \subseteq \up{\sigma}$ for each $\sigma\in \lim\N_\sU$.

  Let $\hat{\eta}_! : \Op_X\to \Op_{\lim \N_\sU}$ be defined for all $U\in \Op_X$ as the union of principal up sets
  \[ \hat{\eta}_!(U)\defined \bigcup_{x\in U} \up{\hat{\eta}}(x).\]
  Then for all $U\in\Op_X$ we have
  \[ \hat{\eta}^{-1}\hat{\eta}_!(U) = \Meet \sU_{\hat{\eta}_!(U)} = \bigcap_{p\in \P} \bigcup_{\sigma^p\in \eta_!^p(U)} \sU_\sigma^p.\]
  Because $\eta_!^p$ is essential for all $p\in \P$, $U\subseteq \sU_\sigma^p$ for all $\sigma^p\in \eta_!^p(U)$, so $U \subseteq \hat{\eta}^{-1}\hat{\eta}_!(U)$.

  For any $\sigma\in \lim \N_\U$ we have
  \[ \hat{\eta}_!\hat{\eta}^{-1}(\up{\sigma}) = \bigcup_{x\in \hat{\eta}^{-1}(\up{\sigma})} \up{\hat{\eta}(x)}.\]
  If $x\in \hat{\eta}^{-1}(\up{\sigma})$ then $x\in \sU^p(i)$ for all $p\in\P$ and $i\in\rho_p(\sigma)$.
  So $\hat{\eta}(x)\supseteq \sigma$ implies $\hat{\eta}(x)\in\,\up{\sigma}$ for all $x\in\hat{\eta}^{-1}(\up{\sigma})$,
  thus $\hat{\eta}_! \hat{\eta}^{-1}(\up{\sigma}) \subseteq \up{\sigma}$.
\end{proof}

By Theorem~\ref{thm:poset_pullback}, Lemma~\ref{lem:eta_multi_cts} implies that $\hat{\eta}^\dagger\hat{\eta}_*(F)$ is a spatial cosheaf for any precosheaf $F$ on $X$.
We can therefore define the \emph{projective} multiscale mapper functor as follows.

\begin{definition}[$\sU$-Multiscale Mapper Functor]
  The \defn{\emph{projective} $\sU$-multiscale mapper functor} is defined
  \begin{align*}
    \pro{\sM_\sU}\defined \hat{\eta}^\dagger\hat{\eta}_* : \PCSh(X) &\longrightarrow \CShsp(X)\\
      F &\longmapsto \hat{\eta}^\dagger\hat{\eta}_*(F).
  \end{align*}
  $\pro{\sM_\sU}$ takes each precosheaf $F$ on $X$ to the \defn{\emph{projective} $\sU$-multiscale mapper precosheaf} defined
  \begin{align*}
    \pro{\sM_\sU}(F) : \Op_{X}&\longrightarrow \Set\\
    U &\longmapsto F\Big( \Meet \sU_U \Big).
  \end{align*}
\end{definition}

\begin{theorem}\label{thm:multi_iso}
  Let $X$ be a locally connected topological space and let $\sU : \op{\P}\to\pOp_X$ be a strict refinement of locally finite good open covers of $X$.
  If $\P$ is finite then for any precosheaf $F$ on $X$
  \[\pro{\sM_\sU}(F)\simeq \lim F^\sU.\]
\end{theorem}
\begin{proof}
  Because right adjoint preserve limits and $\eta^\dagger_p$ and $\eta_*^p$ are right adjoints for all $p\in \P$,
  \[ \pro{\sM_\sU} = \hat{\eta}^\dagger\hat{\eta}_* \simeq \lim_{p\in\P} \eta^\dagger_p\eta_*^p.\]
  Because $X$ is locally connected, $\lim \N_\sU$ is a poset endowed with the specialization topology.
  Therefore, because $\hat{\eta}$ is surjective and essential by Lemma~\ref{lem:eta_multi_cts},
  the desired result follows from Corollary~\ref{cor:poset_pullback}.
\end{proof}

\subsection{Multiscale Mapper in Practice}\label{sec:multi_examples}

Let $(L,f)$ be a locally connected space over $X$ and let $\sU : \op{\P}\to \pOp_X$ be a strict refinement of locally finite good open covers of $X$.
The goal of this section is to extend the results of Section~\ref{sec:mapper_examples} to the multiscale setting.
This will be done by showing how the results of the previous section can be used to compute a THD that corresponds to a filtration of pixelized spaces.

Intuitively, the desired THD is the category of elements of the filtration of sets
\begin{align*}
  \hat{R}^\sU_f : \op{\P}\times \lim \N_\sU&\longrightarrow \Set\\
    (p,\sigma) &\longmapsto R_f\big(\sU_\sigma^p\big);
\end{align*}
That is, the poset $\El(\hat{R}^\sU_f)$ forms a tree with vertices $\big((p,\sigma),\beta_p \big)$ for all $p\in \P$ and $\sigma\in\lim\N_\sU$, and each connected component $\beta_p\in R_f\big(\sU_\sigma^p\big)$,
and edges induced by the functoriality of the Reeb cosheaf (Figure~\ref{fig:multi_mapper_thd}).
The following proposition states that this THD is equivalent to the merge tree of a filtration of spaces that can be computed in practice.

\begin{figure}[ht]
  \includegraphics[width=\textwidth]{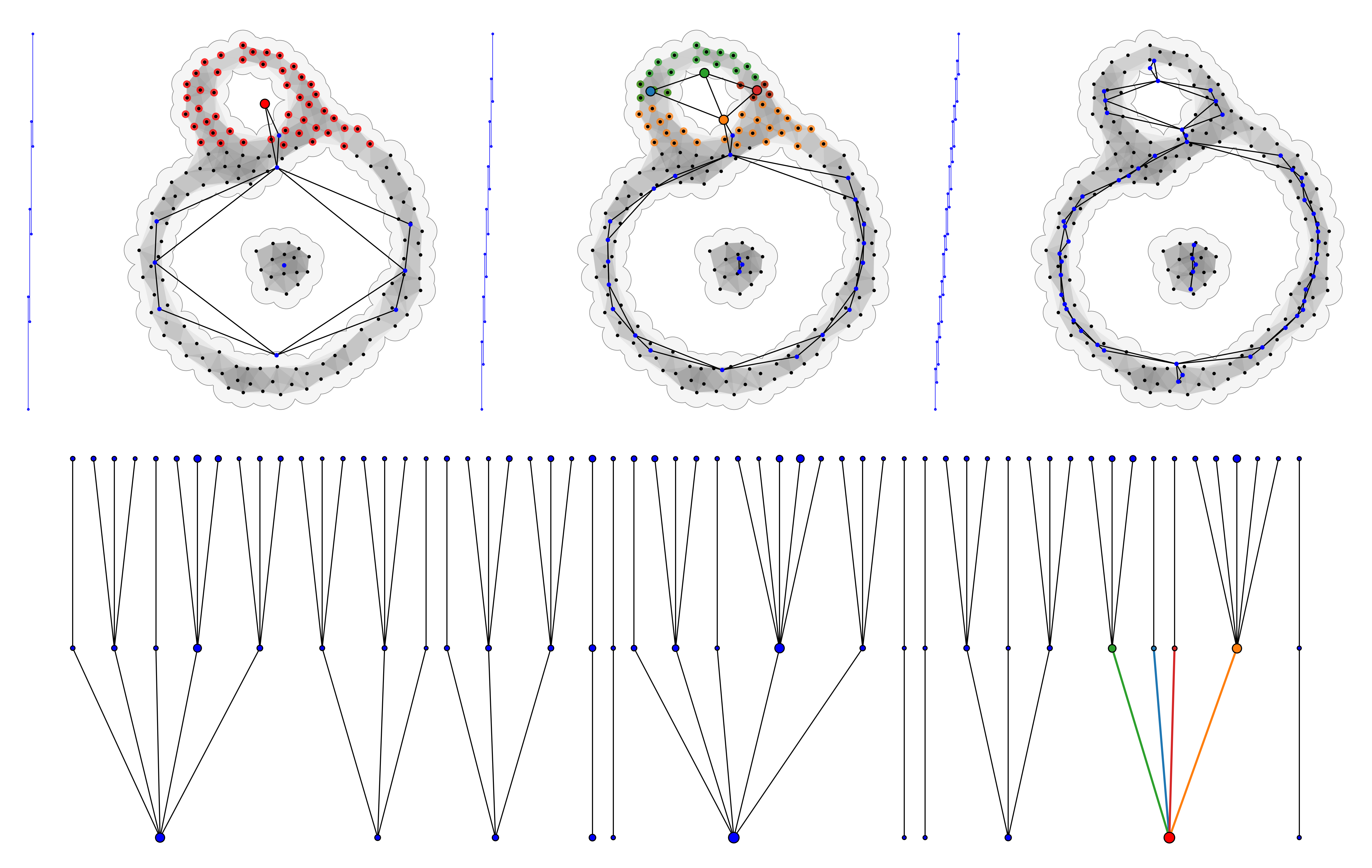}
  \caption{Multiscale mapper (top) and its corresponding THD (bottom).
      The filter function $f_P : P\to \RR$ takes each point of a finite sample to its $y$-coordinate.
      As the cover of $\im f_P\subset \RR$ by intervals is refined (left of each, in blue), so is the corresponding nerve.}\label{fig:multi_mapper_thd}
\end{figure}

\begin{notation}
  As in Section~\ref{sec:mapper_examples}, let $\sU^f : \op{\P}\to \pOp_L$ denote the function defined $\sU^f(p) \defined f^{-1}\sU(p) : I^p\to \Op_L$
  and let $\N_\sU^f\defined \N_{\sU^f} : \op{\P}\to \Toplc$ denote the filtration of spaces defined $\N_\sU^f(p) \defined \N_{\sU(p)}^f$.
\end{notation}

\begin{proposition}\label{prop:multi_thd}
  The poset $\El(\hat{R}^\sU_f)$ is isomorphic to the merge tree of $\N_\sU^f$.
\end{proposition}
\begin{proof}[Proof Sketch]
  It is a straightforward exercise to show that $\sU^f$ is a filtration of covers and that $\N_\sU^f\defined \N_{\sU^f} : \op{\P}\to \Toplc$ is a filtration of locally connected spaces.
  By Proposition~\ref{prop:display_nerve_iso}, $\N_{\sU(p)}^f\simeq \dis(R_{\eta_{\sU(p)}\circ f})$ for all $p\in \P$,
  so the merge tree of $\N_\sU^f$ has elements equivalent to pairs $\big(p,(\beta_p,\sigma^p)\big)$ for all $p\in\P$, $\sigma\in\lim\N_\sU$, and $\beta_p\in R_f(\sU_\sigma^p)$
  where $\big(q,(\beta_q,\sigma^q)\big)\preceq\big (p,(\beta_p,\sigma^p)\big)$ if $q\geq p$ and $\N_\sU^f\b{q\geq p}(\beta_q,\sigma^q) = (\beta_p,\sigma^p)$.
  Because $\sigma^p = \rho_p(\sigma)$ for all $p\in\P$ it follows that $\cT_{\N_\sU^f}\simeq \El(\hat{R}^\sU_f)$ as desired.
\end{proof}

Using the results of the previous section, we can compute the merge tree of $\N_\sU^f$ as follows.

Theorem~\ref{thm:multi_iso} implies that, for all $U\in \Op_X$,
there is a bijective correspondence between the elements of $\sM_\sU(R_f)(U) = R_f\big(\Meet \sU_U \big)$ (i.e. the connected components of $f^{-1}\big(\Meet \sU_U \big)$)
and the elements of $\lim R_f^\U(U)$: consistent families of elements $\beta_p\in R_f(\sU^p_U)$ for which $R_f\b{\sU^q_U\subseteq \sU^p_U}(\beta_q) = \beta_p$ for all $q\geq p$.
It follows that the desired THD can be efficiently computed as the set
\begin{equation}\label{eq:multi_thd_compute}
  \El(\hat{R}^\sU_f) = \Big\{\Big(p,\big(\sigma^p, R_f\big[\Meet \sU_\sigma\subseteq \sU_\sigma^p\big](\alpha)\big) \Big)
  \ \big|\
  p\in\P,\ \sigma^p\in \N_\sU(p),\text{ and } \alpha\in R_f\big(\Meet\sU_\sigma\big)\Big\}.
\end{equation}

\begin{example}
  Let $(X,\dist)$ be a locally connected metric space and let $f : X\to \RR$ be a continuous function.
  Let $P\subset X$ be a finite subspace and let $f_P\defined f|_P : P\to \RR$ denote the restriction of $f$ to $P$.
  Let $\sU : \op{\NN}\to \pOp_{\im f_P}$ be a strict refinement of good open covers of the image of $f_P$ and note that,
  because $P$ is finite, $\im f_P$ is finite, so $\sU$ is locally finite.
  Then $\sU^{f_P} : \op{\NN}\to \pOp_P$ is a refinement of covers of $P$ defined for all $n\in \NN$ as above:
  \begin{align*}
    \sU^{f_P}(n) : \bigsqcup R_{f_P}\sU(n) &\longrightarrow \Op_P\\
      (\beta_n,\sigma^n) &\longmapsto \beta_n,
  \end{align*}
  and $\N_\sU^{f_P} : \op{\NN}\to \Toplc$ is a filtration of locally connected spaces $\N_{\sU(-)}^{f_P}\simeq \dis(R_{\eta_{\sU(-)}\circ f})$.
  If there exists some $\ell \in \NN$ such that $\N_\sU^{f_P}(n) = \N_\sU^{f_P}(\ell)$ for all $n > \ell$, then the filtration of covers is parameterized by a finite poset,
  so Theorem~\ref{thm:multi_iso} implies that the desired THD (Equation~\eqref{eq:multi_thd_compute}) can be computed as follows:
  \begin{algorithm}
    \For{$n < \ell$}
    {
      \For{$\sigma^n\in \N_\sU(n)$}
      {
        \For{$\alpha\in R_{f_P}\big(\Meet \sU_\sigma \big)$}
        {
          $\alpha_n\gets R_{f_P}\big[\Meet \sU_\sigma\subseteq \sU_\sigma^n \big](\alpha)$\\
          add vertex $\big(n, (\sigma, \alpha_n)\big)$\\
          \If{$n > 1$}
          {
            $\alpha_{n-1} \gets R_{f_P}\big[\sU_\sigma^{n}\subseteq \sU_\sigma^{n-1} \big](\alpha_n)$\\
            $\sigma^{n-1} \gets \N_\sU\b{n\geq n-1}(\sigma^n)$\\
            add edge from $\big(n, (\sigma, \alpha_n)\big)$ to $\big(n-1, (\sigma^{n-1}, \alpha_{n-1})\big)$
          }
        }
      }
    }
  \end{algorithm}
\end{example}

\bibliographystyle{plain}
\bibliography{bibliography}

\appendix

\section{Categories cont.}\label{sec:cats2}

For any category \cat{C} we write $A,B\in\cat{C}$ to denote objects $A,B\in\ob(\cat{C})$ and $f: A\to B$ to denote arrows $f\in\hom_\cat{C}(A,B)$ in \cat{C}.
The \defn{opposite category} $\op{\cat{C}}$ associated with \cat{C} has objects $\Ob(\op{\cat{C}}) = \Ob(\cat{C})$ and arrows $\hom_{\op{\cat{C}}}(B,A) = \hom_\cat{C}(A,B)$ for all $A,B\in\cat{C}$.

\begin{definition}[Functor]
  Let \cat{B} and \cat{C} be categories.
  A \defn{functor} $F : \cat{B}\to \cat{C}$ consists of
  \begin{description}[style=multiline,leftmargin=2.7cm,font=\normalfont,align=right]
    \item[\emph{(Objects)}] an \emph{object map} $F(-) : \ob(\cat{B})\to \ob(\cat{C})$;
    \item[\emph{(Arrows)}] for all $x,y\in\cat{B}$, \emph{arrow maps} $F\b{-} : \hom_\cat{B}(x,y)\to \hom_\cat{C}(F(x),F(y))$;
  \end{description}
  subject to the following conditions:
  \begin{description}
    \item[FUN1]\emph{(Unit)} For all $x\in\cat{B}$, $F\b{\id_x} = \id_{F(x)}$.
    \item[FUN2]\emph{(Composition)} For all composable pairs $f,g\in\Mor(\cat{B})$, $F\b{g \circ f} = F\b{g}\circ F\b{f}$.
  \end{description}
\end{definition}

\begin{notation}
  We refer to the objects $F(x)\in\cat{C}$ as the \emph{components} of $F$ in \cat{C},
  and the morphisms $F\b{f}\in\Mor(\cat{C})$ as the \emph{arrow maps} of $F$ in \cat{C}.
\end{notation}

\begin{definition}[Natural Transformation]
  Let $F,G : \cat{B}\rightrightarrows\cat{C}$ be parallel functors.
  A \defn{natural transformation} $\theta : F\Rightarrow G$ is a morphism of functors
  given by an object map $\theta_{(-)} : F(-)\to G(-)$ such that Diagram~\eqref{dgm:natural} commutes for all $f\in\mor(\cat{B})$.
  \begin{equation}\label{dgm:natural}
    \begin{tikzcd}[row sep=large]
          b \arrow[d,"f"]
      &   F(b)\arrow[d,"F\b{f}"]\arrow[rr,"\theta_b"]
      &&  G(b)\arrow{d}{G\b{f}}
      \\
          c
      &   F(c)\arrow[rr,"\theta_c"]
      &&  G(c)
    \end{tikzcd}
  \end{equation}
\end{definition}

For any pair of categories \cat{B},\cat{C}, the \defn{functor category} $\cat{C}^\cat{B} = \Fun{\cat{B},\cat{C}}$ has functors $F:\cat{B}\to\cat{C}$ as objects and natural transformations $\theta : F\Rightarrow G$ as arrows.

\subsection{Slice Categories and the Category of Elements}\label{sec:slice}

\begin{definition}[Slice Categories]
  For any object $C$ in a category $\cat{C}$ the category of \defn{objects over $C$} is the category $\cat{C}/C$ with pairs $(A,f)$ as objects for each $f\in\Hom_\cat{C}(A,C)$,
  and arrows $\gamma : (A,f)\longuto (B,g)$ for each $\supp{\gamma}\in \Hom_\cat{C}(A,B)$ with $f = g\circ \supp{\gamma}$.
  \begin{equation}
    \begin{tikzcd}[sep=large]
      A \arrow[rr,dotted,"\supp{\gamma}"description]
        \arrow[dr,"f"description]
        \arrow[dr,phantom,""{name=f,above}]
      & & B
        \arrow[dl,"g"description]
        \arrow[dl,phantom,pos=0.5,""{name=g,above}]
      \\
      & C. &
      \arrow[from=f,to=g,shift right=1,shorten >=2ex,shorten <=1.5ex]
      \arrow[from=f,to=g,shift right=2,symbol=\bdot]
    \end{tikzcd}
  \end{equation}
  The category of \defn{objects under $C$} is defined $C/\cat{C} \defined \op{\cat{C}}/C$.
\end{definition}

\begin{definition}[Category of Elements]
  Let $F : \cat{C}\to \Set$ be a functor from a small category \cat{C} to \Set.
  The \defn{category of elements} $\El(F)$ of $F$ has an object $(C,\beta)$ for each objet $C\in\cat{C}$ and each \emph{element} $\beta\in F(C)$,
  and arrows $s : (C',\beta')\to (C,\beta)$ for each $\supp{s} : C'\to C$ of $\cat{C}$ with $F\b{\supp{s}}(\beta') = \beta$.
\end{definition}

\begin{notations}\hfill
  \begin{enumerate}
    \item Objects of $\El(F)$ may be denoted $\beta_C \defined (C,\beta)\in\El(F)$.
    \item Letting $\rho_C : \lim F\to F(C)$ denote the canonical projection, for any $\beta_C\in\El(F)$ let
      \[\cev{\beta}_C\defined \{\alpha\in\lim F\mid \rho_C(\alpha) = \beta\}.\]
    \item If $\cat{C} = \P$ is a poset then $\El(F)$ is a poset in which
      $\beta'_{p'}\preceq \beta_p$ if $p'\leq p$ and $F\b{p'\leq p}(\beta') = \beta$.
  \end{enumerate}
\end{notations}

\subsection{Limits and Adjunctions}\label{sec:limits}

Let $\cat{C}$ be a category.
For any small category $\cat{J}$ the \defn{diagonal functor} $\Delta_\cat{J} : \cat{C}\to \cat{C}^\cat{J}$ takes each object $C\in\cat{C}$ to the constant functor $\Delta_\cat{J}(C) : \cat{J}\to \cat{C}$ defined $j\mapsto C$ for all $j\in \cat{J}$.
A \defn{cone} from an object $C\in\cat{C}$ to a functor $F:\cat{J}\to \cat{C}$ is a natural transformation $\theta : \Delta_\cat{J}(C)\Rightarrow F$.
Dually, a \defn{cocone} from a $F$ to $C$ is a natural transformation $\varphi : F\Rightarrow \Delta_\cat{J}(C)$.

\begin{definition}[Limits]
  A cone $\rho : \Delta_\cat{J}(L)\Rightarrow F$ is universal if, for each cone $\theta : \Delta_\cat{J}(C)\Rightarrow F$ there is a unique arrow $u : C\to L$ with $\rho_j\circ u = \theta_j$ for all $j\in\cat{J}$:
  \begin{equation}
    \begin{tikzcd}[column sep=huge]
      && F(j)\arrow[dd,"F\b{f}"]\\
      C\arrow[r,dotted,"\exists ! u"description]\arrow[urr, bend left=10,"\theta_j"]\arrow[drr, bend right=10,"\theta_k"'] &
      L\arrow[ur,"\rho_j"]\arrow[dr,"\rho_k"'] &\\
      && F(k)
    \end{tikzcd}
  \end{equation}
  If it exists, the universal cone $\rho : \Delta_\cat{J}(L)\Rightarrow F$ is the \defn{(projective) limit} $\lim F$ of $F$.
  We refer to the vertex $L$ as the limit $\lim F$ of $F$ and the components $\rho_j : \lim F\to F(j)$ as the associated canonical projections.
\end{definition}

\begin{examples}
  Let $\cat{J}$ be a small category and let \P be a poset.
  \begin{enumerate}
    \item The limit of a functor $F : \cat{J}\to \P$ is the greatest lower bound or \defn{meet} $\Meet F = \lim F$.
    \item For any functor $F : \cat{J}\to \Set$ the elements of the projective limit
      \[ \proj F\simeq \Big\{\{\beta_j\}\in\prod_{j\in \cat{J}} F(j)\mid F\b{f}(\beta_{k}) = \beta_j\text{ for all } f\in\Hom_\cat{J}(k,j)\Big\}\]
      may be regarded as \emph{consistent} subsets of $\El(F)$ indexed by objects in $\cat{C}$~\cite{woolf08fundamental}.
    \end{enumerate}
\end{examples}

\begin{definition}[Colimit]
  A cocone $\iota : F\Rightarrow \Delta_\cat{J}(L)$ is universal if, for each cocone $\varphi : F\Rightarrow\Delta_\cat{J}(C)$ there is a unique arrow $u : L\to C$ with $u\circ\iota_j = \varphi_j$ for all $j\in\cat{J}$:
  \begin{equation}
    \begin{tikzcd}[column sep=huge]
      F(j)\arrow[dd,"F\b{f}"]\arrow[drr,bend left=10,"\varphi_j"]\arrow[dr,"\iota_j"]&&\\
      & L\arrow[r,dotted,"\exists ! u"description] & C\\
      F(k)\arrow[urr,bend right=10,"\varphi_k"']\arrow[ur,"\iota_k"'] &&
    \end{tikzcd}
  \end{equation}
  If it exists, the universal cocone $\iota : F\Rightarrow\Delta_\cat{J}(L)$ is the \defn{colimit} or \defn{inductive limit} $\colim F$ of $F$.
  We will refer to the vertex $L$ as the colimit $\colim F$ of $F$ and the components $\iota_j : F(j)\to\colim F$ as the canonical coprojections.
\end{definition}

\begin{examples}
  Let $\cat{J}$ be a small category and let \P be a poset.
  \begin{enumerate}
    \item The colimit of a functor $F : \cat{J}\to \P$ is the lowest upper bound or \defn{join} $\Join F = \colim F$.
    \item The inductive limit of a functor $F : \cat{J}\to \Set$ may be expressed as the quotient
      \[ \colim F \simeq \Big(\bigsqcup_{j\in \cat{J}} F(j)\Big)/\!\sim\]
      where $\sim$ is the equivalence relation generated by $\alpha_j\sim\beta_k$ if there exists some $f : j\to k$ such that $F\b{f}(\alpha) = \beta$.
  \end{enumerate}
\end{examples}

\begin{definition}[Filtrant]
  A category $\cat{I}$ is \defn{filtrant} if
  \begin{enumerate}
    \item $\cat{I}$ is non empty,
    \item for any $i$ and $j$ in $\cat{I}$ there exists $k\in \cat{I}$ and morphisms $i\to k$ and $j\to k$,
    \item for any parallel morphisms $f,g : i\rightrightarrows j$, there exists a morphism $h : j\to k$ such that $h\circ f= h\circ g$.
  \end{enumerate}
  In particular, a poset \P is filtrant if it is directed.
\end{definition}

\begin{proposition}[Kashiwara and Schapira~\cite{kashiwara05categories} Proposition 3.1.3]\label{prop:filtrant}
  Let $F : \cat{I}\to \Set$ be a functor with $\cat{I}$ small and filtrant.
  Then
  \[ \colim F\simeq \Big(\bigsqcup_{i\in\cat{I}} F(i)\Big)/\! \sim\]
  where $\alpha_i\sim\alpha'_j$ \emph{if and only if} there exist $f : i\to k$ and $g : j\to k$ such that $F\b{f}(\alpha)=F\b{g}(\alpha')$.
\end{proposition}

\begin{theorem}[Mac Lane~\cite{mac13categories} Theorem IX.2.1]\label{thm:limit_interchange}
  If \cat{P} is a finite category and \cat{J} is small and filtered category,
  then for any bifunctor $F : P\times J\to \Set$ the following canonical arrow is a bijection:
  \[\kappa : \colim_{j\in\cat{J}}\lim_{p\in\cat{P}} F(p,j)\longisorightarrow \lim_{p\in\cat{P}}\colim_{j\in\cat{J}} F(p,j).\]
\end{theorem}

\begin{definition}[Pullback]
  A \defn{pullback} or \defn{fibered product} of arrows $f : A\to C$ and $g : B\to C$ in a category \cat{C} is given by an object $A\times_C B$ and projection maps $\rho_A : A\times_C B\to A$ and $\rho_B : A\times_C B\to B$ satisfying the following universal property:
  \begin{quote}
    \emph{For all objects $X\in\cat{C}$ and arrows $p : X\to A$ and $q : X\to B$ such that $f\circ p = g\circ q$, there exists a unique arrow $u : X\to A\times_C B$ such that $p = \rho_A\circ u$ and $q = \rho_B\circ u$:}
    \begin{equation}
      \begin{tikzcd}[sep=huge]
        X\arrow[drr,bend left=10,"q"]\arrow[ddr,bend right=10,"p"]\arrow[dr,dotted,"\exists!u"description] &&\\
        & A\times_C B
          \arrow[r, dotted, "\rho_B"]
          \arrow[d, dotted, "\rho_A"]
        & B\arrow[d,"g"] \\
        & A\arrow[r,"f"] & C
      \end{tikzcd}
    \end{equation}
  \end{quote}
\end{definition}

\begin{example}
  The pullback of functions $f : A\to C$ and $g : B\to C$ in \Set is given by the set
  \[ A\times_C B = \{(a,b)\in A\times B\mid f(a) = g(b)\}.\]
\end{example}

\begin{definition}[Change of Base Functor]
  Suppose \cat{C} has pullbacks and let $f :A\to C$ be an arrow in \cat{C}.
  The \defn{change of base functor} associated with the object $(A,f)\in\cat{C}/C$ over $C$ is defined
  \begin{align*}
    f^\star : \cat{C}/C &\longrightarrow \cat{C}/A\\
    (B,g) &\longmapsto (A\times_{C} B, g^f)
  \end{align*}
  where $A\times_{C} B$ is the pullback of $f$ along $g : B\to C$ (see Appendix~\ref{sec:limits}).
  \begin{equation}
    \begin{tikzcd}[sep=huge]
      A\times_{C} B
          \arrow[r, dotted, "f^g"]
          \arrow[d, dotted, "g^f"description]
      & B \arrow[d,"g"]
      \\
      A   \arrow[r,"f"]
      & C.
    \end{tikzcd}
  \end{equation}
  The change of base functor is right adjoint to the functor $\Sigma_f : \cat{C}/A\to \cat{C}/C$ associated with $f$ that takes an object $(D,h)\in\cat{C}/A$ over $A$
  to the object over $C$ given by the composition $f\circ h : Z\to C$ (see Theorem I.9.4~\cite{maclane92sheaves}):
  \begin{align*}
    \Sigma_{f} : \cat{C}/A&\longrightarrow \cat{C}/C\\
      (D,h) &\longmapsto (Z, f\circ h).
  \end{align*}
  That is, we have the following adjunction:
  \[\adjoint{\cat{C}/A}{\cat{C}/C.}["\Sigma_f"]["f^\star"]\]
  The unit and counit of $\Sigma_f\dashv f^\star$ will be denoted
  \begin{align*}
    \mathbf{(Unit)}\hspace{5ex} \upsilon^f : \id_{\cat{C}/A}&\Longrightarrow f^\star\Sigma_f,\hspace{15ex}\\
    \mathbf{(Counit)}\hspace{5ex} \chi^f : \Sigma_f f^\star&\Longrightarrow \id_{\cat{C}/C}.\hspace{15ex}
  \end{align*}
\end{definition}

\begin{definition}[Kashiwara and Schapira~\cite{kashiwara05categories} Definition 2.2.6]\label{def:stable_base}
  We say that colimits in \cat{C} indexed by \cat{I} are \defn{stable by base change} if,
  for any morphism $f : Y\to Z$ of \cat{C} the change of base functor $f^\star : \cat{C}/Z\to \cat{C}/Y$
  commutes with colimits indexed by \cat{I}.
  Equivalently, for any inductive system $\{X_i\}_{i\in I}$ in \cat{C} and any pair of morphisms $Y\to Z$ and $\lim_{i\in \cat{I}} X_i\to Z$ in \cat{C},
  we have an isomorphism
  \begin{equation}\label{eq:stable_base}
    \colim_{i\in\cat{I}} (X_i\times_Z Y)\longisorightarrow\big( \lim_{i\in I} X_i \big)\times _Z Y.
  \end{equation}
  If \cat{C} admits small colimits and~\eqref{eq:stable_base} holds for any small category \cat{I},
  we shall say that small colimits in \cat{C} are \defn{stable by base change}.
\end{definition}

Importantly, the category \Set admits small colimits that are stable by base change.

\begin{definition}[Pushout]
  A \defn{pushout} or \defn{fibered coproduct} of arrows $f : A\to B$ and $g : A\to C$ in a category \cat{C} is given by an object $B\sqcup_A C$ and projection maps $\iota_B : B\to B\sqcup_A C$ and $\iota_C : C\to B\sqcup_A C$ satisfying the following universal property:
  \begin{quote}
    \emph{For all objects $X\in\cat{C}$ and arrows $p : B\to X$ and $q : C\to X$ such that $p\circ f = q\circ g$, there exists a unique arrow $u : B\sqcup_A C\to X$ such that $p = u\circ \iota_B$ and $q = u\circ \iota_C$:}
    \begin{equation}
      \begin{tikzcd}[sep=huge]
        A
          \arrow[r,"g"]\arrow[d,"f"]
        & C
          \arrow[d,dotted,"\iota_C"]
          \arrow[ddr,bend left=10,"q"] & \\
        B
          \arrow[r,dotted,"\iota_B"]
          \arrow[drr,bend right=10,"p"]
        & B\sqcup_A C
          \arrow[dr,dotted,"\exists!u"description]
        &\\
        && X.
      \end{tikzcd}
    \end{equation}
  \end{quote}
\end{definition}

\begin{example}
  The pushout of functions $f : A\to B$ and $g : A\to C$ in \Set is given by the set the quotient
  \[ B\sqcup_A C = \big(B\sqcup C)/\!\sim\]
  where $\sim$ is the equivalence relation generated by $(b,0)\sim (c,1)$ if there exists some $a\in A$ such that $f(a) = b$ and $g(a) = c$.
\end{example}

\begin{definition}[Adjoint Functors (Kashiwara and Schapira~\cite{kashiwara05categories} Definition 1.5.2)]
  Let $L : \cat{C}\to \cat{C}'$ and $R : \cat{C}'\to \cat{C}$ be functors.
  We write $L\dashv R$ and say that $L$ is \defn{left adjoint} to $R$ (equivalently, $R$ is \defn{right adjoint} to $L$)
  if there exists an isomorphism of bifunctors from $\op{\cat{C}}\times \cat{C}$ to $\Set$:
  \[ \Hom_{\cat{C}'}\big(L(-), -\big)\simeq \Hom_{\cat{C}}(-,R(-)).\]
  Let $X\in\cat{C}$. Applying the isomorphism above with $X$ and $L(X)$, we find that the isomorphism
  \[ \Hom_{\cat{C}'}\big(L(X), L(X))\simeq \Hom_{\cat{C}}\big(X, RL(X))\big),\]
  and the identity of $L(X)$ defines a morphism $X\to RL(X)$.
  Similarly, we construct $LR(Y)\to Y$, and these morphisms are functorial with respect to $X$ and $Y$.
  Hence, we have constructed natural transformations
  \begin{align*}
    \mathbf{(Unit)}\hspace{5ex } \eta &: \id_{\cat{C}} \Longrightarrow RL\hspace{15ex}\\
    \mathbf{(Counit)}\hspace{5ex} \epsilon &: LR \Longrightarrow \id_{\cat{C}'}\hspace{15ex}.
  \end{align*}
  Equivalently, $L\dashv R$ if there exist natural transformations $\eta : \id_{\cat{C}}\Rightarrow RL$ and $\epsilon : LR\Rightarrow\id_{\cat{C}'}$
  such that the following diagrams of natural transformations commute:
  \begin{equation}
    \begin{tikzcd}[sep=huge]
      L
        \arrow[dr, equal, "\id_L"']
        \arrow[r, Rightarrow,"L\eta"]
      & LRL
        \arrow[d, Rightarrow, "\epsilon L"]
      && R
        \arrow[d,Rightarrow,"\eta R"']
        \arrow[dr,equal,"\id_R"]
      &\\
      & L
      && RLR
        \arrow[r,Rightarrow,"R\epsilon"']
      & R.
    \end{tikzcd}
  \end{equation}
\end{definition}

\subsection{Cosheaves cont.}\label{sec:cosheaves2}

\begin{definition}[Cosheaf]\label{def:cosheaf}
  A precosheaf $F$ on $X$ is a \defn{cosheaf} if, for every open set $U\in\Op_X$ and every open cover $\U : I\to \Op_U\subseteq \Op_X$ of $U$,
  the set $F(U)$ is isomorphic to the coequalizer of the diagram
  \[ \bigsqcup_{i,j\in I} F\big(\U(i)\cap \U(j)\big) \rightrightarrows \bigsqcup_{i\in I} F\big(\U(i)\big).\]
  That is, $F(U)$ is isomorphic to the colimit of the family of maps
  \[ F\big(\U(i)\cap \U(j)\big) \rightrightarrows F\big(\U(i)\big)\sqcup F\big(\U(j)\big)\text{ for all } \U(i)\cap\U(j)\neq\emptyset.\]
\end{definition}

If $F$ is a precosheaf on $X$ and $\U : I\to \Op_U\subseteq \Op_X$ is a cover of $U\in\Op_X$ then the pushout of $F\b{\U(i)\cap\U(j)\subseteq \U(i)}$
and $F\b{\U(i)\cap\U(j)\subseteq \U(j)}$ (Diagram~\eqref{dgm:cosheaf_pushout}) is equivalent to the quotient $\big(F\big(\U(i)\big)\cap F\big(\U(j)\big)\big)/\!\sim$
where $\sim$ is the equivalence relation \emph{generated by} $(\delta,i)\sim (\zeta,j)$ if there exists some $\alpha\in F\big(\U(i)\cap \U(j)\big)$ that maps to $\delta\in F\big(\U(i)\big)$ and $\zeta\in F\U(j)$.
\begin{equation}\label{dgm:cosheaf_pushout}
  \begin{tikzcd}[sep=huge]
    F\big(\U(i)\cap\U(j)\big)\arrow[r]\arrow[d]
    &F\big(\U(j)\big)\arrow[d, dotted] \\
    F\big(\U(i)\big)\arrow[r, dotted]
    & F(U).
  \end{tikzcd}
\end{equation}

If $F$ is a cosheaf then $F(U)\simeq \big(F\big(\U(i)\big)\sqcup F\big(\U(j)\big)\big)/\!\sim$ implies that the elements $\beta\in F(U)$
may be identified with equivalence classes $\b{\beta_i}$ of elements $\beta_i\in F\big(\U(i)\big)$ and $\beta_j\in F\big(\U(j)\big)$ for which
\begin{enumerate}
  \item $F\b{\U(i)\subseteq U}(\beta_i) = F\b{\U(j)\subseteq U}(\beta_j) = \beta$ and
  \item there exists some $\alpha\in F\big(\U(i)\cap \U(j)\big)$ such that
    \begin{align*}
      &F\b{\U(i)\cap \U(j)\subseteq \U(i)}(\alpha) = \beta_i\\
      \text{and } &F\b{\U(i)\cap \U(j)\subseteq \U(j)}(\alpha) = \beta_j.
    \end{align*}
\end{enumerate}
Moreover, for every open set $U\in\Op_X$ and every cover $\U : I\to \Op_U\subseteq \Op_X$,
the universal arrow making Diagram~\eqref{dgm:cosheaf} commute for all $\sigma\subseteq\tau$ in $\N_\U$ is an isomorphism~\cite{curry14sheaves}:
\[ \h : \colim_{\sigma\in\N_\U}  F(\U_\sigma) \longisorightarrow F(U).\]
This property is often referred to as the \emph{cosheaf axiom} in the literature~\cite{curry14sheaves}.
\begin{equation}\label{dgm:cosheaf}
  \begin{tikzcd}[column sep=huge]
    F(\U_\tau)
      \arrow[dd,"F\b{\U_\tau\subseteq \U_\sigma}"]
      \arrow[drr,bend left=10,"F\b{\U_\tau\subseteq U}"description]
      \arrow[dr, tail, "\iota_\tau"]
    &&\\
    & \displaystyle\colim_{\sigma\in\N_\U} F(\U_\sigma)
      \arrow[r,dotted,"\exists ! \h "description]
    & F(U)\\
    F(\U_{\sigma})
      \arrow[urr,bend right=10,"F\b{\U_\sigma\subseteq U}"description]
      \arrow[ur, tail, "\iota_\sigma"'] &&
  \end{tikzcd}
\end{equation}

\begin{notations}
  Let $F$ be a precosheaf on $X$.
  \begin{enumerate}
    \item The restriction of $F$ to an open set $U\in\Op_X$ (regarded as a subspace of $X$) will be denoted
      \[ F_U\defined F|_{\Op_U} : \Op_U\longrightarrow \Set.\]
      Importantly, if $F$ is a cosheaf on $X$, then $F_U$ is a cosheaf on $U$ for all $U\in\Op_X$.
    \item The restriction of $F$ to the open neighborhoods of a point $x\in X$ will be denoted
      \[ F^x \defined F|_{\Nb_X(x)} : \Nb_X(x)\longrightarrow \Set.\]
  \end{enumerate}
\end{notations}

\section{The Display Space of a Cosheaf}\label{sec:display}

Let $X$ be a topological space and let $F$ be a precosheaf on $X$.
The category of elements $\El(F)$ is a poset in which $\beta'_{U'}\preceq \beta_U$ if $U'\subseteq U$ and $F\b{U'\subseteq U}(\beta') = \beta$.
As in~\cite{funk95display}, we endow $\El(F)$ with the \emph{co}specialization topology in order to obtain
a topological space with points $\beta_U$ for each $U\in\Op_X$ and $\beta\in F(U)$.

\begin{definition}[Total Locale~\cite{funk95display}]
  The \defn{total locale} of a precosheaf $F$ on $X$ is the topological space $\sE_X(F)$ with points given by pairs
  $\beta_U\defined (U,\beta)\in \El(F)$ for each open set $U\in\Op_X$ and each \emph{element} $\beta\in F(U)$,
  and open sets generated by principal down sets
  \[ \down{\beta_U} \defined \big\{ \beta'_{U'}\in \sE_X(F)\mid U'\subseteq U\text{ and } F\b{U'\subseteq U}(\beta') = \beta\big\}.\]
\end{definition}

\begin{remark}
  The term \emph{locale} refers to a generalization of topological spaces used in the original work~\cite{funk95display}.
  Much of our results are a direct application of this work in the special case of topological spaces.
  Extending the results of this paper to the localistic setting is the subject of future work.
  We direct the interested reader to Mac Lane and Moerdijk~\cite{maclane92sheaves} for a full treatment of localistic topoi.
\end{remark}

\begin{notation}
  Let $\id_X : U\mapsto\{\pt\}$ denote the terminal (pre)cosheaf on $X$ and let $\Op_X(-) : X\to \sE_X(\id_X)$
  denote the continuous function that takes each point $x\in X$ to the up-closed, down-directed subset $\down{\Op_X(x)}\in\sE_X(\id_X)$.
\end{notation}

The display space of a precosheaf $F$ on $X$ was originally defined~\cite{funk95display} as the pullback
of $\Op_X(-)$ along the continuous map induced by the unique natural transformation $F\Rightarrow\id_X$ in \Top:
\begin{equation}
  \begin{tikzcd}[sep=huge]
    \dis(F)
      \arrow[r, dotted, "\pi_F"]
      \arrow[d, dotted, "\gamma_F"]
      &\sE_X(F) \arrow[d, dotted, "\exists!"]\\
    X\arrow[r, "\Op_X(-)"]
    & \sE_X(\id_X).
  \end{tikzcd}
\end{equation}
Alternatively, we can define the display space as the (topological) coproduct of costalks as follows
(see also Brown et al.~\cite{brown21probabilistic} and Woolf~\cite{woolf08fundamental} Appendix B).

\begin{definition}[Costalk]
  For any $x\in X$ and any open neighnorhood $U\in\Op_X(x)$ let $\rho_U^x : \lim F^x\to F(U)$ be the canonical projection
  and let $p_U^x : \lim F^x\to \sE_X(F)$ be defined $\alpha\mapsto (U,\rho_U^x(\alpha))$.
  Let $\dis_x(F)$ denote the topological space given by topologizing the costalk $\lim F^x$
  with the initial topology associated with the projective system
  \[ \big\{p_U^x : \lim F^x \to \sE_X(F)\big\}_{U\in\Op_X(x)}.\]
\end{definition}

The points $\alpha\in\dis_x(F)$ may be regarded as subsets of $\sE_X(F)$
that are ``consistent'' on the open neighborhoods of $x$ in the sense that, for all open sets $U\supseteq U'\ni x$,
\[ \rho_U^x(\alpha) = F\b{U'\subseteq U}\circ \rho_{U'}^x(\alpha).\]
Recalling that the topology of the total locale $\sE_X(F)$ is generated by principal down sets $\down{\beta}_U$
for each $U\in\Op_X$ and $\beta\in F(U)$, it follows that the (initial) topology on $\dis_x(F)$ is generated by basic open sets
\[ \cev{\beta}_U^x \defined \big(p_U^x\big)^{-1}(\down{\beta}_U) = \big\{\alpha\in \lim F^x \mid \rho_U^x(\alpha) = \beta\big\}.\]

We will conclude this section with some useful results.

Let $X$ be a topological space and let $\P$ be a poset endowed with the specialization topology.

\begin{proposition}\label{prop:rho_bi}
  For any precosheaf $G$ on $\P$ and $p\in\P$,
  \[\dis_p(G)\simeq G(\up{p}).\]
\end{proposition}
\begin{proof}
  Let $p\in\P$.
  For all $\beta\in G(\up{p})$ there exists a consistent set of elements
  \[ \alpha = \{G\b{\up{p}\subseteq S}(\beta)\}_{S\ni p}\in \dis_p(G)\]
  such that $\rho_{\up{p}}^p(\alpha) = \beta$, so the canonical projection $\rho_{\up{p}}^p$ is surjective.
  It therefore suffices to show that $\rho_{\up{p}}^p : \dis_p(G) \twoheadrightarrow G(\up{p})$ is injective.

  Suppose $\rho_{\up{p}}^p(\alpha) = \rho_{\up{p}}^p(\alpha')$ for some $\alpha,\alpha'\in \dis_p(G)$.
  Because $\down{p}$ is initial in $\Nb_\P(p)$, $\rho_{\up{p}}^p(\alpha) = \rho_{\up{p}}^p(\alpha')$ implies that
  $\rho_S^p(\alpha) = G\b{\up{p}\subseteq D}\big(\rho_{\up{p}}^p(\alpha)\big) = G\b{\up{p}\subseteq S}\big(\rho_{\up{p}}^p(\alpha')\big) = \rho_S^p(\alpha')$
  for all $S\ni p$.
  Recalling that
  \[\dis_p(G) = \lim G^p \simeq \big\{ \{\alpha_S\}\in\prod_{S\ni p} G(S)\mid G\b{S'\subseteq S}(\alpha_{S'}) = \alpha_S\text{ for all } S'\subseteq S \big\},\]
  it follows that $\alpha = \alpha'$, so we may conclude that $\rho_{\up{p}}^p$ is injective and therefore a bijection as desired.
\end{proof}

\begin{remark}
  Let $G$ be a precosheaf on a poset $\P$ and recall the terminal precosheaf $\id_\P : U\mapsto \{\pt\}$.
  Then $\Op_\P(-) : \P\to \sE_\P(\id_\P)$ is a homeomorphism that takes each $p\in\P$ to the up-closed, down directed subset $\up{(\down{p})}\in \sE_\P(\id_\P)$.
  It follows that $\dis(G)\simeq \sE_\P(G)$.
\end{remark}

\begin{lemma}\label{lem:special_spatial}
  Every precosheaf on a poset $\P$ endowed with the specialization topology is spatial.
\end{lemma}
\begin{proof}
  Let $G$ be a precosheaf on $\P$ and let $p\in\P$ and $\beta\in G(\up{p})$.
  By Proposition~\ref{prop:rho_bi}, $\cev{\beta}_{\up{p}}\neq\emptyset$ for all $p\in\P$ and $\beta\in G(\up{p})$,
  so it suffices to show that $\cev{\beta}_{\up{p}}$ is connected.

  Suppose for the sake of contradiction that $\cev{\beta}_{\up{p}} = \cev{\beta}'_{\up{p'}}\cup \cev{\beta}''_{\up{p''}}$
  and $\cev{\beta}'_{\up{p'}}\cap \cev{\beta}''_{\up{p''}} = \emptyset$ for some $\beta'\in G(\up{p'})$ and $\beta''\in G(\up{p''})$.
  By Proposition~\ref{prop:rho_bi}, there exists a unique $\alpha\in\dis_p(G)$ 
  such that $\rho_{\up{p}}^p(\alpha) = \beta$, so $(\alpha,p)\in\cev{\beta}_{\up{p}}$.
  So either $(\alpha,p)\in\cev{\beta}'_{\up{p'}}$ or $(\alpha,p)\in \cev{\beta}''_{\up{p''}}$
  under the assumption that $\cev{\beta}_{\up{p}} = \cev{\beta}'_{\up{p'}}\cup \cev{\beta}''_{\up{p''}}$.

  Assume w.l.o.g. that that $(\alpha,p)\in \cev{\beta}'_{\up{p'}}$.
  Then $p'\leq p$ and $\rho_{\up{p'}}^p(\alpha) = \beta'$ implies that
  \begin{align*}
    G\b{\up{p}\subseteq \up{p'}}(\beta) &= G\b{\up{p}\subseteq \up{p'}}\circ \rho_{\up{p}}^{p}(\alpha)
      = \rho_{\up{p'}}^p(\alpha)
      = \beta'.
  \end{align*}
  If $\cev{\beta}''_{\up{p''}}\neq\emptyset$ then Proposition~\ref{prop:rho_bi} implies that there exists some
  $\alpha''\in \dis_{p''}(G)$ 
  such that $\rho_{\up{p''}}^{p''}(\alpha'') = \beta''$,
  so $(\alpha'',p'')\in\cev{\beta}''_{\up{p''}}$.
  Because $\cev{\beta}_{\up{p}} = \cev{\beta}'_{\up{p'}}\cup \cev{\beta}''_{\up{p''}}$, $\cev{\beta}''_{\up{p''}}$ is a subset of $\cev{\beta}_{\up{p}}$.
  So $(\alpha'',p'')\in \cev{\beta}_{\up{p}}$ implies that $p\leq p''$ and $\rho_{\up{p}}^{p''}(\alpha'') = \beta$, thus
  \[ \beta
    = \rho_{\up{p}}^{p''}(\alpha'')
    = G\b{\up{p''}\subseteq \up{p}}\circ \rho_{\up{p''}}^{p''}(\alpha'')
    = G\b{\up{p''}\subseteq \up{p}}(\beta'').\]

  Because we have already shown that $G\b{\up{p}\subseteq \up{p'}}(\beta) = \beta'$ we have
  \[ G\b{\up{p''}\subseteq \up{p'}}(\beta'')
    = G\b{\up{p}\subseteq \up{p'}}\circ G\b{\up{p''}\subseteq \up{p}}(\beta'')
    = G\b{\up{p}\subseteq \up{p'}}(\beta) = \beta'.\]
  So $\rho_{\up{p'}}^{p''}(\alpha'') = G\b{\up{p''}\subseteq \up{p'}}\circ \rho_{\up{p''}}^{p''}(\alpha'') = \beta'$
  implies that $(\alpha'',\beta'')\in \cev{\beta}'_{\up{p'}}$: a contradiction, as we have assumed that
  $\cev{\beta}'_{\up{p'}}\cap \cev{\beta}''_{\up{p''}} = \emptyset$.
  So $\cev{\beta}''_{\up{p''}} = \emptyset$, and we may therefore conclude that $\cev{\beta}_{\up{p}}$ is connected.
\end{proof}

The open sets of the pullback $X\times_\P\dis(G)$ of a continuous function $h : X\to \P$ along $\gamma_G$ are generated by pullbacks $U\times_\P\cev{\beta}_{S}$
of the restrictions $h|_U$ and $\gamma_G|_{\cev{\beta}_{S}}$ for $U\in\Op_X$ and $S\in\Op_\P$.
If $U'\times_\P \beta'_{h_!(U')}\subseteq U\times_\P \beta_{h_!(U)}$ then,
letting $\eta:\id_{\Op_X} \Rightarrow h^{-1}h_!$ denote the unit of $h_!\dashv h^{-1}$, we have that
$\eta(U')\subseteq \eta(U)$ and $h^\dagger(G)\b{U'\subseteq U}(\beta') = \beta$.

\begin{lemma}\label{lem:poset_pullback}
  Let $G$ be a precosheaf on $\P$ and let $\beta\in F(S)$ and $\beta'\in F(S')$ for open sets $S,S'\in\Op_\P$.
  If $\cev{\beta}'_{S'}\subseteq \cev{\beta}_S$ then $S'\subseteq S$ and $G\b{S'\subseteq S}(\beta') = \beta$.
\end{lemma}
\begin{proof}
  For all $p\in S'$ and $\alpha\in\cev{\beta'}_{S'}^p\subseteq \dis_p(G)$ such that $\rho_S'^p(\alpha) = \beta'$
  we have that $\rho_S^p(\alpha) = \beta$, so $(\alpha,p)\in\cev{\beta}_S$ implies $p\in S$, so
  \[ \rho_{S}^p(\alpha) = G\b{S'\subseteq S}\circ \rho_{S'}^p(\alpha)\implies G\b{S'\subseteq S}(\beta') = \beta.\]
\end{proof}

\subsection{The Spatial Inverse Image Functor}\label{sec:inverse}

Let $h : X\to N$ be a continuous function from $X$ to a locally connected topological space.
Although the left Kan extension of a precosheaf along $h^{-1}$ is left adjoint to the direct image,
the image of a cosheaf on $X$ under this right adjoint is not a cosheaf in general.
If $h$ is \emph{essential}, precomposition with the left adjoint $h_!\dashv h^{-1}$ is a right adjoint $h^\dagger$ to $h_*$.
The goal of this section is to show that $h^*$ is equivalent to $h^\dagger$ when $h$ is a surjective essential map from a locally connected space to a poset endowed with the \emph{co}specialization topology.

The following result by Funk~\cite{funk95display} addresses the pullback stability of the display space functor.

\begin{theorem}[Funk~\cite{funk95display} Theorem 1.4]\label{thm:funk_pullback}
  Let $g : X\twoheadrightarrow Z$ be a surjective continuous function and let $G$ be a precosheaf on $Z$.
  If $g$ is essential then there is a canonical map $\pi_G^g : \dis(g^\dagger(G))\to \dis(G)$ which makes the following diagram a pullback in \Top:
  \begin{equation}
    \begin{tikzcd}[sep=huge]
      \dis(g^\dagger(G))\arrow[r, dotted,"\pi_G^g"]\arrow[d, dotted, "\gamma_{g^\dagger(G)}"]
        &\dis(G) \arrow[d, "\gamma_G"]\\
      X\arrow[r,"g"]
      & Z.
    \end{tikzcd}
  \end{equation}
\end{theorem}

Corollary~\ref{cor:adjoint_equiv} follows directly from the fact that $g^\star\gamma(G)$ is a pullback (Diagram~\ref{dgm:adjoint_equiv}).

\begin{equation}\label{dgm:adjoint_equiv}
  \begin{tikzcd}[sep=huge]
    X\times_Z \dis(G)
      \arrow[r, dotted]
      \arrow[d, dotted,"\gamma^g_{G}"]
      &\dis(G) \arrow[d,"\gamma_G"]\\
    X\arrow[r,"g"]
    & Z.
  \end{tikzcd}
\end{equation}

\begin{corollary}\label{cor:adjoint_equiv}
  If $g : X\twoheadrightarrow Z$ is surjective and essential then $g^\star\gamma\simeq \gamma g^\dagger$.
\end{corollary}

Note that the following result is the first time we have required that the space $X$ be locally connected.

\begin{theorem}\label{thm:poset_pullback}
  Let $h : X\twoheadrightarrow \P$ be a surjective essential function from a \emph{locally connected}
  space $X$ to a poset $\P$ endowed with the specialization topology.
  Then $h^\dagger(G)$ is a spatial cosheaf on $X$ for any precosheaf $G$ on $\P$.
\end{theorem}
\begin{proof}
  Because $h$ is surjective and essential Corollary~\ref{cor:adjoint_equiv} implies $\gamma h^\dagger \simeq h^\star\gamma = \big(X\times_\P\dis(G),\gamma_G^h\big)$
  (Diagram~\eqref{dgm:poset_pullback}), so it suffices to show that the pullback $X\times_\P\dis(G)$ of $h$ along $\gamma_G$ is locally connected.
  \begin{equation}\label{dgm:poset_pullback}
    \begin{tikzcd}[sep=huge]
      X\times_\P\dis(G)
        \arrow[r,dotted, "h^{\gamma_G}"]
        \arrow[d,dotted, "\gamma_G^h"]
      & \dis(G)
        \arrow[d,"\gamma_G"]
      \\
      X
        \arrow[r,"h"]
      & \P.
    \end{tikzcd}
  \end{equation}

  The pullback $X\times_\P \dis(G)$ has a basis of open sets $\beta_{h_!(U)}^U\defined U\times_\P \beta_{\eta_!(U)}$ for $U\in\Op_X$ $p\in\P$ and $\beta\in G(\up{p})$.
  By Proposition~\ref{prop:rho_bi}, $\beta_{h_!(U)}^U\neq \emptyset$ for all nonempty $U\in\Op_X$.
  Moreover, because $X$ is locally connected, it has a basis of connected open sets, so it suffices to show that $\beta_{h_!(U)}^U$ is connected for all connected $U\in\Op_X$.

  Let $U\in \Op_X$ be a connected basic open set of $X$ and suppose $\beta_{h_!(U)}^U = \delta_{h_!(V)}^V\cup \zeta_{h_!(W)}^W$
  for some $V,W\in\Op_X$, $\beta\in Gh_!(U)$, $\delta\in Gh_!(V)$, and $\zeta\in Gh_!(W)$.
  Then $V = V\cup W$ so, because $U$ is connected, $V\cap W\neq \emptyset$, so there exists a point $x\in V\cap W$.
  Moreover, by Lemma~\ref{lem:poset_pullback},
  $\beta_{h_!(U)}^U = \delta_{h_!(V)}^V\cup \zeta_{h_!(W)}^W$ implies that
  $G\b{h_!(V)\subseteq h_!(U)}(\delta) = \beta$ and
  $G\b{h_!(W)\subseteq h_!(U)}(\zeta) = \beta$.

  By Lemma~\ref{lem:special_spatial}, $G$ is a spatial cosheaf, so $h_!(U) = h_!(V)\cup h_!(W)$ implies $Gh_!(U)$ is the pushout of
  $G\b{h_!(V)\cap h_!(W)\subseteq h_!(V)}$ and $G\b{h_!(V)\cap h_!(W)\subseteq h_!(W)}$.
  It follows that there exists some $\alpha\in G\big(h_!(V)\cap h_!(W)\big)$ such that
  \begin{align*}
    G\b{h_!(V)\cap h_!(W)\subseteq h_!(V)}(\alpha) &= \delta\\
    \text{and } G\b{h_!(V)\cap h_!(W)\subseteq h_!(W)}(\alpha) &= \zeta.
  \end{align*}
  That is, there exists a point $\big(x,\big(\alpha,h(x)\big)\big)\in \delta_{h_!(V)}^V\cap \zeta_{h_!(W)}^W$,
  so we may conclude that $\beta_{h_!(U)}^U$ is connected as desired.
\end{proof}

\begin{corollary}\label{cor:poset_pullback}
  If $X$ is locally connected and $h : X\twoheadrightarrow \P$ is surjective and essential then $h^*\simeq h^\dagger$.
\end{corollary}
\begin{proof}
  Because $h$ is surjective and essential, $h^\star \gamma \simeq \gamma h^\dagger$ by Corollary~\ref{cor:adjoint_equiv},
  so $h^* = \lambda h^\star\gamma \simeq \lambda\gamma h^\dagger$.
  Because $\P$ is a poset endowed with the specialization topology,
  Theorem~\ref{thm:poset_pullback} implies that $h^\dagger(G)$ is spatial for any precosheaf $G$ on $\P$,
  so the counit yields a natural isomorphism
  \[ h^*(G) \simeq\lambda\gamma h^\dagger(G)\Longrightarrow h^\dagger(G).\]
\end{proof}

\subsection{Mapper Reeb Graphs}
Let $\U : I\to \Op_X$ be a locally finite good open cover of a topological space $X$.

\begin{definition}[Mapper Graph]
  For any precosheaf $F$ on $X$ let
  \[ \dis^\U(F)\defined X\times_{\N_\U} \dis(\eta_*(F)) = \big\{\big(x,(\alpha, \sigma)\big)\in X\times \dis(\eta_*(F))\mid \eta(x) = \sigma\big\}\]
  denote the pullback of $\gamma_{\eta_*(F)}$ along $\eta$ and let $\gamma_F^\U : \dis^\U(F) \to X$ denote
  the projection $(\alpha,x)\mapsto x$ of $\dis^\U(F)$ onto $X$:
  \begin{equation}
    \begin{tikzcd}[sep=huge]
      \dis^\U(F)
        \arrow[r, dotted]
        \arrow[d, dotted, "\gamma_F^\U"]
      & \dis(\eta_*(F))\arrow[d, "\gamma_{\eta_*(F)}"]
      \\
      X\arrow[r,"\eta"]
      & \N_\U
    \end{tikzcd}
  \end{equation}
  The \defn{$\U$-mapper graph functor} takes cosheaves on $X$ to the pullback of $\gamma\eta_*$ along $\eta$
  \begin{align*}
    \gamma^\U\defined \eta^\star\gamma\eta_* : \PCSh(X) &\longrightarrow \Toplc/X\\
      F &\longmapsto \big( \dis^\U(F),\gamma_F^\U)
  \end{align*}
\end{definition}

\begin{remark}
  The pixelization $F^\U$ of a precosheaf is the Reeb cosheaf of the mapper graph
  \[ R_{\gamma^\U(F)} = \lambda\gamma^\U(F) = \eta^*\eta_*(F).\]
\end{remark}

Let $\sU : \op{\P} \to \pOp_X$ be a strict refinement of locally finite good open covers of $X$.
For any precosheaf $F$ on $X$ let
\[ \dis^\sU(F)
  \defined X\times_{\lim\N_\sU} \dis(\hat{\eta}_*(F))
  = \big\{\big(x,(\alpha, \sigma)\big)\in X\times \dis(\hat{\eta}_*(F))\mid \hat{\eta}(x) = \sigma\big\}\]
denote the pullback of $\gamma_{\hat{\eta}_*(F)}$ along $\hat{\eta}$ and let $\gamma_F^\sU\defined \gamma_{\hat{\eta}_*(F)}^\eta$
denote the continuous projection $(\alpha,x)\mapsto x$ of $\dis^\sU(F)$ onto $X$
\begin{equation}\label{dgm:pullback}
  \begin{tikzcd}[sep=huge]
    \dis^\sU(F)
      \arrow[r, dotted]
      \arrow[d, dotted, "\gamma_F^\sU"]
    & \dis(\hat{\eta}_*(F))\arrow[d, "\gamma_{\hat{\eta}_*(F)}"]
    \\
    X\arrow[r,"\hat{\eta}"]
    & \lim\N_\sU
  \end{tikzcd}
\end{equation}
The \defn{$\sU$-multiscale mapper graph functor} takes cosheaves on $X$ to the pullback of $\gamma\hat{\eta}_*$ along $\hat{\eta}$.
\begin{align*}
  \hat{\gamma}^\sU\defined \hat{\eta}^\star\gamma\hat{\eta}_* : \CSh(X) &\longrightarrow \Toplc/X\\
    F &\longmapsto \big( \dis^\sU(F),\gamma_F^\sU)
\end{align*}
Because $\hat{\gamma}^\sU$ is a composition of right adjoints,
Theorem~\ref{thm:multi_graph} follows from the fact that right adjoints preserve limits.

\begin{theorem}\label{thm:multi_graph}
  Let $X$ be a locally connected topological space.
  Then for any strict refinement $\sU : \op{\P}\to\pOp_X$ of locally finite good open covers,
  \[\hat{\gamma}^\sU \simeq \lim_{p\in \P}\gamma^{\sU(p)}.\]
\end{theorem}

\section{Reeb Filtrations and Future Work}\label{sec:reeb_future}

Let $(L,f)$ be a locally connected space over $X$ and let $h : X\to N$ be a continuous function to a locally connected space $N$:
\begin{equation}
  \begin{tikzcd}
    L \arrow[rr,dotted]
      \arrow[dr,"f"']
    && N.\\
    & X\arrow[ur,"h"'] &
  \end{tikzcd}
\end{equation}
The map $(L,h)\mapsto (L,h\circ f)$ defines a left adjoint $\Sigma_h : \Toplc/X\to \Toplc/N$ to the change of base functor (see Appendix~\ref{sec:cats2}):
\[\adjoint{\Toplc/X}{\Toplc/N.}["\Sigma_h"]["h^\star"]\]
The composition of $\Sigma_h$ with the Reeb functor is equal to the direct image of the Reeb functor along $h$ Reeb cosheaf of the composition $h\circ f$:
\[ \lambda\Sigma_h(L,f) = \lambda(L,h\circ f) = R_{h\circ f} = \pi_0(h\circ f)^{-1} = (\pi_0 f^{-1})h^{-1} = h_*\lambda(L,f).\]

\begin{notation}
  let $h : X\to N$ be a continuous function from $X$ to a locally connected space $N$.
  The direct image of the Reeb functor along $h$ will be denoted
  \[\lambda_h\defined h_*\lambda = \lambda\Sigma_h : \Toplc/X\longrightarrow \CShsp(N)\]
  and the pullback of $\gamma$ along $h$ will be denoted
  \[ \gamma^h\defined h^\star\gamma : \CShsp(N)\longrightarrow \Toplc/X\]

  Composition of $\lambda\dashv\gamma$ with $\Sigma_h\dashv h^\star$ yields an adjunction $\lambda_h\dashv \gamma^h$:
  \begin{equation}
    \begin{tikzcd}
        \Toplc/X
                \arrow[r, shift left=0.62ex, phantom, ""{name=x,above}]
                \arrow[r, shift left=0.85ex, "\Sigma_h"]
      & \Toplc/Z\arrow[l, shift left=0.42ex, phantom, ""{name=y,below}]
                \arrow[r, shift left=0.62ex, phantom, ""{name=xx,above}]
                \arrow[r, shift left=0.85ex, "\lambda"]
                \arrow[l, shift left=0.85ex, "h^\star"]
      & \CSh(N)
                \arrow[l, shift left=0.42ex, phantom, ""{name=yy,below}]
                \arrow[l, shift left=0.85ex,"\gamma"]
        \arrow[from=x, to=y, symbol={\scriptstyle\dashv}]
        \arrow[from=xx, to=yy, symbol={\scriptstyle\dashv}]
      \end{tikzcd}
      \yields
      \adjoint{\Toplc/X}{\CSh(N).}["\lambda_h"]["\gamma^h"]
  \end{equation}
  Recalling the unit $\upsilon^h$ and counit $\chi^h$ of $\Sigma_h\dashv h^\star$, the unit and counit of $\lambda_h\dashv\gamma^h$ are defined
  \begin{align*}
    \mathbf{(Unit)}&\hspace{5ex} \Gamma^h\defined h^\star\Gamma\Sigma_h\bdot \upsilon^h : \id_{\Toplc/X}\Longrightarrow \gamma^h\lambda_h\hspace{15ex}&\\
    \mathbf{(Counit)}&\hspace{5ex} \Lambda^h\defined \Lambda\bdot \lambda\chi^h\gamma : \lambda_h\gamma^h\Longrightarrow\id_{\CSh(N)}.  \hspace{15ex}&
  \end{align*}
  Moreover, if $X$ is locally connected and $h$ is essential then Corollary~\ref{cor:adjoint_equiv} implies that $\Gamma^h$ and $\Lambda^h$ can be equivalently expressed in terms of the unit $\theta^h$ and counit $\epsilon^h$ of $h_*\dashv h^\dagger$:
  \begin{align*}
    \textbf{(Unit)}\hspace{5ex} \Gamma^g &\simeq  \gamma\theta^g\lambda\bdot \gamma \hspace{15ex}\\
    \textbf{(Counit)}\hspace{5ex} \Lambda^g &\simeq \e^g\bdot g_*\Lambda g^\dagger.\hspace{15ex}
  \end{align*}
\end{notation}

\begin{example}[Spatial Pixelization]
  The \emph{(spatial) pixelization}~\cite{botnan20relative} of the Reeb cosheaf $R_f$ by $h$ is given by the spatial inverse image of the direct image along $h$:
  \[ h^*h_*(R_f) = \lambda\gamma^h\lambda_h(L, f).\]
  The unit of $\lambda_h\dashv\gamma^h$ at $R_f$ yields a natural transformation from the Reeb cosheaf of $(L,f)$ to its pixelization by $h$:
  \[ \lambda\Gamma^h_f : R_f\Rightarrow h^*h_*(R_f).\]
\end{example}

\begin{remark}
  The unit of $\lambda_\eta\dashv\gamma^\eta$ yields a canonical map from $(L,f)$ to its $\U$-mapper Reeb graph:
  \[ \Gamma^\U_f :  (L,f)\longuto\gamma^\U(R_f).\]
  Composition with the Reeb functor yields a canonical natural transformation that takes the Reeb cosheaf to its pixelization by $\U$:
  \[ \lambda\Gamma^\U_f : R_f\Longrightarrow R_f^\U.\]
  Theorem~\ref{thm:mapper} therefore implies that composition with the counit $\Lambda_f^\U : R_f^\U\Rightarrow \M_\U$
  yields a canonical natural transformation that takes the Reeb cosheaf of a locally connected space $(L,f)$ over $X$ to its $\U$-mapper Reeb cosheaf:
  \[ \Lambda_f^\U\bdot \lambda\Gamma^\U_f : R_f\Longrightarrow \M_\U(R_f).\]
\end{remark}

\subsection{Reeb Filtrations}\label{sec:reeb_filt}

Let $X$ be a topological space.
The Reeb cosheaf of a locally connected space $(L,f)$ over $X$ associates each open set of $X$ with the connected components of its pre-image.
The display space functor integrates this information as the Reeb graph $\gamma(R_f) = \big(\dis(R_f),\gamma_f\big)$.
The goal of this section is to extend this construction to filtrations of spaces over $X$.

Let $\varphi : \P\to \Toplc/X$ be a filtration of locally connected topological spaces over $X$
and suppose $\L : \P\to\Toplc$ is a filtration of locally conencted spaces such that $\colim \varphi = \big(\colim\L,\check{\varphi}\big)$
where $\check{\varphi} : \colim\L\to X$ is the universal cocone for which Diagram~\eqref{dgm:colimit} commutes for all $p\leq q$.
\begin{equation}\label{dgm:colimit}
  \begin{tikzcd}[column sep=huge]
    \L(p)
      \arrow[dd,"\L\b{p\leq q}"description]
      \arrow[drr,bend left=10,"\varphi(p)"description]
      \arrow[dr, dotted, "\iota^p"description]
    &&\\
    & \colim\L
      \arrow[r,dotted,"\exists !\check{\varphi}"description]
    & X\\
    \L(q)
      \arrow[urr,bend right=10,"\varphi(q)"description]
      \arrow[ur, dotted, "\iota^q"description]
    &&
  \end{tikzcd}
\end{equation}

\begin{notation}
  Let $\ind{R_\varphi} \defined R_{\hat{\varphi}}$ denote the Reeb cosheaf of the universal arrow $\check{\varphi} : \colim \L\to X$.
\end{notation}

The (horizontal) composition of $\varphi$ with the Reeb functor yields a filtration of Reeb cosheaves
\begin{align*}
  R_\varphi\defined \lambda \varphi : \P&\longrightarrow \CShsp(X)\\
    p &\longmapsto \pi_0\varphi(p)^{-1}.
\end{align*}
Because left adjoints preserve colimits the colimit of $R_\varphi = \lambda\varphi$ is equivalent to the Reeb cosheaf of $\colim \varphi = (\colim \L,\check{\varphi})$:
\[ \colim R_\varphi = \colim\lambda\varphi = \lambda\big(\colim\varphi\big) = \ind{R_{\varphi}}.\]

\begin{proposition}\label{prop:reeb_commute}
  Let $h : X\to N$ be a continuous function from $X$ to a locally connected space $N$.
  If $\P$ is directed then
  \[ \gamma^h\lambda_h\big(\colim\varphi\big)\simeq \colim \gamma^h\lambda_h\varphi.\]
\end{proposition}
\begin{proof}
  Because left adjoints preserve colimits and $\lambda_h\dashv\gamma^h$ we have
  \[ \gamma^h\lambda_h\big(\colim\varphi\big) \simeq \gamma^h\big(\colim\lambda_h\varphi\big).\]
  Because $N$ is locally finite and $\P$ is directed Proposition~\ref{prop:display_commute} implies
  \[ \gamma^h\big(\colim\lambda_h\varphi\big) = h^\star \gamma\big(\colim\lambda_h\varphi\big) \simeq h^\star\big(\colim\gamma\lambda_h\varphi\big).\]
  Similarly, because $h^\star$ is defined as the pullback of $h$ along $\gamma$, and because colimits are stable by base change in \Set (see Definition~\ref{def:stable_base}),
  we obtain the desired result from the following interchange of limits:
  \[ \gamma^h\lambda_h\big(\colim\varphi\big)\simeq h^\star\big(\colim\gamma\lambda_h\varphi\big)\simeq \colim h^\star\gamma\lambda_h\varphi = \colim \gamma^h\lambda_h\varphi.\]
\end{proof}

Let $\U : I\to\Op_X$ be a locally finite good open cover of $X$.
Corollary~\ref{cor:reeb_commute} follows directly from Proposition~\ref{prop:reeb_commute}
and the observation that $\gamma^\U R_\varphi = \eta^\star\gamma \eta_*\lambda \varphi = \gamma^\eta\lambda_\eta\varphi$.

\begin{corollary}\label{cor:reeb_commute}
  If $\P$ is directed then $\gamma^\U\big(\ind{R_\varphi}\big)\simeq \colim \gamma^\U R_\varphi$.
\end{corollary}

\subsection{Future Work}\label{sec:future}

Let $\P$ be a finite poset and let $\Q$ be a directed poset, and let $X$ be a locally connected topological space.
Let $\sU : \op{\P}\to\pOp_X$ be a refinement of locally finite good open covers of $X$
and let $\varphi : \Q\to \Toplc/X$ be a filtration of locally connected topological spaces over $X$.
The composition of $R_\varphi = \lambda\varphi$ with the $\sU$-mapper Reeb functor yields a functor
\[ \gamma^\sU R_\varphi : \Q\longrightarrow \Fun{\op{\P},\Toplc/X}.\]
Let $p\in \P$.
Because $\Q$ is directed Corollary~\ref{cor:reeb_commute} implies that
\[ \colim_{q\in Q} \gamma^{\sU(p)} R_{\varphi(q)} \simeq \gamma^{\sU(p)} R_{\check{\varphi}}.\]
Because $\lambda$ is a left adjoint it commutes with colimits so
\[ \colim_{q\in Q} R_{\varphi(q)}^{\sU(p)} \simeq \gamma^{\sU(p)} \ind{R_\varphi}.\]

Because $\P$ is finite and $X$ is locally connected Theorem~\ref{thm:multi_graph} implies
\[ \lim_{p\in \P}\colim_{q\in \Q} \gamma^{\sU(p)} R_{\varphi(q)} \simeq \lim_{p\in \P} \gamma^{\sU(p)} R_{\check{\varphi}}\simeq \hat{\gamma}^\sU \ind{R_\varphi}.\]
Theorem~\ref{thm:mapper} therefore implies that
\[ \lim_{p\in \P}\colim_{q\in \Q} R_{\varphi(q)}^{\sU(q)}
  \simeq \lim_{p\in \P} \big(\ind{R_{\varphi}}\big)^{\sU(q)}
  \simeq \pro{\sM_{\sU}}(\ind{R_\varphi}).\]
That is, the cosections of $\lim\colim R_\varphi^\sU$ can be computed for $U\in\Op_X$ as
\[ \pro{\sM_{\sU}}\big(\ind{R_\varphi}\big)(U) = \bigcup_{q\in Q} \pi_0 \varphi_q^{-1}\Big(\Meet \sU_U\Big)\]

\begin{remark}
  Suppose $\P = \Q$ is a finite directed poset.
  By exponential adjunction, the functor $\gamma^\sU R_\varphi$ is equivalent to a
  $(\op{\P}\times\P)$-indexed filtration of locally connected spaces over $X$
  \[ \gamma^\sU_\varphi : \op{\P}\times\P\longrightarrow \Toplc/X.\]
  Composition with the Reeb functor yields a filtration of spatial cosheaves
  \[  R^\sU_\varphi \defined \lambda\gamma^\sU_\varphi: \op{\P}\times\P\longrightarrow \CShsp(X).\]
  Exploration of the (co)end $\int R^\sU_\varphi$ as a generalized category of elements
  (also known as the Grothendieck construction) is the subject of future work.
  We also invite the reader to consider the following diagram:
  \begin{equation}
    \begin{tikzcd}[column sep=huge]
      \L(p)
        \arrow[dd,"\L\b{p\leq q}"description]
        \arrow[drr,bend left=10,"\varphi_p"description]
        \arrow[dr, dotted, "\iota^p"description]
        \arrow[rrrr, bend left=10, dotted]
      &&
      & & \N_\sU(p)\\
      & \colim\L
        \arrow[r,dotted,"\exists !\check{\varphi}"description]
      & X
        \arrow[r, dotted, "\exists !\hat{\eta}"description]
        \arrow[urr, bend left=15, "\eta_p"description]
        \arrow[drr, bend right=15, "\eta_q"description]
      & \lim \N_\sU
        \arrow[ur,dotted,"\rho_p"description]
        \arrow[dr,dotted,"\rho_q"description]\\
      \L(q)
        \arrow[urr,bend right=10,"\varphi_q"description]
        \arrow[ur,dotted,"\iota^q"description]
        \arrow[rrrr, bend right=10, dotted]
      &&
      & & \N_\sU(q)\arrow[uu, "\N_\sU\b{q\geq p}"description]
    \end{tikzcd}
  \end{equation}
\end{remark}

\section{Omitted Proofs}\label{sec:proofs}

\begin{proof}[Proof of Proposition~\ref{prop:eta_open}]
  Let $\sigma\in\N_\U$.
  We will begin by showing that $\eta^{-1}(\up{\sigma}) = \U_\sigma$.

  If $x\in\U_\sigma$ then $x\in \U(i)$ for all $i\in\sigma$,
  which implies that $\eta(x) \supseteq \sigma$ so $x\in \eta^{-1}(\up{\sigma})$.
  Conversely, if $x\in \eta^{-1}(\up{\sigma})$ then $\eta(x)\in \up{\sigma}$ implies $x\in \U_x\subseteq \U_\sigma$,
  so $\eta^{-1}(\up{\sigma}) = \U_\sigma$ for all $\sigma\in \N_\U$.
  Because $\U$ is a locally finite good open cover, $\eta^{-1}(\up{\sigma}) = \U_\sigma = \bigcap_{i\in\sigma} \U(i)$
  is an intersection of \emph{finitely} many (basic) open sets, so $\eta : X\to \N_\U$ is continuous with
  \begin{align*}
    \eta^{-1} : \Op_{\N_\U} &\longrightarrow \Op_X\\
      S&\longmapsto \U_S = \bigcup_{\sigma\in S} \U_\sigma.
  \end{align*}
  It remains to show that $\eta$ is essential.

  Let $\eta_! : \Op_X\to \Op_{\N_\U}$ be the monotone function defined for $U\in\Op_X$ as
  \[ \eta_!(U) \defined \bigcup_{x\in U}\up{\eta(x)}.\]
  In order to show that $\eta$ is essential it suffices to show that, for all $S\in\Op_{\N_\U}$ and $U\in \Op_X$
  $\eta_!(U)\subseteq S$ if and only if $U\subseteq \eta^{-1}(S)$.

  Suppose $\eta_!(U)\subseteq S$.
  Because $\eta^{-1}$ is monotone, $\eta^{-1}\eta_!(U)\subseteq \eta^{-1}(S)$.
  Therefore, because $x\in \U_x$ for all $x\in U$, it follows that
  \[ U\subseteq \bigcup_{x\in U}\U_x = \U_U = \eta^{-1}\eta_!(U) \subseteq \eta^{-1}((S).\]
  Conversely, suppose $U\subseteq \eta^{-1}(S)$.
  Once again, because $\eta_!$ is monotone,
  \[ \eta_!(U)\subseteq \eta_!\eta^{-1}(S) = \eta_!(\U_S) = \bigcup_{x\in \U_S} \up{\eta(x)}.\]
  If $\eta(x)\in \eta_!(\U_S)$ then there exists some $\sigma\in S$ such that $\eta(x)\in\U_\sigma$,
  which implies $\eta(x)\supseteq \sigma$.
  Because $S$ is up-closed, it follows that $\up{\eta(x)}\subseteq S$, thus
  \[ \eta_!(U) \subseteq \bigcup_{x\in\U_S} \up{\eta(x)} = \bigcup_{\sigma\in S}\up{\sigma} \subseteq S.\]
\end{proof}

\begin{proof}[Proof of Proposition~\ref{prop:display_nerve_iso}]
  By Proposition~\ref{prop:rho_bi},
  \[ \dis(R_{\eta\circ f}) = \bigsqcup_{\sigma\in \N_\U}\lim_{S\ni \sigma} R_f(\U_S) \simeq \bigsqcup_{\sigma\in \N_\U} R_f(\U_\sigma),\]
  so it suffices to show that $\N_\U^f$ is isomorphic to the coproduct.

  For any $B\subseteq \bigsqcup R_f\U$, let $B_0 \defined \big\{\beta\in R_f\big(\U(i)\big)\mid (\beta,i)\in B\big\}$ and $B_1 \defined  \big\{ i \in I\mid (\beta,i)\in B \big\}$.
  Let $\phi : \N_\U^f\to \bigsqcup_{\sigma\in \N_\U} R_f(\U_\sigma)$ be defined $\phi(B) = \Big(\Meet B_0, B_1\Big)$
  and let $\psi : \bigsqcup_{\sigma\in \N_\U} R_f(\U_\sigma)\to \N_\U^f$ be defined
  \[ \psi(\alpha,\sigma) = \bigsqcup_{i\in \sigma}R_f\b{\U_\sigma\subseteq \U(i)}(\alpha).\]
  Then
  \[ \phi\circ\psi(\alpha,\sigma) = \phi\big(\bigsqcup_{i\in \sigma}R_f\b{\U_\sigma\subseteq \U(i)}(\alpha)\big) = \big(\bigcap_{i\in \sigma}R_f\b{\U_\sigma\subseteq \U(i)}(\alpha),\sigma \big) = (\alpha,\sigma).\]
  It remains to show that
  \[ B = \bigsqcup_{i\in B_1} R_f\b{ \U_{B_1}\subseteq \U(i)}\big( \Meet B_0\big).\]
  Because $B\in\N_\U^f$, $\Meet B_0\neq\emptyset$ implies $\Join B_0$ is connected and,
  for all $(\beta_j,j),(\beta_k,k)\in B$,
  \[ R_f\big[\U(j)\subseteq \bigcup_{i\in B_1}\U(i)\big](\beta_j) = R_f\big[\U(k)\subseteq\bigcup_{i\in B_1}\U(i)\big](\beta_k).\]
  Because $R_f$ is a cosheaf, it follows that there exists some $\alpha\in R_f(\U_{B_1})$ such that
  \begin{align*}
    &R_f\b{\U_{B_1}\subseteq \U(j)}(\alpha) = \beta_j\\
    \text{and } &R_f\b{\U_{B_2}\subseteq \U(j)}(\alpha) = \beta_k.
  \end{align*}
  Because $\alpha \supseteq \Meet B_0$ we have that $\alpha = \Meet B_0$, so $\bigsqcup_{i\in B_1}R_f\b{\U_{B_1}\subseteq \U(i)}(\Meet B_0) = B$ as desired.
\end{proof}

\begin{proposition}\label{prop:display_commute}
  Let $N$ be a locally finite topological space and let $\F : \P\to \PCSh(N)$ be a filtration of cosheaves on $N$.
  If $\P$ is directed then
  \[ \gamma\big(\colim \F\big)\simeq \colim\gamma \F.\]
\end{proposition}
\begin{proof}
  Because $\gamma$ is defined as the projection of the display space onto $X$
  it suffices to show that the display space of $\colim \F$ is homeomorphic to the colimit $\colim \dis \F$.

  Recall that the display space of $\colim \F$ is defined
  \[ \dis\big(\colim \F\big) = \coprod_{n\in N} \lim_{S\ni n} \colim_{p\in\P} \F(p).\]
  Because colimits commute with coproducts it suffices to show that $\dis_n\big(\colim \F\big)\cong\colim \dis_n \F$.

  Let $n\in N$.
  Because $N$ is locally finite, $\Op_N(n)$ is finite.
  Therefore, because $\P$ is directed, we obtain the desired result from
  the fact that directed colimits commute with finite limits
  (see Kashiwara and Schapira~\cite{kashiwara05categories} Chapter 3):
  \[ \dis_n\big(\colim \F\big)
    = \lim_{S\ni n}\Big(\colim_{p\in\P} \F(p)\Big)
    \cong \colim_{p\in\P} \Big(\lim_{S\ni n} \F(p)\Big)
    = \colim \dis_n \F.\]
\end{proof}

Let $\sU : \op{\P}\to\mathbf{pOp}_X$ be a refinement of locally finite good open covers of a topological space $X$.
Let $\lim \N_\sU$ be the topological space endowed with the initial topology associated with the system of canonical projections
\[ \big\{ \rho_t : \lim\N_\sU\to \N_\sU(t)\big\}_{t\in\P}.\]
There is a natural partial order on $\lim\N_\sU$ in which $\sigma\subseteq\tau$ if $\sigma^t\subseteq \tau^t$ for all $t\in\P$.
Lemma~\ref{lem:limit_basis} implies that the initial topology associated with the system of canonical projections
is equivalent to the specialization topology on $\lim\N_\sU$ regarded as a poset.

\begin{lemma}\label{lem:limit_basis}
  The collection of principal up sets $\{\up{\sigma}\mid \sigma\in\lim\N_\sU\}$ forms a basis for $\Op_{\lim\N_\sU}$.
\end{lemma}
\begin{proof}
  Because $\Op_{\N_\sU(t)}$ is generated by the collection $\{\up{\sigma^t}\in\Op_{\N_\sU(t)}\mid \sigma^t\in\N_\sU(t)\}$ for all $t\in\P$,
  the initial topology on $\lim\N_\sU$ is generated by open sets of the form $\rho_t^{-1}(\up{\sigma^t}) = \{\tau\in\lim\N_\sU\mid \sigma^t\subseteq \tau^t\}$.
  We will show that, for all $t\in\P$ and $\sigma^t\in \N_\sU(t)$, there exists a set $S\subseteq \lim\N_\sU$ such that
  \[ \rho_t^{-1}(\up{\sigma^t}) = \bigcup_{\sigma'\in S} \up{\sigma'}.\]

  Let $t\in\P$ and $\sigma^t\in\N_\sU(t)$, and let $S = \{\sigma'\in\lim\N_\sU\mid{\sigma'}^t = \sigma^t\}$.
  If $\tau\in\bigcup_{\sigma'\in S}\up{\sigma'}$ then ${\sigma'}^p\subseteq \tau^p$ for all $p\in\P$.
  In particular, $\sigma^t = {\sigma'}^t \subseteq \tau^t$, so $\tau\in\rho_t^{-1}(\up{\sigma^t})$.
  Conversely, if $\tau\in\rho_t^{-1}(\up{\sigma^t})$ then let $\sigma'$ be defined for $p\in\P$ as
  \[ {\sigma'}^p = \begin{cases} \N_\sU\b{p\leq t}(\sigma^t)&\text{ if } p\leq t\\ \{i\in\tau^p\mid \supp{\sU}\b{t\leq p}(i)\in\sigma^t\} &\text{ otherwise.}\end{cases}\]
  Because $\N_\sU\b{p\leq q}({\sigma'}^q) = {\sigma'}^p$ for all $p\leq q$ in \P we have that $\sigma'\in\lim\N_\sU$.
  For all $p\leq t$, $\tau\in\rho_t^{-1}(\up{\sigma^t})$ implies ${\sigma'}^t = \sigma^t\subseteq \tau^t$, so
  \[ {\sigma'}^p = \N_\sU\b{p\leq t}(\sigma^t)\subseteq \N_\sU\b{p\leq t}(\tau^t) = \tau^p\]
  by the functoriality of $\N_\sU$.
  Otherwise, if $p\not\leq t$ then ${\sigma'}^p\subseteq \tau^p$ by definition, so we may therefore conclude that $\sigma'\in S$,
  thus $\tau\in \bigcup_{\sigma'\in S}\up{\sigma'}$ as desired.
\end{proof}

\end{document}